\documentclass[11pt]{article}
\usepackage[margin=1in]{geometry}
\usepackage{fullpage}
\usepackage{amsmath,amssymb,amsthm,mathtools,bbm}
\usepackage{natbib}
\usepackage{authblk}
\usepackage[colorlinks=true,linkcolor=blue,citecolor=blue,urlcolor=blue]{hyperref}
\usepackage[capitalise,noabbrev]{cleveref}
\crefname{assumption}{Assumption}{Assumptions}

\makeatletter
\let\oldappendix\appendix
\renewcommand{\appendix}{%
  \oldappendix
  \crefalias{section}{appendix}%
  \crefalias{subsection}{appendix}%
  \crefalias{subsubsection}{appendix}%
}
\makeatother

\crefname{appendix}{Appendix}{Appendices}
\Crefname{appendix}{Appendix}{Appendices}

\usepackage{booktabs}
\usepackage{enumitem}
\usepackage{xcolor}
\usepackage{tikz}
\usetikzlibrary{arrows.meta, positioning}

\usepackage{aliascnt}
\newtheorem{theorem}{Theorem}[section]
\newaliascnt{proposition}{theorem}
\newtheorem{proposition}[proposition]{Proposition}
\aliascntresetthe{proposition}
\newaliascnt{lemma}{theorem}
\newtheorem{lemma}[lemma]{Lemma}
\aliascntresetthe{lemma}
\newaliascnt{corollary}{theorem}
\newtheorem{corollary}[corollary]{Corollary}
\aliascntresetthe{corollary}
\newaliascnt{assumption}{theorem}
\newtheorem{assumption}[assumption]{Assumption}
\aliascntresetthe{assumption}
\newaliascnt{definition}{theorem}
\newtheorem{definition}[definition]{Definition}
\aliascntresetthe{definition}
\newaliascnt{remark}{theorem}
\newtheorem{remark}[remark]{Remark}
\aliascntresetthe{remark}
\newaliascnt{example}{theorem}
\newtheorem{example}[example]{Example}
\aliascntresetthe{example}
\crefname{proposition}{Proposition}{Propositions}
\crefname{lemma}{Lemma}{Lemmas}
\crefname{corollary}{Corollary}{Corollaries}
\crefname{definition}{Definition}{Definitions}
\crefname{remark}{Remark}{Remarks}
\crefname{example}{Example}{Examples}

\newcommand{\indep}{\perp\!\!\!\!\!\!\perp}
\newcommand{\E}{\mathbb{E}}

\newcommand{\R}{\mathbb{R}}
\newcommand{\Var}{\mathrm{Var}}
\newcommand{\Cov}{\mathrm{Cov}}

\newcommand{\cB}{\mathcal{B}}

\newcommand{\cV}{\mathcal{V}}
\newcommand{\cT}{\mathcal{T}}
\newcommand{\cN}{\mathcal{N}}
\newcommand{\cH}{\mathcal{H}}
\newcommand{\cG}{\mathcal{G}}
\newcommand{\cF}{\mathcal{F}}
\newcommand{\cL}{\mathcal{L}}
\DeclareMathOperator*{\argmin}{arg\,min}

\title{Causal Robustness from Synthetic Environments}
\author[1]{Linying Yang}
\author[2]{Cyrill Scheidegger}
\author[3]{Xinwei Shen}
\author[2]{Peter Bühlmann}
\affil[1]{Department of Statistics, University of Oxford}
\affil[2]{Seminar for Statistics, ETH Zurich}
\affil[3]{Department of Statistics, University of Washington}
\date{\today}

\begin{document}
\maketitle

\begin{abstract}
We study robust prediction when the training data come from a single observational distribution. Causality-oriented methods for robust prediction usually use heterogeneity across observed environments both to learn an invariant prediction function and to determine the perturbations against which it is robust. We show that both can be recovered from one distribution under a structural additive-noise mechanism for the response. Under the structural model considered here, the causal prediction function has constant conditional mean squared error given its prediction. We generate synthetic environments by reweighting the covariate distribution along the prediction and use a Cramér--von Mises criterion to measure departures from this invariance. The resulting penalised risk admits a regularised worst-case representation, linking the criterion to robustness against covariate shifts induced by reweighting observations according to their predicted values. Under an additional orthogonality condition, we identify the causal prediction function as the unique risk minimiser subject to this invariance restriction. We establish consistency and an explicit rate, derive a goodness-of-fit test for the invariance restriction, and characterise the population target under misspecification. With multiple observed environments, the construction applies environment by environment and can be combined with Distributional Robustness via Invariant Gradients (DRIG) to control two complementary classes of shifts.
\end{abstract}

\section{Introduction}\label{sec:introduction}

Prediction functions chosen solely for their performance under the training distribution can lose accuracy after a distribution shift. Distributionally robust optimisation (DRO) addresses this problem by minimising the worst-case population risk over a class $\mathcal{P}$ of plausible test distributions \citep{esfahani2018data,meinshausen2018causality,duchi2021learning}. For a loss function $\ell$, its general form is
\begin{equation*}
\min_{g\in\cG}\ \sup_{P\in\mathcal{P}}\ \E_P\big[\ell(Y,g(X))\big].
\end{equation*}
Here $g$ ranges over a function class $\cG$, and $\mathcal{P}$ is the perturbation class against which the prediction function is protected. We use the squared loss $\ell(y,t)=(y-t)^2$. If $\mathcal P$ contains only the training distribution, replacing the population expectation by its empirical counterpart gives least-squares empirical risk minimisation (ERM).

A standard DRO construction takes $\mathcal{P}$ to be a ball $\{Q:D(Q,P)\leq\rho\}$ around the training distribution $P$, for a chosen distance or divergence $D$. Such a ball describes the size of a shift without modelling how it arises. Group DRO instead takes the worst case over predefined groups \citep{sagawa2020distributionally}, so the possible test distributions are built from these groups. A distance ball large enough to cover relevant structured shifts may contain perturbations that cannot arise under the data-generating mechanism. A smaller ball may omit relevant shifts.

Causal models can restrict the perturbation class by specifying which structural equations may change and which remain invariant. Existing methods learn such invariances from multiple observed environments or from so-called anchor variables, which are extensions of instrumental variables. Invariant causal prediction uses shifts across environments to find a subset $S$ for which $P(Y\,|\,X_S)$ is stable \citep{peters2016causal}. Anchor regression \citep{rothenhausler2021anchor} and Distributional Robustness via Invariant Gradients \citep[DRIG;][]{shen2026causality} use anchor variables or observed environments to learn directions of shift. Anchor regression penalises the residual component explained by the anchor variables, while DRIG penalises excess risk across environments. Their linear robustness results protect against shifts along the learned directions. Recovery of the causal prediction function requires additional identification conditions. Causal Dantzig \citep{rothenhausler2019causal} targets this function by requiring the expected product of each covariate and the prediction residual to be invariant across environments. Invariant risk minimisation seeks a representation for which the same predictor is risk-optimal across observed environments \citep{arjovsky2019invariant}, whereas risk extrapolation penalises variation of the risks across environments \citep{krueger2021out}. All these methods rely on heterogeneity supplied by environments or anchor variables.

We provide detailed reviews in \Cref{sec:related}.

\subsection{Contributions}
We ask what can be learned when no environmental labels or anchor variables are available. This requires an implication of the causal model in \Cref{ass:structural} that is visible in one observational distribution. We construct synthetic environments by reweighting this distribution and develop a distributionally robust method that enforces invariance across the resulting distributions. Under the conditions in \Cref{sec:ident}, the method recovers the causal prediction function. \Cref{rem:gaussian_linear} clarifies their scope under joint Gaussianity. Even when the causal prediction function is not identified, the criterion at finite $\lambda$ can still provide robustness against covariate reweightings based on the prediction. The following list gives a technical outline of our results.
\begin{itemize}
\item \emph{Synthetic environments and the CvM penalty.} For a candidate prediction function $g$ of $Y$, we reweight the covariates according to $g(X)$. Under these reweightings, the conditional MSE given the prediction is unchanged while the marginal distribution of the prediction varies. A Cram\'er--von Mises (CvM) penalty aggregates the resulting risk differences (\Cref{sec:reweight,sec:cvm}).

\item \emph{Distributional robustness.} We add the CvM penalty to prediction risk, with a regularisation parameter $\lambda$ that trades fit under the observational distribution against departure from invariance. For finite $\lambda>0$, the resulting objective has a regularised worst-case representation over perturbations that reweight observations according to their predictions. The perturbation class is therefore induced by the CvM construction rather than specified externally (\Cref{sec:dro}).

\item \emph{Identification and statistical guarantees.} As $\lambda\to\infty$, the penalty enforces constant conditional MSE, or equivalently invariant risk across the synthetic environments. Under an orthogonality condition and suitable regularity conditions, the population minimisers then approach the causal prediction function $g_0$. We prove consistency and a rate for the empirical estimator and derive a goodness-of-fit test for this invariance restriction (\Cref{sec:ident,sec:theory}).

\item \emph{Multiple observed environments.} When multiple environments are observed, we compute the CvM penalty separately within each environment and combine it with DRIG, thereby exploiting the observed heterogeneity. The two terms protect against different classes of shifts (\Cref{sec:multienv}).

\item \emph{Empirical performance.} In synthetic experiments satisfying \Cref{ass:structural}, the method recovers the causal prediction function $g_0$ and remains accurate under interventions on variables causally downstream of $Y$. On the RPE1 Perturb-seq data of \citet{replogle2022mapping}, following the experimental setup of \citet{shen2026causality}, the estimator combining DRIG with the environment-weighted CvM penalty attains the lowest worst-case error in most settings.
\end{itemize}

\subsection{\texorpdfstring{Benchmark methods and related work}{Benchmark methods and related work}}\label{sec:related}

We state the objectives of anchor regression and DRIG because both are experimental benchmarks, and the DRIG objective is combined with our penalty in \Cref{sec:multienv}. Their robustness classes also clarify how the methods differ from our reweighting construction, as illustrated in \Cref{fig:intervention_dags}. We use linear prediction functions $X^\top\beta$ because the cited worst-case guarantees for both methods are established under linear structural models. The DRIG training objective itself can be evaluated over a broader prediction class; this does not by itself extend its structural robustness theorem beyond the linear setting.

Anchor regression assumes a linear structural model $W = \mathrm B W + \mathrm{M} A + \varepsilon$, where $W$ collects the observed variables $(X,Y)$ together with any hidden confounders and the exogenous anchor variable $A$ acts through the term $\mathrm{M} A$. Here $\mathrm B$ and $\mathrm M$ are fixed coefficient matrices describing, respectively, the effects among the components of $W$ and the effects of $A$ on $W$. With $\Gamma_A$ the $L^2$ projection onto $\mathrm{span}(A)$, so that $\mathrm{Id} - \Gamma_A$ projects onto its orthogonal complement, anchor regression solves
\begin{equation}\label{eq:anchor_obj}
\hat\beta_\gamma
=
\arg\min_\beta
\E\big[\big((\mathrm{Id}-\Gamma_A)(Y-X^\top\beta)\big)^2\big]
+
\gamma\,\E\big[\big(\Gamma_A(Y-X^\top\beta)\big)^2\big].
\end{equation}
The objective weights the residual's component in $\mathrm{span}(A)$ by $\gamma$; at $\gamma=1$ it is ordinary least squares. Provided zero anchor-projected residual is attainable, the penalty drives the residual orthogonal to the anchor variable as $\gamma\to\infty$, recovering a prediction function invariant to shifts induced by the anchor variable \citep[Lemma~2]{rothenhausler2021anchor}. Such a shift replaces the term $\mathrm{M} A$ by a vector $\delta$ in the column space of $\mathrm{M}$, independent of the noise $\varepsilon$ though not of $A$, giving the perturbed system $W = \mathrm B W + \delta + \varepsilon$. The minimiser $\hat\beta_\gamma$ minimises the worst-case prediction risk over these shifts whose second moment satisfies $\E[\delta\delta^\top] \preceq \gamma\,\E\big[(\mathrm{M} A)(\mathrm{M} A)^\top\big]$, so larger $\gamma$ protects against larger shifts. The independence between $\delta$ and $\varepsilon$ is what the worst-case identity uses \citep[Theorem~1 and Section~8.5]{rothenhausler2021anchor}.

DRIG uses several observed environments with distributions $\{P_e\}_{e\in\mathcal E}$, one observational and the others interventional. For each $e$, $W_e=(X_e,Y_e)$ has distribution $P_e$. Its robustness result assumes, in every environment, the linear structural model $W_e=\mathrm BW_e+\varepsilon_e$, with $(\mathrm{Id}-\mathrm B)$ invertible and $\mathrm B$ invariant across environments.
Changes in the distribution of $\varepsilon_e$ induce changes in $W_e$ through the invariant structural equations. Let $e_0$ be the observational environment and let $\omega_e\geq0$ with $\sum_e\omega_e=1$. Writing $S_e=\E_{P_e}[\varepsilon_e\varepsilon_e^\top]$, DRIG assumes $S_{e_0}\preceq S_e$ for every $e$ \citep[Assumption~1]{shen2026causality}. In particular, $\sum_e\omega_e(S_e-S_{e_0})\succeq0$. Its population objective is
\begin{equation}\label{eq:drig_explicit}
\mathcal{L}_{\mathrm{DRIG}}(\beta; \gamma)
=
\E_{P_{e_0}}\big[(Y - X^\top\beta)^2\big]
+
\gamma \sum_{e \in \mathcal{E}} \omega_e
\Big(
\E_{P_e}\big[(Y - X^\top\beta)^2\big]
-
\E_{P_{e_0}}\big[(Y - X^\top\beta)^2\big]
\Big),
\end{equation}
the reference risk plus $\gamma$ times the weighted excess risk across environments. It equals the reference risk at $\gamma=0$ and the pooled risk at $\gamma=1$. For a random perturbation vector $\delta$, let $Q_\delta$ be the distribution induced by $W=(\mathrm{Id}-\mathrm B)^{-1}\delta$. The linear DRIG objective satisfies
\begin{equation}\label{eq:drig_dro}
\begin{aligned}
\mathcal{L}_{\mathrm{DRIG}}(\beta; \gamma)
&=
\sup_{Q \in \mathcal{B}_\gamma} \E_Q\big[(Y - X^\top\beta)^2\big], \\
\text{where}\quad \mathcal{B}_\gamma
&=
\Big\{ Q_\delta : \E[\delta\delta^\top] \preceq S_{e_0} + \gamma \sum_{e \in \mathcal{E}} \omega_e (S_e - S_{e_0}) \Big\}.
\end{aligned}
\end{equation}
Thus larger $\gamma$ permits stronger interventions in directions estimated from the observed environments \citep[Theorem~3]{shen2026causality}. Replacing $X^\top\beta$ by a general $g(X)$ still defines a valid training objective. The identity \eqref{eq:drig_dro}, however, is the robustness theorem for the linear structural model and linear prediction function above.

Both methods are used as benchmarks, and the DRIG objective in \eqref{eq:drig_explicit} enters the combined estimator in \Cref{sec:multienv}. DRIG evaluates risks under the observed heterogeneous distributions $P_e$ in \eqref{eq:drig_explicit}. Our synthetic environments are not substituted for these distributions; they arise from reweighting a single distribution according to the prediction, as illustrated in \Cref{fig:intervention_dags}.

Our construction is also related to covariate shift adaptation and importance weighting \citep{shimodaira2000improving,sugiyama2007covariate}, which likewise reweight the covariate distribution. Those methods target a specified test distribution. We instead generate a family of synthetic environments without specifying a target distribution. A related criterion is the conditional variance penalty of \citet{heinze2021conditional}. Their method uses a discrete identifier such as the person identity, and groups observations sharing the same response and identifier. It penalises within-group variation of the prediction or loss. Our groups are replaced by reweightings based on the prediction and require no identifier.

\subsection{Setup}\label{sec:setup}
Let $X\in\R^p$ be the covariate vector and $Y\in\R$ the response, with $\E[Y^2]<\infty$. For a random variable $Z$ with finite second moment, write $\|Z\|_P^2=\E_P[Z^2]$. For a measurable function $f$ of $X$, write $\|f\|_{P_X}^2=\E_{P_X}[f(X)^2]$; the space $L^2(P_X)$ consists of those functions with finite norm. For a measurable prediction function $g:\R^p\to\R$, write $F_g=g(X)$ for the resulting prediction random variable.
The prediction error is $Y-g(X)$, and $R_g=(Y-g(X))^2$ denotes its square. The conditional MSE given the prediction is
\begin{align*}
\sigma_g^2(t) = \E\big[R_g \,\big|\, F_g = t\big],
\end{align*}
and the marginal MSE $\bar\sigma_g^2=\E[R_g]$ is the risk, also denoted by $\cL(g)$. Let $P_g$ denote the distribution of $F_g$, and let $H_g(s)=P(F_g\leq s)$ be its distribution function.
We retain the subscript $g$ throughout.\footnote{We index these quantities by $g$, since the analysis compares prediction functions under the reference distribution $P$. They also depend on $X$ and $P_X$, but we suppress this dependence in the notation because $X$ and $P_X$ are held fixed as $g$ varies.} The conditional MSE is not in general a conditional variance; they agree when $\E[Y\,|\, g(X)]=g(X)$. We consider the following prediction class.
\begin{definition}[Prediction function class]\label{def:Gclass}
$\cG$ is a class of measurable functions $g : \R^p \to \R$ such that, for every $g \in \cG$, $g(X) \in L^2(P_X)$, $\bar\sigma_g^2 > 0$, and $H_g$ is continuous.
\end{definition}

The causal mechanism for $Y$ is specified as follows.
\begin{assumption}[Structural additive noise model]\label{ass:structural}
We assume $Y \gets g_0(X) + \varepsilon$ with $g_0 \in \cG$, $\varepsilon \indep g_0(X)$, $\E[\varepsilon] = 0$, and $\Var(\varepsilon) = \sigma^2 > 0$.
\end{assumption}

We use the assignment symbol to specify the structural mechanism for $Y$, rather than merely an equality under the observational distribution. This structural interpretation cannot be tested from one distribution alone. In a standard additive-noise structural equation model, $g_0$ is the structural function of the parents of $Y$, and independence of the noise from those parents implies $\varepsilon\indep g_0(X)$. This condition, imposed in \Cref{ass:structural}, is weaker than $\varepsilon\indep X$. Thus $X$ may contain descendants of $Y$, provided $g_0$ does not depend on their coordinates. Hidden confounding is also allowed when $\varepsilon$ remains independent of $g_0(X)$, even if it is correlated with other components of $X$.

Under \Cref{ass:structural}, the prediction error of $g_0$ is $\varepsilon$, independent of $g_0(X)$. Hence its conditional MSE is constant given the predicted value. In particular, $\sigma_{g_0}^2(t)=\Var(\varepsilon)=\sigma^2$ for $P_{g_0}$-almost every $t$. This motivates the class defined as follows.

\begin{definition}[Invariant prediction functions]\label{def:gi} The class of invariant prediction functions is
\begin{align}\label{eq:invariance_restriction}
    \cG_I \coloneqq \big\{g \in \cG : \sigma_g^2(t) \equiv c_g\ \mbox{for}\ P_g\mbox{-almost every}\ t, \ \mbox{for some constant}\ c_g \big\}.
\end{align}
\end{definition}

Since $\sigma_{g_0}^2$ is constant, $g_0\in\cG_I$. \Cref{sec:reweight} shows that $g\in\cG_I$ exactly when its risk is invariant across the synthetic environments generated by reweighting through the prediction $g(X)$.  Although constant conditional MSE restricts the prediction class, it constrains only the conditional second moment of the prediction error $Y-g(X)$, allowing its conditional distribution to vary in other respects. Other prediction functions may also satisfy this restriction, so it does not by itself identify $g_0$. We therefore impose the additional orthogonality condition in \Cref{ass:ortho} to identify $g_0$ from $P$ alone.

The main notation is collected below.
\begin{center}
\small
\begin{tabular}{@{}ll@{}}
\toprule
Symbol & Definition \\
\midrule
$F_g=g(X)$ & Prediction random variable induced by $g$ \\
$R_g=(Y-g(X))^2$ & Squared prediction error \\
$H_g(s)=P(F_g\leq s)$ & Distribution function of the prediction \\
$\sigma_g^2(t)=\E[R_g\,|\,F_g=t]$ & Conditional MSE given the prediction \\
$\bar\sigma_g^2=\E[R_g]=\cL(g)$ & Marginal MSE, or prediction risk \\
$U_g=H_g(F_g)$ & Rank of the prediction \\
$\psi_g(u)$ & Centred conditional MSE on the rank scale \\
$\Psi_g(u)=\int_0^u\psi_g(v)\,\mathrm{d}v$ & Cumulative conditional-MSE deviation \\
$\cT(g)=\int_0^1\Psi_g(u)^2\,\mathrm{d}u$ & CvM penalty \\
\bottomrule
\end{tabular}
\end{center}

\section{Reweighting and Synthetic Environments}\label{sec:reweight}

\Cref{def:gi} defines $\cG_I$ from a single observational distribution $P$, requiring no external environments. Yet the name \emph{invariant} points to behaviour across environments; we justify it by generating such environments from $P$, reweighting its covariate distribution.

For a nonnegative weight $w$ with $\E[w(X)]=1$, define $P_w$ by $\mathrm dP_w/\mathrm dP=w(X)$. Such a weight depends on $x$ but not on $y$.

\begin{proposition}\label{prop:reweight}
Suppose the weight factors as $w(x) = \tilde w(s(x))$ through a summary $V = s(X)$ of the covariates. Then, as regular conditional distributions, reweighting leaves the conditional distribution of $Y$ invariant given $V$ and given the full covariate vector $X$:
\begin{align*}
P_w(Y \,|\, V) = P(Y \,|\, V) \qquad \text{and} \qquad P_w(Y \,|\, X) = P(Y \,|\, X),
\end{align*}
each equality holding for $P_w$-almost every value of the conditioning variable.
More generally, if $w$ depends only on $X_{S_w}$, then $P_w(Y\,|\, X_S)=P(Y\,|\, X_S)$ for every $S\supseteq S_w$.
\end{proposition}

These reweighted distributions are synthetic covariate-shift environments. Only the covariate marginal changes, while both $P(Y\,|\,X)$ and $P(Y\,|\,V)$ are preserved. The proof is provided in \Cref{app:reweight}. This differs from anchor regression and DRIG (\Cref{fig:intervention_dags}), where shifts enter structural equations and may affect $X$, $Y$, and hidden variables jointly. Our construction generates covariate shifts from the available distribution rather than from observed interventions.

\begin{figure}[!tb]
\centering
\begin{tikzpicture}[
    >={Stealth[length=2mm]},
    every node/.style={font=\small},
    var/.style={circle, draw, minimum size=8mm, inner sep=1pt, fill=gray!10},
    hid/.style={circle, draw, dashed, minimum size=8mm, inner sep=1pt},
    instr/.style={circle, draw=gray, dashed, text=gray, minimum size=8mm, inner sep=1pt},
    shiftedge/.style={->, thick},
    rwt/.style={-{Stealth[length=2.4mm]}, double, double distance=1pt, thick},
    ev/.style={font=\scriptsize},
]
% Panel (a): anchor variable / DRIG -- observed source perturbs X, Y, H
\begin{scope}[local bounding box=panelA]
  \node[hid] (Ha) at (0.9,1.35) {$H$};
  \node[var] (Aa) at (-2.2,0) {$A,e$};
  \node[var] (Xa) at (0,0) {$X$};
  \node[var] (Ya) at (1.8,0) {$Y$};
  \draw[shiftedge] (Aa) -- (Xa);
  \draw[shiftedge] (Aa) to[bend left=18] (Ha);
  \draw[shiftedge] (Aa) to[bend right=32] (Ya);
  \draw[->,dashed] (Ha) -- (Xa);
  \draw[->,dashed] (Ha) -- (Ya);
  \draw[->,thick] (Xa) -- (Ya);
  \node[ev] at (-0.2,-1.3) {$A$ or $e$ perturbs $X,Y,H$};
  \node[font=\bfseries\small] at (-0.2,-1.95) {(a) Anchor regression / DRIG};
\end{scope}
% Panel (b): ours -- reweight the covariate distribution through g(X); H present but untouched
\begin{scope}[shift={(5.0cm,0)}, local bounding box=panelB]
  \node[hid] (Hb) at (0.9,1.35) {$H$};
  \node[var] (Xb) at (0,0) {$X$};
  \node[var] (Yb) at (1.8,0) {$Y$};
  \draw[rwt] (0,-1.05) -- (Xb);
  \draw[->,dashed] (Hb) -- (Xb);
  \draw[->,dashed] (Hb) -- (Yb);
  \draw[->,thick] (Xb) -- (Yb);
  \node[ev] at (0,-1.28) {reweight $w(g(X))$};
  \node[ev] at (1.9,-0.6) {$P(Y \,|\, X)$ fixed};
  \node[font=\bfseries\small] at (0.9,-1.95) {(b) Synthetic environments};
\end{scope}
\end{tikzpicture}
\caption{The distribution shifts underlying anchor regression and DRIG (a) and our construction (b). In (a), an \emph{observed} exogenous source (the anchor variable $A$, or a change of environment $e$) can perturb the structural equations for $X$, for $Y$, or for the hidden confounder $H$. In (b), reweighting through the prediction $g(X)$ alters the covariate distribution while $P(Y \,|\, X)$ is preserved. Hidden confounding does not affect the robustness guarantee, only the identification of $g_0$, as discussed in \Cref{rem:confounding}.}
\label{fig:intervention_dags}
\end{figure}
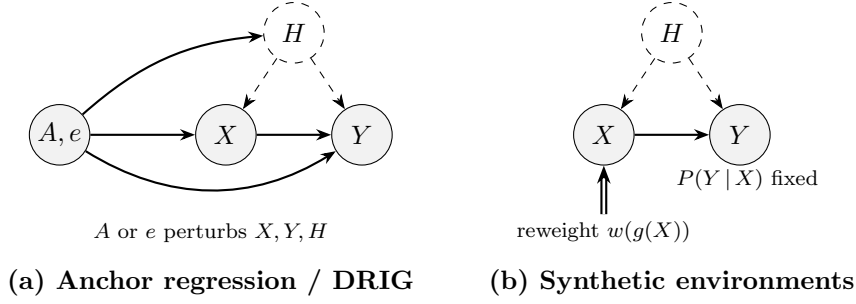

By definition, $g \in \cG_I$ exactly when its conditional MSE $\sigma_g^2(t)$ is constant along the prediction. We therefore consider weights of the form $w(x)=\tilde w(g(x))$, where $\tilde w\geq0$ and $\E[\tilde w(F_g)]=1$. Applying \Cref{prop:reweight} with $V=g(X)$ gives $P_w(Y\,|\, g(X))=P(Y\,|\, g(X))$. Thus $\sigma_g^2(t)=\E[(Y-t)^2\,|\, F_g=t]$ is unchanged, while $\tilde w$ reweights the marginal distribution of $F_g$. The risk under $P_w$ is
\begin{align*}
    \E_{P_w}[R_g] = \E_P\big[\tilde w(F_g)\,R_g\big] = \E_P\big[\tilde w(F_g)\,\sigma_g^2(F_g)\big].
\end{align*}
This risk equals $\bar\sigma_g^2$ for every such $\tilde w$ if and only if $\sigma_g^2(F_g)$ is almost surely constant. Thus, $g\in\cG_I$ exactly when its \emph{risk is invariant under every reweighting through $g(X)$}. This scalar characterisation avoids an exhaustive search over covariate subsets.

One way to check whether a given $g$ lies in $\cG_I$ is to compute its largest violation over normalised weights, formulated as
\begin{align}
\label{eq:adver_search}
    \sup_{\tilde w \geq 0,\,\E[\tilde w(F_g)] = 1} \Big|\E\big[\tilde w(F_g)\,R_g\big] - \bar{\sigma}_g^2\Big|,
\end{align}
an adversarial search over a class of admissible weights.  The supremum is zero precisely when $g \in \cG_I$, so it detects invariance exactly. The search is impractical, though, because it depends on the choice of weight class and ranges over infinitely many $\tilde w$.

A single one-parameter family of weights, however, characterises $\cG_I$, so the adversarial search in \eqref{eq:adver_search} is unnecessary. For thresholds $s$ with $P(F_g \le s) > 0$, consider the indicator weights on the prediction scale,
\begin{align*}
    \tilde w_{g,s}(t) = \frac{\mathbbm{1}\{t \leq s\}}{P(F_g \leq s)}.
\end{align*}
For the weight $\tilde w_{g,s}$, the marginal risk under the reweighted distribution is
\begin{align*}
    \E\big[\tilde w_{g,s}(F_g)\,R_g\big] = \frac{\E\big[R_g\,\mathbbm{1}\{F_g \leq s\}\big]}{P(F_g \leq s)} = \E\big[R_g \,|\, g(X) \leq s\big],
\end{align*}
the risk of $g$ conditional on predictions not exceeding $s$. Reweighting by the normalised indicator thus amounts to conditioning on this event. Setting this conditional risk equal to $\bar{\sigma}_g^2$ at every $s$ gives $\E\big[R_g\,\mathbbm{1}\{F_g \leq s\}\big] = \bar{\sigma}_g^2\,P(F_g \leq s)$. This equality holds for every $s$ exactly when $g \in \cG_I$, so the adversarial search over all weights reduces to this one-parameter family, as formulated below.

\begin{proposition}[Indicator weights are a complete class]\label{prop:indicators}
$g \in \cG_I$ if and only if $\E\big[R_g \, \mathbbm{1}\{F_g \leq s\}\big] = \bar{\sigma}_g^2\,P(F_g \leq s)$ for all $s \in \R$.
\end{proposition}

The proof is given in \Cref{app:indicators}. This equivalence reduces the characterisation of $\cG_I$ to a one-parameter family of moment conditions, which the next section aggregates over $s$ into a single penalty.

\section{The CvM Penalty}\label{sec:cvm}
In \Cref{prop:indicators}, each threshold $s$ determines a separate condition. We combine these conditions by squaring each violation and integrating against the distribution of the prediction. The result is a single measure of how far $g$ departs from invariance, the Cram\'er--von Mises (CvM) criterion:

\begin{align}\label{eq:cvm}
\cT(g) \coloneqq \int \big(\E\big[R_g\,\mathbbm{1}\{F_g \leq s\}\big] - \bar{\sigma}_g^2P(F_g \leq s)\big)^2 \mathrm{d}P_g(s).
\end{align}

The criterion thus weights each threshold $s$ by the probability mass the prediction $F_g$ places near it, so thresholds in the tails of the prediction contribute less.

The two terms inside the parentheses are each $\bar{\sigma}_g^2$ times a distribution function of the prediction. One is the ordinary distribution function $H_g(s) = P(F_g \le s)$; the other is the squared-error-weighted $\widetilde{H}_g(s) \coloneqq \E\big[R_g\,\mathbbm{1}\{F_g \le s\}\big]/\bar{\sigma}_g^2$, which reweights the distribution of the prediction by $R_g$. Thus
\begin{align*}
\cT(g) = \bar{\sigma}_g^4 \int \big(\widetilde{H}_g(s) - H_g(s)\big)^2 \,\mathrm{d}P_g(s).
\end{align*}
The integral is the Cram\'er--von Mises distance between $\widetilde H_g$ and $H_g$ \citep{anderson1952asymptotic}, scaled by $\bar\sigma_g^4$. The two distribution functions coincide exactly when $\sigma_g^2$ is constant, giving the following equivalence.

\begin{proposition}[CvM penalty and invariance]\label{prop:completeness}
$\cT(g) = 0$ if and only if $g \in \cG_I$.
\end{proposition}
The proof is given in \Cref{app:completeness}. By \Cref{prop:completeness}, $\cT(g)$ is zero exactly on $\cG_I$ and strictly positive otherwise, so it measures departure from invariance. We therefore propose the CvM-penalised objective
\begin{align*}
    \cL_{\mathrm{CvM}}(g; \lambda) \coloneqq \E\big[\big(Y - g(X)\big)^2\big] + \lambda\,\cT(g),
\end{align*}
in which $\lambda \geq 0$ controls the trade-off between prediction risk and departure from invariance. Write the set of population minimisers as
\begin{align}\label{eq:estimator}
    \cG_\lambda^\ast \coloneqq \argmin_{g \in \cG}\,\cL_{\mathrm{CvM}}(g; \lambda),
\end{align}
and let $g_\lambda$ denote any element of $\cG_\lambda^\ast$. The set need not be a singleton. In \Cref{sec:ident} we give conditions under which its elements approach the causal function $g_0$ as $\lambda \to \infty$.

\section{Robustness and identification}\label{sec:dro}
The same penalty has two roles. For finite $\lambda>0$, it yields robustness against prediction-dependent reweightings. As $\lambda\to\infty$, it enforces the invariance restriction used for identification. At $\lambda=0$, the population objective reduces to MSE minimisation; its empirical counterpart is least-squares ERM. 

\subsection{Perturbation class}\label{sec:dro:class}
To state the finite-$\lambda$ robustness result, we first define the perturbation family represented by the CvM penalty. Put $U_g=H_g(F_g)$, the rank of the prediction, so that $U_g\sim\mathrm{Uniform}(0,1)$. The perturbations are indexed by the Hilbert space
\begin{align}\label{eq:H}
\cH = \Big\{\, h : [0,1] \to \R \,\Big|\, h(1) = 0, h \text{ absolutely continuous}, \int_0^1 (h'(u))^2\,\mathrm{d}u < \infty \,\Big\},
\end{align}
equipped with $\langle h_1,h_2\rangle_\cH=\int_0^1h_1'(u)h_2'(u)\,\mathrm du$. Here $u$ is a point in $[0,1]$. Since $U_g\sim\mathrm{Uniform}(0,1)$, every integrable function $f$ satisfies $\int_0^1 f(u)\,\mathrm{d}u=\E[f(U_g)]$.

To each $h \in \cH$ we associate a normalised signed perturbation $Q_h$ on $\R^p \times \R$, which depends on the prediction function $g$ through the rank $H_g(g(x))$, with density ratio
\begin{align}\label{eq:Qh}
\frac{\mathrm{d}Q_h}{\mathrm{d}P}(x,y)=
w_h(x)
\coloneqq
1 + h\big[H_g\big(g(x)\big)\big] - \int_0^1 h(u)\,\mathrm{d}u.
\end{align}
Because $U_g \sim \mathrm{Uniform}(0,1)$,
\begin{align*}
\E_P\big[w_h(X)\big]=1 + \E\big[h(U_g)\big] - \int_0^1 h(u)\,\mathrm{d}u
=1,
\end{align*}
so $Q_h$ has total mass one. Since $w_h$ depends only on $x$, $Q_h$ factorises as
\begin{align*}
Q_h(\mathrm dx,\mathrm dy)=w_h(x)P_X(\mathrm dx)P_{Y\,|\, X}(\mathrm dy\,|\, x).
\end{align*}
Thus only the covariate marginal is reweighted. If $w_h\geq0$, $Q_h$ is a probability distribution; otherwise it is a signed measure.

\subsection{Distributional robustness equivalence}\label{sec:dro:rkhs}
For $Q_h$ in \eqref{eq:Qh}, consider the risk difference $\E_{Q_h}[R_g]-\E_P[R_g]$.  Define the centred conditional MSE on the quantile scale by $\psi_g(u)=\sigma_g^2(H_g^{-1}(u))-\bar\sigma_g^2$, where $H_g^{-1}(u)=\inf\{t:H_g(t)\geq u\}$. The next result represents the risk difference and the CvM penalty in the same Hilbert space.

\begin{proposition}[The CvM penalty as a squared RKHS norm]\label{thm:rkhs}
The space $\cH$ in \eqref{eq:H} is a reproducing kernel Hilbert space (RKHS) with kernel $k_0(v_1,v_2) = 1 - \max(v_1,v_2)$ for $v_1,v_2 \in [0,1]$. For every $g \in \cG$, the kernel mean embedding $\mu_g \coloneqq \int_0^1 \psi_g(v)\,k_0(v,\,\cdot\,)\,\mathrm{d}v$ belongs to $\mathcal{H}$ and satisfies, for every $h \in \cH$,
\begin{align}\label{eq:risk-change}
\E_{Q_h}[R_g] - \E_P[R_g] = \int_0^1 \psi_g(u)\,h(u)\,\mathrm{d}u = \langle h, \mu_g \rangle_\cH,
\end{align}
and
\begin{align*}
\cT(g) = \big\|\mu_g\big\|_\mathcal{H}^2 = \int_0^1 \Psi_g(u)^2\,\mathrm{d}u,
\qquad \Psi_g(u) \coloneqq \int_0^u \psi_g(v)\,\mathrm{d}v.
\end{align*}
\end{proposition}
The proof is in \Cref{app:rkhs}. The two identities provide the required link. The CvM criterion is the squared RKHS norm of $\mu_g$, while the inner product of $\mu_g$ with $h$ gives the change in risk under the corresponding reweighting. This is an instance of representing a distributional discrepancy by an RKHS norm \citep{gretton2012kernel,sejdinovic2013equivalence}.

For every $\lambda>0$, completing the square gives
\begin{align}
\label{eq:complete_square}
\langle h,\mu_g\rangle_\cH-\frac{\|h\|_\cH^2}{4\lambda}
=\lambda\|\mu_g\|_\cH^2-\frac{\|h-2\lambda\mu_g\|_\cH^2}{4\lambda}\leq\lambda\|\mu_g\|_\cH^2,
\end{align}
with equality only at $h=2\lambda\mu_g$. By \Cref{thm:rkhs}, the inner product is the change in risk in \eqref{eq:risk-change}. Taking the supremum therefore turns $\E_P[R_g]+\lambda\cT(g)$ into a regularised worst-case risk over $\{Q_h\}$.

\begin{theorem}[Distributional robustness equivalence]\label{thm:dro}
For all $\lambda>0$ and $g\in\cG$,
\begin{equation}\label{eq:dro}
    \E_P[R_g] + \lambda\,\cT(g)
= \sup_{h \in \cH} \bigg\{\E_{Q_h}\big[\big(Y - g(X)\big)^2\big] - \frac{\|h\|_\cH^2}{4\lambda}\bigg\},
\end{equation}
where $Q_h$ is the normalised perturbation defined in \eqref{eq:Qh}, and the supremum is attained at the unique maximiser $h^\ast = 2\lambda\mu_g$. Since \eqref{eq:dro} holds for every $g$, the CvM-penalised estimator $g_\lambda$ of \eqref{eq:estimator} also minimises the worst-case risk:
\begin{equation}\label{eq:dro_minimax}
g_\lambda \in \argmin_{g\in\cG} \sup_{h \in \cH} \bigg\{\E_{Q_h}\big[\big(Y - g(X)\big)^2\big] - \frac{\|h\|_\cH^2}{4\lambda}\bigg\}.
\end{equation}
\end{theorem}
For a fixed $g$, \eqref{eq:dro} identifies the CvM-penalised risk with a regularised worst-case risk over the reweightings $Q_h$. The minimisation in \eqref{eq:dro_minimax} therefore chooses the prediction function with the smallest such worst-case value. When $Q_h$ is a probability distribution, it changes the relative frequency of prediction levels while preserving the conditional distribution of $(X,Y)$ within each level. At $g=g_0$ in the standard structural model above, the reweighting depends on the parents only through $g_0(X)$ and does not allow arbitrary changes in their joint distribution.

The bound in \eqref{eq:complete_square} shows directly that the supremum is finite. Larger values of $\lambda$ reduce the quadratic cost of a given perturbation. For a fixed $g$, the supremum in \eqref{eq:dro} chooses $h$ to maximise the risk under $Q_h$ minus the quadratic cost $\|h\|_\cH^2/(4\lambda)$. Increasing $\lambda$ makes a given perturbation less costly. The maximiser is $h^\ast=2\lambda\mu_g$, whose size satisfies $\|h^\ast\|_\cH=2\lambda\cT(g)^{1/2}$ and is therefore proportional to $\lambda$. Thus $\lambda$ controls the size of the maximising perturbation through its cost, while the family $\{Q_h\}$ remains unchanged. The proof is in \Cref{app:dro}.

The exact representation permits signed $Q_h$ because $w_h$ need not be nonnegative. The next corollary restricts the supremum to probability reweightings. Its proof is in \Cref{app:signed}.
\begin{corollary}[Restriction to probability reweightings]\label{cor:signed}
The penalised objective upper-bounds the worst-case risk over the probability distributions in the perturbation family:
\begin{align*}
\sup_{h \in \cH \,:\, w_h \geq 0} \bigg\{\E_{Q_h}\big[\big(Y-g(X)\big)^2\big] - \frac{\|h\|_\cH^2}{4\lambda}\bigg\} \leq \E_P[R_g] + \lambda\,\cT(g).
\end{align*}
Equality holds if $w_{h^\ast}\geq 0$; in particular, it holds when $\lambda\leq\{4\cT(g)^{1/2}\}^{-1}$. Thus the signed formulation is exact for probability reweightings on this range and otherwise gives an upper bound. 
\end{corollary}

Unlike the distance-based DRO construction in \Cref{sec:introduction}, the family $\{Q_h\}$ is induced by prediction-dependent reweighting and, by the factorisation in \Cref{sec:dro:class}, changes the covariate factor alone. Changes in the conditional response distribution are therefore outside the guarantee. Kernel DRO also measures distributional shifts using an RKHS norm \citep{staib2019distributionally,zhu2021kernel}, but the kernel and the form of perturbation class are specified externally. By contrast, in \eqref{eq:Qh}, the observational distribution and $g$ determine the rank coordinate $H_g(g(x))$ on which $h$ acts, while the CvM criterion specifies the RKHS in \Cref{thm:rkhs}.

\subsection{Identification}\label{sec:ident}

At finite $\lambda$, \Cref{thm:dro} gives the robustness guarantee. As $\lambda\to\infty$, the penalty forces $\cT(g_\lambda)\to0$, so any $L^2(P_X)$-limit point lies in $\cG_I$. Identification then requires $g_0$ to be the unique risk minimiser over $\cG_I$. This can fail when another invariant prediction function is correlated with $\varepsilon$ and attains smaller risk. The following condition rules out this case.
\begin{assumption}[Orthogonality]\label{ass:ortho} For every $g\in\cG_I$, $\E\big[g(X)\varepsilon\big]=0$. 
\end{assumption}

\Cref{ass:ortho} rules out prediction functions whose conditional MSE given the prediction remains constant but whose prediction still correlates with the structural noise. Such a function can arise under hidden confounding. A prediction function that uses the confounded covariates may have constant conditional MSE given its prediction, so it lies in $\cG_I$, while its prediction and $\varepsilon$ both depend on $H$ and are correlated, so $\E[g(X)\varepsilon]\neq 0$; $g_0$ then need not minimise risk over $\cG_I$.
Under \Cref{ass:ortho}, $g_0$ minimises the risk over $\cG_I$, and every minimiser agrees with $g_0$ in $L^2(P_X)$, as the next proposition states.

\begin{proposition}[Unique minimiser in $L^2(P_X)$]\label{prop:unique}
Under \cref{ass:structural,ass:ortho}, $g_0$ minimises $\E\big[\big(Y-g(X)\big)^2\big]$ over $\cG_I$, and every minimiser $g$ satisfies $\|g-g_0\|_{P_X}=0$.
\end{proposition}
The proof is given in \Cref{app:unique}.

\begin{remark}[The jointly Gaussian case]\label{rem:gaussian_linear}
Gaussianity of $\varepsilon$ alone is not a problem, even when $g_0$ is linear. Suppose instead that $(X,Y)$ is jointly Gaussian and that $\cG$ contains the population least-squares predictor $g_{\rm LS}$. Then $g_{\rm LS}$ has the smallest squared-error risk, and its residual is independent of $X$. Hence $g_{\rm LS}\in\cG_I$. If $g_{\rm LS}\ne g_0$ in $L^2(P_X)$, then $\cL(g_{\rm LS})<\cL(g_0)$ and
\begin{align*}
\E\big[g_{\rm LS}(X)\varepsilon\big]
=\frac{1}{2}\big\{\cL(g_0)-\cL(g_{\rm LS})
+\|g_{\rm LS}-g_0\|_{P_X}^2\big\}>0.
\end{align*}
Thus $g_{\rm LS}$ directly violates \Cref{ass:ortho}. Its risk is smaller than that of $g_0$ and its CvM penalty is zero, so it minimises $\cL_{\rm CvM}(\cdot;\lambda)$ for every $\lambda$.
\end{remark}

\begin{remark}[The role of hidden confounding]\label{rem:confounding}
Because the independence in \Cref{ass:structural} is imposed against the prediction $g_0(X)$ rather than the covariate vector $X$, a hidden confounder $H$ acting on both $X$ and $Y$ is admitted. Such confounding does not affect the covariate-shift guarantee of \Cref{thm:dro}. Reweighting through the prediction preserves $P(Y \,|\, X)$ for every weight and every distribution of $H$. $H$ enters instead at the identification step, as discussed after \Cref{ass:ortho}.
\end{remark}

The remaining step is continuity of $\cT$ under $L^2(P_X)$ convergence, established in \Cref{lem:continuity}. It ensures that a limit point of $g_\lambda$ whose penalty converges to $0$ belongs to $\cG_I$.

\begin{proposition}[Identification]\label{prop:identification}
Under \Cref{ass:structural,ass:ortho}, any $L^2(P_X)$-limit point of $g_\lambda$ as $\lambda \to \infty$ that lies in $\cG$ coincides with $g_0$ in $L^2(P_X)$.
\end{proposition}
The proof is in \Cref{app:identification}. If the structural additive noise model is misspecified, the same limit selects the lowest-risk prediction functions satisfying the observable invariance restriction.

\begin{proposition}[Target under misspecification]\label{prop:misspecified}
Suppose $\cG_I$ is nonempty. Every $L^2(P_X)$-limit point of $g_\lambda$ as $\lambda\to\infty$ that lies in $\cG$ belongs to
\begin{align*}
\argmin_{g\in\cG_I}\cL(g),
\quad \text{where }
\cL(g)=\E\big[(Y-g(X))^2\big].
\end{align*}
\end{proposition}
The set in the display may contain several functions. Under misspecification, distinct invariant prediction functions can attain the same minimum observational risk. The proof is in \Cref{app:misspecified}.

\section{Estimation and inference}\label{sec:theory}
We now turn to estimation from an i.i.d.\ sample $\{(X_i, Y_i)\}_{i=1}^n$. Write $P_n$ for the empirical measure of this sample. This section develops the empirical version of the objective, establishes consistency and a rate of convergence for the resulting estimator of $g_0$, and closes with a goodness-of-fit test.

\subsection{The empirical criterion}\label{sec:empirical}

Given a sample, we estimate $\cT(g)$ in \eqref{eq:cvm} by 
\begin{align}\label{eq:sample_cvm}
    \hat{\cT}_n(g) = \frac1n\sum_{j=1}^{n} \Big(\frac1n\sum_{i=1}^{n} R_{g,i}\mathbbm{1}\big\{F_{g,i} \leq F_{g,j}\big\} - \bar{R}_g\cdot \frac1n\sum_{i=1}^{n}\mathbbm{1}\big\{F_{g,i} \leq F_{g,j}\big\}\Big)^2.
\end{align} Here, $F_{g,i}=g(X_i)$, $R_{g,i}=\big(Y_i-g(X_i)\big)^2$, and $\bar R_g=n^{-1}\sum_{i=1}^nR_{g,i}$.
The corresponding estimator of $g_\lambda$ in \eqref{eq:estimator} is
\begin{align}\label{eq:finite_estimator}
\hat{g}_{n,\lambda_n} \in \argmin_{g \in \cG} \frac{1}{n}\sum_{i=1}^{n} \big(Y_i - g(X_i)\big)^2 + \lambda_n\,\hat{\cT}_n(g).
\end{align}
Both the consistency theory of \Cref{sec:asymptotics} and the goodness-of-fit test of \Cref{sec:gof} are built on these two objects.

After sorting the fitted values, computing the empirical penalty \eqref{eq:sample_cvm} takes $O(n\log n)$ time. Obtaining the estimator \eqref{eq:finite_estimator}, however, may require computing gradients when the function class $\cG$ is flexible, as with the neural network classes in \Cref{sec:experiments}, and there the indicator $\mathbbm{1}\{F_{g,i} \le F_{g,j}\}$ makes $\hat{\cT}_n(g)$ non-differentiable in $g$. We therefore replace the indicator by the sigmoid
\begin{align}
\label{eq:anneal}
\sigma\big((F_{g,j} - F_{g,i})/\tau\big)
=\frac{1}{1 + e^{-(F_{g,j} - F_{g,i})/\tau}},
\end{align}
which approaches the indicator away from ties as the temperature $\tau>0$ decreases.\footnote{At a tie, the sigmoid equals $1/2$, whereas the hard indicator equals $1$.} In practice, one can take $\tau$ small or anneal it during training. The guarantees of \Cref{sec:asymptotics,sec:gof} are stated for the exact criterion \eqref{eq:sample_cvm}. Analysing the smoothed criterion would additionally require control of the approximation error.

\subsection{Consistency and rates}\label{sec:asymptotics}

\Cref{sec:ident} identifies the population target, $g_0$. Here we show that the estimator $\hat{g}_{n,\lambda_n}$ converges to $g_0$ as the sample size grows, with the distance between $\hat{g}_{n,\lambda_n}$ and $g_0$ bounded in terms of the sample size $n$, the complexity of $\cG$, and the penalty parameter $\lambda_n$. The bound indicates how $\lambda_n$ may be chosen.

For consistency, we impose the following conditions on $\cG$. We assume that all function classes in the empirical-process arguments are pointwise measurable and that the displayed minimisers exist.
\begin{assumption}[Regularity]\label{ass:regularity} The class $\cG$ defined in \Cref{def:Gclass} satisfies:
\begin{enumerate}[label=(\alph*)]
\item the pseudo-dimension of $\cG$ is $V$, with $1\le V<\infty$, and $\cG_c\coloneqq\{g-s:g\in\cG,\ s\in\R\}$ satisfies $\mathrm{Pdim}(\cG_c)=O(V)$; 
\item there exists an envelope $G$ with $\E\big[G(X)^4\big] < \infty$, and $\E\big[Y^4\big] < \infty$;
\item $\cG$ is closed in $L^2(P_X)$;
\item the class $\Big\{x \mapsto \nu\big(g,\, g(x)\big)^2 : g \in \cG\Big\}$ is $P$-Glivenko-Cantelli, where $\nu(g,s) = \E\big[R_g\,\mathbbm{1}\{F_g \le s\}\big] - \bar\sigma_g^2\,P(F_g \le s)$.
\end{enumerate}
\end{assumption}

Condition~(a) controls the enlarged class obtained by subtracting arbitrary constants, and its bound is immediate if $\cG$ is closed under addition of constants. Condition~(d) ensures the remaining uniform law of large numbers for the deviation class $\{x\mapsto\nu(g,g(x))^2:g\in\cG\}$, whose complexity is not determined by $\mathrm{Pdim}(\cG)$ alone.

For rates we additionally control the complexity of the deviation class.
\begin{assumption}[Deviation class complexity]\label{ass:reg_rate}
The class $\big\{x \mapsto \nu(g,g(x))^2 : g \in \cG\big\}$, with $\nu$ defined in \Cref{ass:regularity}(d), has pseudo-dimension at most $V_2$, and $V_2 \le c\,V$ for an absolute constant $c \ge 1$.
\end{assumption}
\Cref{ass:reg_rate} is used only for the explicit rate. It must be established for the chosen class by bounding the pseudo-dimension of the composed maps $x\mapsto\nu(g,g(x))^2$; it does not follow from $\mathrm{Pdim}(\cG)=V$.

Because $\hat{g}_{n,\lambda_n}$ is data-dependent, the bounds used at it must hold uniformly over $\cG$.
\begin{theorem}[Uniform convergence of the empirical criterion]\label{thm:uniform_t}\label{thm:rate_uniform}
If \Cref{ass:regularity} is satisfied, then the following statements hold:
\begin{enumerate}[label=(\alph*)]
\item
The penalty converges uniformly,
\begin{align*}
\sup_{g\in\cG}\big|\hat\cT_n(g)-\cT(g)\big|\overset{P}{\longrightarrow}0.
\end{align*}
For every fixed $\lambda\geq0$, the full empirical criterion also converges uniformly,
\begin{align*}
\sup_{g\in\cG}\Big|
P_n[(Y-g(X))^2]+\lambda\hat\cT_n(g)
-\cL_{\rm CvM}(g;\lambda)
\Big|\overset{P}{\longrightarrow}0.
\end{align*}
\item
If in addition \Cref{ass:reg_rate} holds, then
\begin{align*}
\sup_{g\in\cG}\big|\hat\cT_n(g)-\cT(g)\big|
&=O_P\big(\sqrt{V/n}\big),\\
\sup_{g\in\cG}\Big|
P_n[(Y-g(X))^2]+\lambda\hat\cT_n(g)
-\cL_{\rm CvM}(g;\lambda)
\Big|
&=O_P\big((1+\lambda)\sqrt{V/n}\big),
\quad \text{for every fixed }\lambda\geq0.
\end{align*}
\end{enumerate}
\end{theorem}
At $\lambda=0$, the second bound is the uniform empirical-risk rate used in the proof of \Cref{thm:consistency_g}(b).
The proof is given in \Cref{app:uniform_t}.

The uniform convergence above connects the empirical estimator to its population target at any fixed penalty level.
\begin{proposition}[Fixed-penalty estimation]\label{prop:fixed_lambda}
Fix $\lambda\geq0$ and let $\hat g_{n,\lambda}$ be any empirical minimiser in \eqref{eq:finite_estimator}. Under \Cref{ass:regularity}, as $n\to\infty$, the estimator converges to the set $\cG_\lambda^\ast$ defined in \eqref{eq:estimator}, in the sense that the $L^2(P_X)$ distance from $\hat g_{n,\lambda}$ to $\cG_\lambda^\ast$ converges to $0$ in probability, and
\begin{align*}
\cL_{\rm CvM}(\hat g_{n,\lambda};\lambda)
-\inf_{g\in\cG}\cL_{\rm CvM}(g;\lambda)
&\overset{P}{\longrightarrow}0.
\end{align*}
If $\cG_\lambda^\ast=\{g_\lambda\}$, then
\begin{align*}
\|\hat g_{n,\lambda}-g_\lambda\|_{P_X}\overset{P}{\longrightarrow}0.
\end{align*}
\end{proposition}
The proof of the fixed-penalty result is in \Cref{app:fixed_lambda}.

A fixed $\lambda$ targets $\cG_\lambda^\ast$. Identification of $g_0$ requires the penalty to grow with $n$. For the growing-penalty estimator, a distance bound requires small values of the penalty $\cT(g)$ and of the excess risk $\cL(g)-\cL(g_0)$ to force $g$ close to $g_0$. The two quantities are collected in the functional
\begin{align*}
    M(g) &\coloneqq \cT(g) + \big(\cL(g)-\cL(g_0)\big)_{+},\quad \mbox{where}
    (a)_{+} = \max(a,0),
\end{align*}
and the requirement becomes a growth condition on $M$.
\begin{assumption}[Quadratic curvature]\label{ass:curvature}
There exist constants $c_0, \rho > 0$ such that
\begin{align*}
    M(g) \ge c_0\,\|g - g_0\|_{P_X}^2 \quad \text{for every } g \in \cG \text{ with } \|g - g_0\|_{P_X} \le \rho.
\end{align*}
\end{assumption}
Under \Cref{ass:structural,ass:ortho}, $M(g)=0$ forces $\cT(g)=0$ and $\cL(g)\leq\cL(g_0)$, so \Cref{prop:completeness,prop:unique} gives $\|g-g_0\|_{P_X}=0$; conversely, $M(g_0)=0$. \Cref{ass:curvature} strengthens this zero-set statement to local quadratic growth.

\begin{theorem}[Consistency and rate of $\hat{g}_{n,\lambda_n}$]\label{thm:consistency_g}\label{thm:rate_g}
Suppose \Cref{ass:regularity,ass:structural,ass:ortho} are satisfied and $\lambda_n \to \infty$ with $\lambda_n = o(n)$. Then the following statements hold:
\begin{enumerate}[label=(\alph*)]
\item $\|\hat{g}_{n,\lambda_n} - g_0\|_{P_X}^2 \to 0$ in probability.
\item If in addition \Cref{ass:reg_rate,ass:curvature} hold, then
\begin{align}\label{eq:rate_main}
    \big\|\hat g_{n,\lambda_n} - g_0\big\|_{P_X}^2 = O_P\,\Big(\frac{1}{\lambda_n} + \sqrt{\frac{V}{n}} + \frac{\lambda_n}{n}\Big).
\end{align}
\end{enumerate}
\end{theorem}
The proof in \Cref{app:consistency_g} also shows $\cT(\hat g_{n,\lambda_n})\to0$ without \Cref{ass:ortho}. This conclusion uses only $\cT(g_0)=0$, which follows from \Cref{ass:structural}; it does not use the identification of $g_0$. The condition $\lambda_n=o(n)$ makes the empirical penalty at $g_0$ negligible. In \Cref{sec:nonlinear}, we investigate how the estimator performs in finite samples as $\lambda$ varies.

In the bound in \eqref{eq:rate_main}, $1/\lambda_n$ is the finite-penalty error, $\sqrt{V/n}$ is the uniform estimation error, and $\lambda_n/n$ is the empirical penalty at $g_0$. Balancing the first and third terms gives $\lambda_n\asymp\sqrt n$ and, for fixed $V$, a squared-error bound of order $n^{-1/2}$. This argument specifies an asymptotic order, not a data-driven value of $\lambda_n$ or its multiplicative constant. 

\subsection{Goodness-of-fit Test}\label{sec:gof}
The penalty also tests the restriction it enforces. For a prediction function $g$ determined independently of the sample on which the statistic is computed, we test
\begin{align*}
H_0: g \in \cG_I \qquad\text{against}\qquad H_1: g \notin \cG_I,
\end{align*}
that is, whether the conditional MSE $\sigma_g^2(t)$ is constant along the prediction. Since $\cT(g) = 0$ exactly on $\cG_I$, the empirical penalty $\hat\cT_n(g)$ measures the departure.

\begin{remark}[Scope of the test]\label{rem:gof_scope}
The null concerns invariance rather than causality. It contains every $g\in\cG_I$, and a single distribution cannot test the orthogonality condition that identifies $g_0$ within $\cG_I$. Thus non-rejection does not establish that $g=g_0$.
\end{remark}

We use the scale-free statistic
\begin{align*}
    \hat{T}_n(g) \coloneqq \frac{n \hat{\mathcal T}_n(g)}{\widehat{\Var}_n(R_g)},
    \qquad
    \widehat{\Var}_n(R_g) = \frac{1}{n}\sum_{i=1}^{n} R_{g,i}^2 - \Big(\frac{1}{n}\sum_{i=1}^{n} R_{g,i}\Big)^2 ,
\end{align*}
where dividing by $\widehat{\Var}_n(R_g)$ removes an unknown scale from the limit. We assume that $\Var(R_g)$ is finite and positive. The denominator $\widehat{\Var}_n(R_g)$ then converges in probability to the positive constant $\Var(R_g)$.

Under $H_0$, $\hat T_n(g)$ is a Cram\'er--von Mises statistic built on the process
\begin{align}\label{eq:marked_process}
V_n(s) = \frac{1}{\sqrt{n}} \sum_{i=1}^{n} \big(R_{g,i} - \bar{R}_g\big)\,\mathbbm{1}\{F_{g,i} \le s\},
\end{align}
an empirical process of the indicators $\mathbbm{1}\{F_{g,i} \le s\}$ in which each indicator is weighted by the centred squared error $R_{g,i} - \bar{R}_g$ of its observation. The process is the empirical counterpart of the comparison between $\widetilde{H}_g$ and $H_g$ in \Cref{sec:cvm}, since $V_n(s)$ equals $\sqrt{n}\,\bar{R}_g$ times the difference between their empirical versions, and the test statistic takes the form
\begin{align*}
\hat{T}_n(g) = \frac{1}{n\,\widehat{\Var}_n(R_g)} \sum_{j=1}^{n} V_n(F_{g,j})^2.
\end{align*}

\begin{corollary}[Null limit, adapted from \citealp{stute1997nonparametric}]\label{cor:null_general}
Let $g \in \cG_I$ with $\Var(R_g) < \infty$. Then $\hat{\cT}_n(g) = O_P(1/n)$. If moreover $\Var(R_g) > 0$, then $\hat{T}_n(g)$ converges in distribution to $\Var(R_g)^{-1} \int G^0(s)^2\,\mathrm{d}H_g(s)$, where $G^0$ is a mean-zero Gaussian process whose covariance is determined by the conditional variance $\Var(R_g \,|\, F_g)$ and by $H_g$.
\end{corollary}
The proof is given in \Cref{app:null}, where the covariance of $G^0$ is stated explicitly. The limit is in general not distribution-free, so its critical value is estimated by the wild bootstrap in \Cref{app:bootstrap}; \Cref{prop:bootstrap_validity} establishes the validity of this calibration. The test rejects $H_0$ when $\hat T_n(g)$ exceeds its critical value. It is consistent against every alternative $g \notin \cG_I$, since $\hat\cT_n(g) \to \cT(g) > 0$ there, so $\hat T_n(g) \to \infty$.

At $g = g_0$ the limit becomes distribution-free.
\begin{corollary}[Distribution-free null at $g_0$]\label{cor:null_g0}
Under \Cref{ass:structural}, with $0 < \Var(\varepsilon^2) < \infty$,
\begin{align*}
    \hat{T}_n(g_0) \overset{d}{\to} \int_0^1 B(u)^2 \,\mathrm{d}u,
\end{align*}
the classical Cram\'er--von Mises distribution, with $B$ a standard Brownian bridge.
\end{corollary}
The proof is given in \Cref{app:null_g0}. Under \Cref{ass:structural}, $R_{g_0}=\varepsilon^2$ and $\varepsilon\indep g_0(X)$. Hence $\Var(R_{g_0} \,|\, F_{g_0})$ is constant, and the limiting process has the covariance of a standard Brownian bridge. The distribution of $\int_0^1 B(u)^2\,\mathrm{d}u$ is classical, with known quantiles \citep{anderson1952asymptotic}, for instance $c_{0.05} = 0.461$.
Thus, for a fixed candidate $g$, we reject $H_0:g=g_0$ when $\hat T_n(g)>0.461$; under $H_0$, the rejection probability converges to $5\%$.

The same statistic can be evaluated on a held-out sample at a prediction function fitted independently. Conditional on the realised fit, the test asks whether that particular function lies in $\cG_I$.  If it lies outside $\cG_I$, consistency implies that the rejection probability tends to one as the held-out sample grows. The present theory does not determine how often a fitted function lies in $\cG_I$, nor does it calibrate the test when the training and held-out sample sizes grow jointly. We therefore interpret rejection as evidence of a departure from invariance detectable at the given held-out sample size. 

\section{Using multiple observed environments}\label{sec:multienv}
The construction so far uses a single observational distribution, where the orthogonality condition of \Cref{ass:ortho} supplies identification. Multiple observed environments, when available, offer a second route. The single-environment check asks that the conditional MSE be constant along the prediction; we now impose it separately within each environment. A prediction function must then have constant conditional MSE along its prediction in every environment, though the level may differ. The identifying content now comes from heterogeneity across environments rather than from orthogonality. The resulting penalty does not replace the single-distribution construction; it is its environment-wise analogue.

Consider $K$ environments $\{P_e\}_{e=1}^K$ with sample sizes $n_e$ and weights
$w_e \coloneqq n_e/n$, where $n = \sum_{e=1}^K n_e$. We write
$\{(X_i^e,Y_i^e)\}_{i=1}^{n_e}$ for the observations from environment $e$. Within each environment $e$, the observations are i.i.d.\ from $P_e$.
We assume that the structural regression function $g_0$ stays the same across all
environments, while the distribution of the prediction error may vary:
\begin{assumption}[Multi-environment structural model]\label{ass:structural_multi}
There exists $g_0$ such that, for every $e \in \{1,\dots,K\}$, $Y^e \gets g_0(X^e) + \varepsilon_e$ with $\varepsilon_e \indep g_0(X^e)$, $\E[\varepsilon_e] = 0$, and $\Var(\varepsilon_e) > 0$. We also require that every $g \in \cG$, $g_0$ included, satisfies \Cref{def:Gclass} under every $P_e$, so that $g(X^e) \in L^2(P_{X,e})$, $\bar\sigma_{g,e}^2 > 0$, and the distribution function $H_{g,e}$ of $g(X^e)$ under $P_e$ is continuous.
\end{assumption}

\Cref{ass:structural_multi} extends \Cref{ass:structural} to each environment. Independence is required only between $\varepsilon_e$ and $g_0(X^e)$, not between $\varepsilon_e$ and the full vector $X^e$. The error distribution and its dependence on other coordinates of $X^e$ may therefore vary with $e$. Structural independence places $g_0$ in every class $\cG_e$ defined below.

For each environment $e$, define the conditional MSE $\sigma_{g,e}^2(t)\coloneqq\E_{P_e}[R_g\,|\, F_g=t]$.
\begin{definition}[Environment-specific invariant class]
     The environment-specific invariant class is
     \begin{align*}
         \cG_e \coloneqq \big\{g \in \cG : \sigma_{g,e}^2(t) \equiv c_{g,e}\ \text{for }P_{g,e}\text{-almost every }t,\ \text{for some constant }c_{g,e}\big\},
     \end{align*}
      where $P_{g,e}$ is the distribution of $g(X^e)$ under $P_e$. We write $\cG_\cap\coloneqq\bigcap_{e=1}^K\cG_e$.
\end{definition}

Under \Cref{ass:structural_multi}, the causal prediction function belongs to every
environment-specific invariant class. The next lemma records this fact.

\begin{lemma}[Causal prediction function is environment-wise invariant]\label{lem:envcvm_causal_in}
Under \Cref{ass:structural_multi}, $g_0 \in \cG_e$ for every $e$, and hence $g_0 \in \cG_\cap$.
\end{lemma}

\begin{figure}[!tb]
\centering
\begin{tikzpicture}[
    >={Stealth[length=2mm]},
    every node/.style={font=\small},
    var/.style={circle, draw, minimum size=8mm, inner sep=1pt, fill=gray!10},
    hid/.style={circle, draw, dashed, minimum size=8mm, inner sep=1pt},
    node distance=1.6cm and 1.4cm
]
\begin{scope}[local bounding box=single]
  \node[hid] (H) at (0, 1.2) {$H$};
  \node[var] (X) at (-1.2, 0) {$X$};
  \node[var] (Y) at (1.2, 0) {$Y$};
  \draw[->] (H) -- (X);
  \draw[->] (H) -- (Y);
  \draw[->] (X) -- (Y);
  \node[font=\bfseries\small] at (0, -0.7) {Single environment};
\end{scope}

\begin{scope}[shift={(6cm,0)}, local bounding box=multi]
  \node[hid] (He) at (0, 1.2) {$H_e$};
  \node[var] (Xe) at (-1.2, 0) {$X^e$};
  \node[var] (Ye) at (1.2, 0) {$Y^e$};
  \draw[->] (He) -- node[font=\scriptsize, left=2pt, midway] {varies with $e$} (Xe);
  \draw[->] (He) -- node[font=\scriptsize, right=2pt, midway] {varies with $e$} (Ye);
  \draw[->, thick] (Xe) -- node[font=\scriptsize, below=1pt, midway] {invariant} (Ye);
  \node[font=\bfseries\small] at (0, -0.7) {Environment $e$};
\end{scope}
\end{tikzpicture}
\caption{Hidden-confounder schematic. Left: one environment with an unobserved confounder $H$. Right: a representative environment $e$ in the multi-environment setting. Across environments, the structural relation from $X^e$ to $Y^e$ is unchanged, while the hidden mechanism and its effects on $X^e$ and $Y^e$ may vary.}
\label{fig:confounder_dag}
\end{figure}
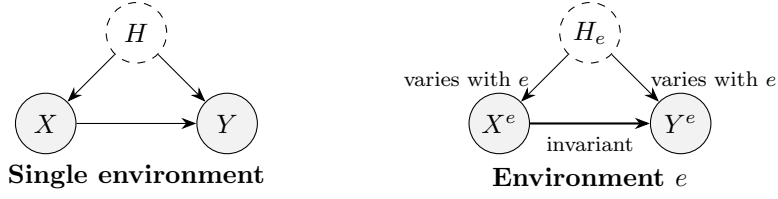

\subsection{Environment-weighted penalty and estimator}\label{sec:multienv:envcvm}

The environment-weighted CvM penalty is obtained by computing the single-environment penalty within each $P_e$ and then averaging them across environments. Within
environment $e$, let $\hat{\cT}_n^e(g)$ denote the sample CvM penalty
\eqref{eq:sample_cvm} computed on the $n_e$ observations from $P_e$, and let
$\cT^e(g)$ be its population analogue obtained by applying \eqref{eq:cvm} under
$P_e$. We then define
\begin{align*}
\cT^{\rm ew}(g) \coloneqq \sum_{e=1}^K w_e\,\cT^e(g),
\qquad
\hat{\cT}_n^{\rm ew}(g) \coloneqq \sum_{e=1}^K w_e\,\hat{\cT}_n^e(g).
\end{align*}
By \Cref{prop:completeness} applied within each environment,
\begin{align*}
\cG_e = \{ g \in \cG : \cT^e(g) = 0 \}.
\end{align*}
Since each $\cT^e$ is nonnegative, the zero set of $\cT^{\rm ew}$ is exactly $\cG_\cap = \bigcap_{e=1}^K \cG_e$.

Thus $\cT^{\rm ew}(g)=0$ means that the conditional MSE of $g$ is flat along its
prediction in every environment, although the constant level may differ across
environments. Constant conditional MSE in every environment is not the same as constant conditional MSE under the pooled mixture $P_{\mathrm{mix}} \coloneqq \sum_{e} w_e P_e$. A prediction function may satisfy either without the other. The intersection class $\cG_\cap$ and the invariant class $\cG_I$ under $P_{\mathrm{mix}}$ are therefore not nested in general.

\Cref{lem:envcvm_causal_in} shows that $g_0$ belongs to every $\cG_e$. To exclude
other prediction functions, we assume that any non-causal $g$ violates the CvM criterion in at least one environment.

\begin{assumption}[Environment heterogeneity]\label{ass:heterogeneity}
For every $g \in \cG \setminus \{g_0\}$, there exists an environment $e$ with $\cT^e(g) > 0$.
\end{assumption}

A richer class $\cG$ can contain more non-causal prediction functions, and the assumption asks that each such function have non-constant conditional MSE along its prediction in at least one environment. Therefore, the larger the class $\cG$ is, the more heterogeneity across the environments identification requires.

\Cref{ass:heterogeneity} is a separation condition. It does not follow merely because the environment distributions differ. It requires those differences to distinguish $g_0$ from every other function in $\cG$. The condition becomes explicit in the following linear Gaussian example.

\begin{example}[Linear Gaussian case]\label{ex:heterogeneity_linear}
Let $g_\beta(x)=x^\top\beta$, let $g_0=g_{\beta_0}$, and suppose $(X^e,\varepsilon_e)$ is jointly centred Gaussian in every environment, with $\Sigma_e=\Cov_{P_e}(X^e)$ and $\Var_{P_e}(X^{e\top}\beta)>0$. Then $g_\beta\in\cG_e$ if and only if
\begin{align*}
\Cov_{P_e}\big(\varepsilon_e,X^{e\top}\beta\big)
+\beta^\top\Sigma_e(\beta_0-\beta)=0.
\end{align*}
If $\varepsilon_e\indep X^e$, this reduces to $\beta^\top\Sigma_e(\beta_0-\beta)=0$. The separation assumption therefore holds exactly when the common solutions of these restrictions, within the admissible parameter class, reduce to $\beta_0$. The derivation is in \Cref{app:heterogeneity_linear}.
\end{example}
The example makes the role of heterogeneity explicit. Each environment contributes one restriction on $\beta$, and \Cref{ass:heterogeneity} requires their only common solution to be $\beta_0$. The same separation principle gives the general identification result.
\begin{theorem}[Identification by environment-weighted CvM]\label{thm:envcvm_identification}
Under \Cref{ass:structural_multi} and \Cref{ass:heterogeneity}, $\cG_\cap = \{g_0\}$. Equivalently, $g_0$ is the unique zero of $\cT^{\rm ew}$ over $\cG$.
\end{theorem}

The proofs are given in \Cref{app:envcvm_causal_in,app:envcvm_identification}. \Cref{ass:heterogeneity} plays the role here that the orthogonality condition of \Cref{ass:ortho} plays for a single distribution. It becomes more plausible as the environments become more heterogeneous, since a non-causal prediction function must retain constant conditional MSE in every environment to remain in $\cG_\cap$; when the environments are identical, $\cG_\cap$ reduces to a single $\cG_e$, and the assumption fails whenever that class contains more than $g_0$.

With $\cT^{\rm ew}(g)$ in hand, we define the multi-environment analogue of the
single-environment objective as
\begin{align*}
\cL_{\rm eCvM}(g)
\coloneqq
\sum_{e=1}^K w_e \E_{P_e}\big[(Y-g(X))^2\big] + \lambda \cT^{\rm ew}(g),
\end{align*}
at the population level, and the empirical version
\begin{align*}
\cL_{n,\rm eCvM}(g)
\coloneqq
\cL_n^{\rm ew}(g) + \lambda_n \hat{\cT}_n^{\rm ew}(g),
\end{align*}
where, writing $P_{n_e}^e$ for the empirical measure of the $n_e$ observations from environment $e$,
\begin{align*}
\cL^{\rm ew}(g) \coloneqq \sum_{e=1}^K w_e\, \E_{P_e}\big[(Y-g(X))^2\big],
\qquad
\cL_n^{\rm ew}(g) \coloneqq \sum_{e=1}^K w_e\, P_{n_e}^e\big[(Y-g(X))^2\big]
\end{align*}
are the population and empirical environment-weighted risks. At the sampling proportions $w_e = n_e/n$ the latter is the pooled average $\cL_n^{\rm ew}(g) = n^{-1}\sum_{e=1}^K \sum_{i=1}^{n_e} \big(Y_i^e-g(X_i^e)\big)^2$.
The environment-weighted estimator is the minimiser
\begin{align}\label{eq:ew_estimator}
\hat{g}_{n,\lambda_n}^{\rm ew} \coloneqq \argmin_{g \in \cG}\, \cL_{n,\rm eCvM}(g).
\end{align}

We now establish the analogues of the single-environment guarantees of \Cref{sec:asymptotics}. Because $\cT^{\rm ew}$ averages the single-environment penalties, the analysis applies within each $P_e$, and the environment-wise bounds combine by averaging.

\begin{proposition}[Uniform convergence of $\hat{\cT}_n^{\rm ew}$]\label{prop:envcvm_uniform}
Under \Cref{ass:regularity,ass:reg_rate} applied within each environment,
\begin{equation*}
\sup_{g \in \cG} \big| \hat{\cT}_n^{\rm ew}(g) - \cT^{\rm ew}(g) \big| = O_P\Big(\sum_{e=1}^K w_e\,n_e^{-1/2}\Big),
\end{equation*}
for fixed $K$, as $n_e\to\infty$ in every environment $e$.
\end{proposition}

Uniform convergence of $\hat{\cT}_n^{\rm ew}$ yields consistency by the argument of \Cref{thm:consistency_g}(a). Write $P_{X,e}$ for the covariate marginal of $P_e$ and $P_X^{\mathrm{mix}}=\sum_e w_eP_{X,e}$ for the covariate marginal of $P_{\mathrm{mix}}$. We hold the proportions $w_e$ fixed, so the target $P_X^{\mathrm{mix}}$ does not vary with $n$.

\begin{proposition}[Consistency of the environment-weighted estimator]\label{prop:envcvm_consistency}

Suppose the proportions $w_e=n_e/n$ are held fixed as $n$ grows and $\cG$ is closed in $L^2(P_X^{\mathrm{mix}})$. Under \Cref{ass:structural_multi,ass:heterogeneity}, and \Cref{ass:regularity} applied within each environment, if $\lambda_n\to\infty$ and $\lambda_n=o(\min_e n_e)$, then the estimator $\hat g_{n,\lambda_n}^{\rm ew}$ in \eqref{eq:ew_estimator} converges to $g_0$ in $L^2(P_X^{\mathrm{mix}})$, and hence in $L^2(P_{X,e})$ for every environment $e$.
\end{proposition}

The proofs are in \Cref{app:envcvm_uniform,app:envcvm_consistency}.
The mixture closedness condition is stated separately because closedness in every environment does not imply closedness under the mixture norm.
An explicit rate for the environment-weighted estimator would use \Cref{ass:reg_rate} within each environment and additionally require an environment-weighted analogue of \Cref{ass:curvature}, which is not imposed here.

\subsection{Distributional-robustness interpretation of $\cT^{\rm ew}$}\label{sec:multienv:dro}
The environment-weighted penalty has a direct distributional-robustness interpretation, in which each environment contributes its own worst-case perturbation and the total objective averages these environment-wise robust risks.

For each environment $e$, let $\mathcal{H}_e$ denote the perturbation class defined
as in \Cref{sec:dro}, but under $P_e$, and let $Q_h^e$ be the corresponding
normalised perturbation measure. Then \Cref{thm:dro} applied within environment $e$
gives
\begin{align*}
\E_{P_e}\big[(Y-g(X))^2\big] + \lambda\,\cT^e(g)
=
\sup_{h \in \mathcal{H}_e}
\bigg\{\E_{Q_h^e}\big[(Y-g(X))^2\big]-\frac{\|h\|_{\mathcal{H}_e}^2}{4\lambda}
\bigg\}.
\end{align*}
Averaging over environments yields
\begin{align}\label{eq:envcvm_dro_combined}
\cL_{\rm eCvM}(g)
=
\sum_{e=1}^K w_e
\sup_{h_e \in \mathcal{H}_e}
\bigg\{
\E_{Q_{h_e}^e}\big[(Y-g(X))^2\big]
-
\frac{\|h_e\|_{\mathcal{H}_e}^2}{4\lambda}
\bigg\}.
\end{align}
Thus $\cL_{\rm eCvM}(g)$ is a weighted sum of environment-wise worst-case risks.
The adversary chooses a separate perturbation in each environment, and the cost of
perturbation is measured by the corresponding RKHS norm.

\subsection{Combining DRIG and environment-specific CvM}
The DRIG objective and the eCvM penalty can be combined. Take the same environment weights $\omega_e=w_e=n_e/n$ in both terms and define
\begin{align*}
\cL_\lambda(g;\gamma)
&\coloneqq \cL_{\mathrm{DRIG}}(g;\gamma)+\lambda\cT^{\rm ew}(g),
&g_{\gamma,\lambda}
&\coloneqq\argmin_{g\in\cG}\cL_\lambda(g;\gamma).
\end{align*}
Its empirical version is
\begin{align}\label{eq:combined_estimator}
\hat g_{n,\gamma,\lambda_n}
\coloneqq
\argmin_{g\in\cG}
\big\{\hat\cL_{\mathrm{DRIG},n}(g;\gamma)
+\lambda_n\hat\cT_n^{\rm ew}(g)\big\}.
\end{align}
This definition does not require $g$ to be linear. 

The two robustness certificates require separate attention. The eCvM identity \eqref{eq:envcvm_dro_combined} holds for general $g$ and does not require a structural model. The finite-radius DRIG identity \eqref{eq:drig_dro} is proved under the linear structural model in \Cref{sec:related}. Under that model, the covariance ordering $\sum_e\omega_e(S_e-S_{e_0})\succeq0$ implies that the DRIG objective is at least the pooled risk for a linear prediction function. This implication uses that the prediction error of $g_\beta(x)=x^\top\beta$ is a linear form of the structural noise. For a nonlinear $g$, ordering the noise covariance matrices does not in general order the corresponding squared prediction errors. Thus nonlinear DRIG remains a well-defined training objective, but \eqref{eq:max_domination} would require a different robustness argument.

\begin{proposition}[Joint robustness bound]\label{prop:joint_bound}
Under the linear DRIG model of \Cref{sec:related}, for $g_\beta(x)=x^\top\beta$, $\gamma\geq1$, and $\lambda\geq0$,
\begin{align}\label{eq:max_domination}
\cL_\lambda(g_\beta;\gamma)
\geq
\max\Big\{
\cL_{\mathrm{DRIG}}(g_\beta;\gamma),
\E_{P_{\mathrm{mix}}}[R_{g_\beta}]+\lambda\cT^{\rm ew}(g_\beta)
\Big\}.
\end{align}
\end{proposition}
The first term in \eqref{eq:max_domination} is the DRIG worst-case risk over $\cB_\gamma$; the second is the eCvM certificate in \eqref{eq:envcvm_dro_combined}. The two shift classes are complementary in kind. A perturbation in $\cB_\gamma$ acts on the structural equations and can alter $P(Y\,|\,X)$, whereas the reweightings underlying $\cT^{\rm ew}$ change the covariate factor alone. The combined objective therefore bounds the DRIG worst-case risk and the regularised eCvM criterion, while neither penalty alone certifies the other class. The proof is in \Cref{app:joint_bound}.

We now fix $\gamma$ and let the penalty parameter $\lambda$ grow. The penalty drives the
solution into the invariant class $\cG_\cap$, and the DRIG term then selects the
most robust prediction function within it. This selection matters most when invariance
leaves $\cG_\cap$ larger than $\{g_0\}$ and the environment interventions remain
strong enough for DRIG to distinguish the functions within it.

\begin{proposition}[Convergence to the DRIG-best invariant prediction function]\label{prop:drig_envcvm_limit}
Fix $\gamma \ge 0$ and assume $\cG_\cap \ne \emptyset$. Suppose that $\cG$ is compact in $L^2(P_X^{\mathrm{mix}})$, and that $\cL_{\mathrm{DRIG}}(\cdot; \gamma)$ and $\cT^{\rm ew}$ are continuous on $\cG$ in $L^2(P_X^{\mathrm{mix}})$. Suppose further that
\begin{align*}
g_\gamma^\ast
\coloneqq
\argmin_{g \in \cG_\cap} \cL_{\mathrm{DRIG}}(g; \gamma)
\end{align*}
is unique. Then
\begin{align*}
g_{\gamma,\lambda} \to g_\gamma^\ast
\qquad \text{in } L^2(P_X^{\mathrm{mix}}),
\qquad \text{as } \lambda \to \infty.
\end{align*}
\end{proposition}

\Cref{lem:compact,lem:continuity} give sufficient conditions for compactness of $\cG$ and continuity of $\cT^{\rm ew}$. Continuity of $\cL_{\mathrm{DRIG}}(\cdot;\gamma)$ follows from \eqref{eq:drig_explicit}. Under \Cref{ass:heterogeneity}, $\cG_\cap=\{g_0\}$ and uniqueness is automatic. Otherwise, uniqueness must be assumed. If the DRIG risk has several minimisers over $\cG_\cap$, different sequences $\lambda_k\to\infty$ may yield different limits. The proposition is useful when invariance leaves several prediction functions and the DRIG risk selects among them.

\section{Experiments}\label{sec:experiments}
We evaluate the proposed estimators on synthetic and real-data benchmarks. Each experimental section specifies whether the exact criterion \eqref{eq:sample_cvm} or its smooth approximation \eqref{eq:anneal} is used. In the linear and nonlinear designs, we trace how the single-distribution estimator \eqref{eq:finite_estimator} moves from ERM towards $g_0$ as the penalty varies. A third simulation introduces an observed anchor variable and compares CvM and eCvM with anchor regression, DRIG, and the combined estimator. The final experiment uses RPE1 Perturb-seq interventions to study how performance changes with the number of training environments.

\Cref{fig:dgps} shows the training mechanisms and representative changes to the child mechanisms. The simulations evaluate recovery of $g_0$ when descendant information is predictive under the training distribution. In the test distributions, we alter only the child mechanism, leaving the structural equation for $Y$ unchanged. Hence the risk of $g_0$ remains stable, whereas predictions that use descendant information may fail.

\begin{figure}[!tbp]
\centering
\begin{minipage}{\textwidth}
\centering

\makebox[\textwidth][c]{%
\begin{tikzpicture}[
    >={Stealth[length=2mm]},
    every node/.style={font=\small},
    var/.style={circle, draw, minimum size=7mm, inner sep=1pt},
    obs/.style={circle, draw, minimum size=7mm, inner sep=1pt, fill=gray!10},
    node distance=1.6cm and 1.4cm
]
%% ----- Linear design -----
\begin{scope}[local bounding box=linscope]
  \node[var] (Z)  {$Z$};
  \node[var, right=of Z] (Y) {$Y$};
  \node[obs, right=of Y] (X2) {$X_2$};
  \node[obs, below=0.45cm of Z] (X1) {$X_1$};
  \node[above=0.05cm of Y, font=\scriptsize] {$Y = 2Z + \varepsilon$};
  \node[above=0.05cm of X2, font=\scriptsize] {$X_2 = \frac{1}{2}Y + \eta$};
  \node[left=0.05cm of X1, font=\scriptsize] {$X_1 = Z$};
  \draw[->] (Z) -- (Y);
  \draw[->] (Y) -- (X2);
  \draw[->] (Z) -- (X1);
  \node[above=0.4cm of Y, font=\bfseries\small] {Linear design (training)};
\end{scope}

%% ----- Linear design shifted -----
\begin{scope}[shift={(8.5cm,0)}, local bounding box=linshiftscope]
  \node[var] (Z2)  {$Z$};
  \node[var, right=of Z2] (Y2) {$Y$};
  \node[obs, right=of Y2, draw=red, very thick] (X22) {$X_2$};
  \node[obs, below=0.45cm of Z2] (X12) {$X_1$};
  \node[above=0.05cm of Y2, font=\scriptsize] {$Y = 2Z + \varepsilon$};
  \node[below=0.05cm of X22, font=\scriptsize, color=red] {$X_2 = -\frac{1}{2}Y + \eta$};
  \node[left=0.05cm of X12, font=\scriptsize] {$X_1 = Z$};
  \draw[->] (Z2) -- (Y2);
  \draw[->, color=red, very thick] (Y2) -- (X22);
  \draw[->] (Z2) -- (X12);
  \node[above=0.4cm of Y2, font=\bfseries\small] {Linear design (test)};
\end{scope}
\end{tikzpicture}%
}

\vspace{0.9cm}

\makebox[\textwidth][c]{%
\begin{tikzpicture}[
    >={Stealth[length=2mm]},
    every node/.style={font=\small},
    var/.style={circle, draw, minimum size=7mm, inner sep=1pt},
    obs/.style={circle, draw, minimum size=7mm, inner sep=1pt, fill=gray!10},
    node distance=1.6cm and 1.4cm
]
%% ----- Nonlinear design -----
\begin{scope}[local bounding box=nlscope]
  \node[var] (Zp)  {$Z_{pa}$};
  \node[var, right=of Zp] (Y3) {$Y$};
  \node[var, right=of Y3] (Zc) {$Z_{ch}$};
  \node[obs, below=0.45cm of Y3] (Xn) {$X$};
  \node[above=0.05cm of Y3, font=\scriptsize] {$Y = \operatorname{log}\big(1+\operatorname{exp}(Z_{pa})\big) + \varepsilon_Y$};
  \node[below=0.05cm of Zc, font=\scriptsize] {$Z_{ch} = Y + \varepsilon_{ch}$};
  \node[below=0.05cm of Xn, font=\scriptsize] {$X = (Z_{pa}+Z_{ch},\, Z_{pa}-Z_{ch})$};
  \draw[->] (Zp) -- (Y3);
  \draw[->] (Y3) -- (Zc);
  \draw[->] (Zp) -- (Xn);
  \draw[->] (Zc) -- (Xn);
  \node[above=0.4cm of Y3, font=\bfseries\small] {Nonlinear design (training)};
\end{scope}

%% ----- Nonlinear design shifted -----
\begin{scope}[shift={(8.5cm,0)}, local bounding box=nlshiftscope]
  \node[var] (Zp2)  {$Z_{pa}$};
  \node[var, right=of Zp2] (Y4) {$Y$};
  \node[var, right=of Y4, draw=red, very thick] (Zc2) {$Z_{ch}$};
  \node[obs, below=0.45cm of Y4] (Xn2) {$X$};
  \node[above=0.05cm of Y4, font=\scriptsize] {$Y = \operatorname{log}\big(1+\operatorname{exp}(Z_{pa})\big) + \varepsilon_Y$};
  \node[below=0.05cm of Zc2, font=\scriptsize, color=red] {$Z_{ch} = -Y + \varepsilon_{ch}$};
  \node[below=0.05cm of Xn2, font=\scriptsize] {$X = (Z_{pa}+Z_{ch},\, Z_{pa}-Z_{ch})$};
  \draw[->] (Zp2) -- (Y4);
  \draw[->, color=red, very thick] (Y4) -- (Zc2);
  \draw[->] (Zp2) -- (Xn2);
  \draw[->] (Zc2) -- (Xn2);
  \node[above=0.4cm of Y4, font=\bfseries\small] {Nonlinear design (test)};
\end{scope}
\end{tikzpicture}%
}

\end{minipage}
\caption{Training mechanisms (left) and representative shifts in the child mechanisms (right) for the linear and nonlinear designs. Red marks the altered child mechanism; the structural equation for $Y$ is unchanged.}
\label{fig:dgps}
\end{figure}
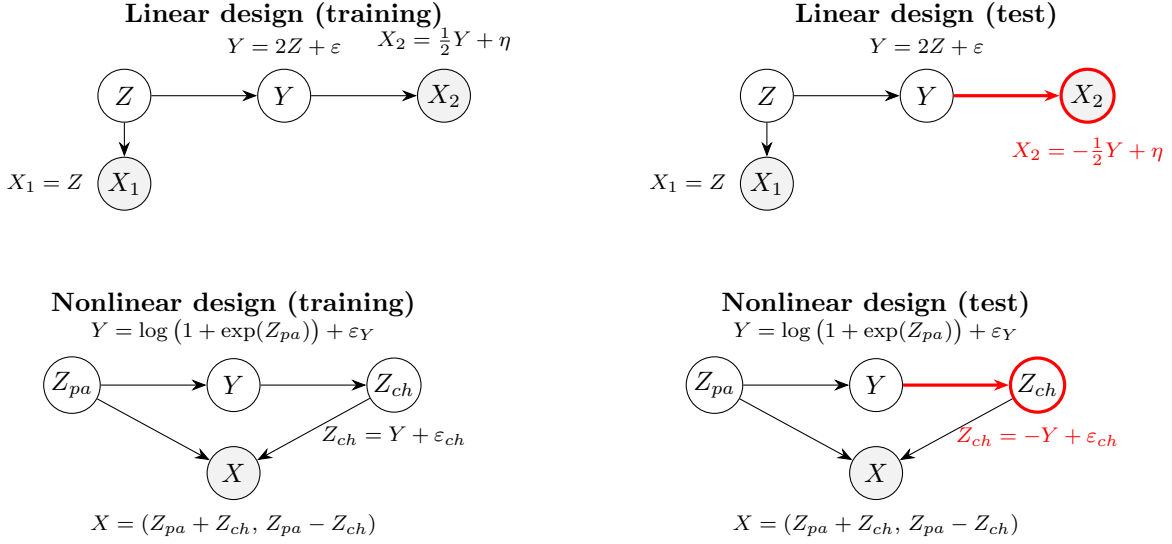

\subsection{\texorpdfstring{Linear design}{Linear design}}\label{sec:linear}
The linear example is the simplest setting in which the method can distinguish the causal parent from a spurious child variable. For the training data, we generate
\begin{align*}
Z \sim \mathcal{N}(0,1), \qquad
Y = 2Z + \varepsilon_Y, \qquad
X_1 = Z, \qquad
X_2 = 0.5\,Y + \varepsilon_{ch},
\end{align*}
where
\begin{align*}
\varepsilon_Y \sim 0.5\,\mathcal{N}(-2,\,0.1^2) + 0.5\,\mathcal{N}(2,\,0.1^2),
\qquad
\varepsilon_{ch} \sim 0.5\,\mathcal{N}(-1.5,\,0.1^2) + 0.5\,\mathcal{N}(1.5,\,0.1^2).
\end{align*}
The non-Gaussian mixtures are deliberate. Under joint Gaussianity, the linear least-squares residual would be independent of the fitted value, so the CvM penalty could not separate that fit from $g_0$ when they differ; see \Cref{rem:gaussian_linear}.
The causal target is $g_0(X)=2X_1$. To probe robustness, we shift only the
child mechanism in the test set:
\begin{enumerate}
    \item Sign flip: $X_2 \gets -\frac{1}{2}Y + \varepsilon_{ch}$;
    \item Tripled child noise: $X_2 \gets Y + 3\,\varepsilon_{ch}$;
    \item Strong sign flip: $X_2 \gets -Y + \varepsilon_{ch}$.
\end{enumerate}
The second shift also doubles the coefficient of $Y$ in the child mechanism, from $1/2$ to $1$.
The training mechanism and one representative test mechanism are shown in \Cref{fig:dgps}.

We use the linear function class
\begin{align*}
g_\beta(X) = \beta_1 X_1 + \beta_2 X_2.
\end{align*}
Since the parameter space is two-dimensional, we minimise \eqref{eq:finite_estimator} by exhaustive search over a $101 \times 101$ grid on $[-1,4]^2$, using the exact penalty in \eqref{eq:sample_cvm}. We omit $\beta=(0,0)$, the only grid point with a constant prediction and hence outside $\cG$.

\begin{figure}[!tb]
\centering
\includegraphics[width=\textwidth]{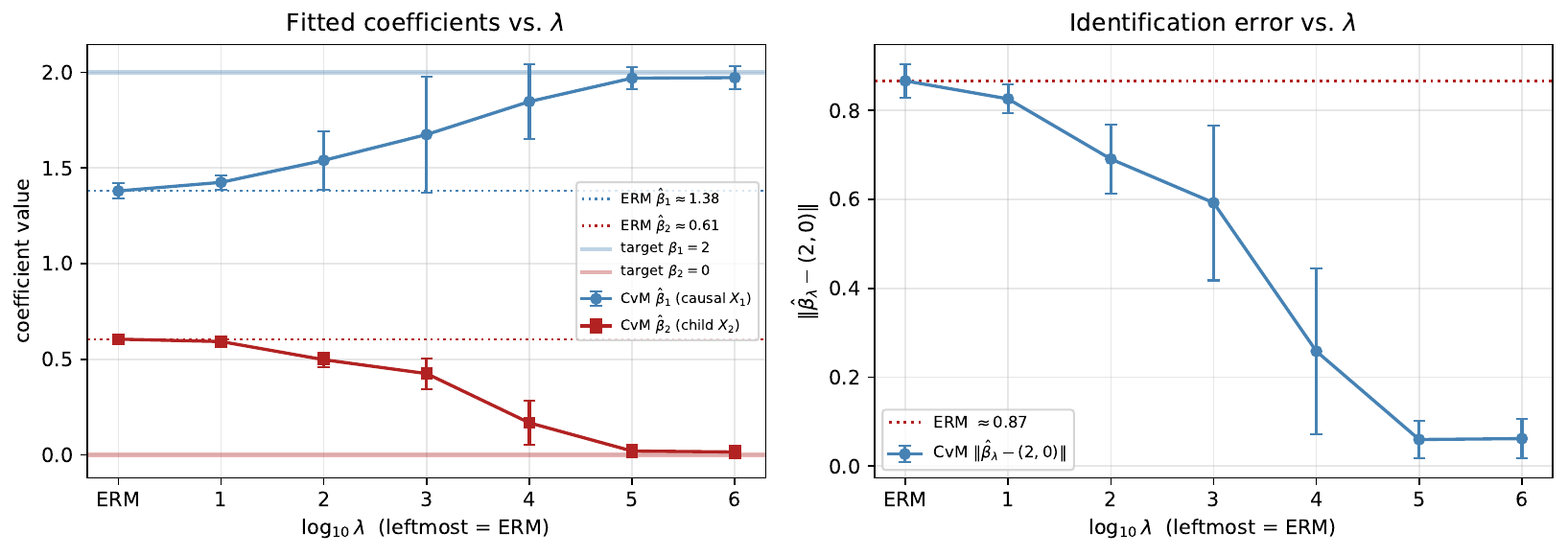}
\caption{Linear design. Fitted coefficients (left) and their distance from $(2,0)$ (right) as $\lambda$ varies. The CvM estimate moves from ERM towards $g_0(X)=2X_1$.}
\label{fig:exp1-weights}
\end{figure}

\begin{table}[!tb]
\centering
\begin{tabular}{@{}lcccc@{}}
\toprule
 Estimator & Non-shift & Sign flip & Tripled child noise & Strong sign flip\\
\midrule
ERM ($\lambda=0$) & $\mathbf{2.86}$ & $9.17$ & $8.93$ & $13.86$\\
Proposed ($\lambda=10^6$) & $4.01$ & $\mathbf{4.03}$ & $\mathbf{4.02}$& $\mathbf{4.01}$\\
Gain over ERM & $-40.6\%$ & $+56.1\%$ & $+55.0\%$ & $+71.1\%$\\
\midrule
causal target $g_0$ & $\approx 4.01$ & $\approx 4.01$ & $\approx 4.01$ & $\approx 4.01$\\
\bottomrule
\end{tabular}
\caption{Linear design. Mean test MSE over $20$ replications under the training mechanism and three shifts in the child mechanism, with $2000$ training and $1000$ test observations.}
\label{tab:exp1-robust}
\end{table}

ERM assigns nonzero weight to $X_2$ and attains MSE $2.86$ under the training mechanism, but its MSE lies between $8.93$ and $13.86$ after the shifts. At $\lambda=10^6$, CvM nearly recovers $(2,0)$ and has MSE about $4.01$ in every setting. Thus the penalty trades fit under the training mechanism for stability. Although each test environment admits a lower-risk linear fit that uses $X_2$, $g_0$ uniquely minimises worst-case linear risk across the three shifts (\Cref{app:linear_minimax}).

\subsection{\texorpdfstring{Nonlinear design}{Nonlinear design}}\label{sec:nonlinear}
We use the nonlinear design to study recovery of a nonlinear causal prediction function within a flexible class, beyond selecting the causal variable in a linear model. We generate
\begin{align*}
Z_{pa}\sim \mathrm{Uniform}[-2,2], \qquad
\varepsilon_Y,\, \varepsilon_{ch} \overset{\text{i.i.d.}}{\sim} \mathrm{Exp}(1)-1,
\end{align*}
and set
\begin{align*}
Y = \operatorname{log}\big(1+\operatorname{exp}(Z_{pa})\big) + \varepsilon_Y,
\qquad
Z_{ch} = Y + \varepsilon_{ch},
\qquad
X = (Z_{pa}+Z_{ch},\, Z_{pa}-Z_{ch}),
\end{align*}
so both coordinates mix causal-parent and child information. The causal prediction function
is therefore
\begin{align*}
g_0(X) = \operatorname{log}\big(1+\operatorname{exp}\big((X_1+X_2)/2\big)\big).
\end{align*}
We consider the following shifts in the child mechanism:
\begin{enumerate}
    \item Sign flip: $Z_{ch} \gets -Y + \varepsilon_{ch}$;
    \item Tripled child noise: $Z_{ch} \gets Y + 3\,\varepsilon_{ch}$;
    \item Strong sign flip: $Z_{ch} \gets -2Y + \varepsilon_{ch}$.
\end{enumerate}
The nonlinear function class is implemented by a two-layer fully connected MLP. We use the smooth approximation in \eqref{eq:anneal}; further training details are given in \Cref{app:nonlinear}. \Cref{tab:nonlinear} reports results over $50$ replications, with training and test sample sizes of $2000$ each.

\begin{table}[!tb]
\centering
\setlength{\tabcolsep}{1.8pt}
\begin{tabular}{@{}lccccccc@{}}
\toprule
& \multicolumn{7}{c}{Penalty parameter $\lambda$}\\
\cmidrule(l){2-8}
& $0$ (ERM) & $10$ & $100$ & $500$ & $1{,}000$ & $5{,}000$ & $10{,}000$ \\
\midrule
$\|\hat g_\lambda - g_0\|_{P_X}$ & $0.70\,(0.01)$ & $0.60\,(0.01)$ & $0.41\,(0.02)$ & $0.41\,(0.04)$ & $0.48\,(0.04)$ & $0.67\,(0.05)$ & $0.80\,(0.04)$ \\
\midrule
no shift        & $0.50\,(0.00)$ & $0.55\,(0.01)$ & $0.84\,(0.03)$ & $1.07\,(0.05)$ & $1.17\,(0.04)$ & $1.46\,(0.06)$ & $1.67\,(0.06)$ \\
sign flip           & $3.06\,(0.06)$ & $2.31\,(0.10)$ & $2.39\,(0.10)$ & $1.95\,(0.12)$ & $1.99\,(0.16)$ & $1.84\,(0.12)$ & $2.05\,(0.22)$ \\
tripled child noise & $2.41\,(0.05)$ & $2.05\,(0.06)$ & $1.19\,(0.03)$ & $1.24\,(0.06)$ & $1.34\,(0.09)$ & $1.58\,(0.07)$ & $1.87\,(0.10)$ \\
strong sign flip    & $5.01\,(0.18)$ & $3.34\,(0.25)$ & $4.08\,(0.28)$ & $3.01\,(0.30)$ & $2.97\,(0.37)$ & $2.30\,(0.23)$ & $2.64\,(0.43)$ \\
\bottomrule
\end{tabular}
\caption{Nonlinear design. Means of the $L^2(P_X)$ error relative to $g_0$ and test MSE over $50$ replications, with standard errors in parentheses. The training and test sample sizes are both $2000$. $\lambda=0$ is ERM. The causal prediction function $g_0$ has MSE approximately $1$ under each test distribution.}
\label{tab:nonlinear}
\end{table}

All positive values of $\lambda$ considered in the experiment give smaller MSE than ERM under each of the three distribution shifts. This improvement comes at the cost of larger MSE when there is no shift in the test distribution. Moderate penalties, $\lambda=100$ or $500$, bring $\hat g_\lambda$ closest to $g_0$ in $L^2(P_X)$. As $\lambda$ increases beyond these values, the $L^2(P_X)$ error grows. However, the maximum MSE across the three shifts is smallest at $\lambda=5000$. Thus the penalty value that gives the closest estimate of $g_0$ differs from the one that gives the best performance across these test distributions.

One possible explanation for the increase in $L^2(P_X)$ error at larger $\lambda$ is that the empirical penalty can be nonzero even at $g_0$. Because of sampling variation, another function may have a smaller penalty on the training sample without satisfying population invariance. When $\lambda$ is large, even a small reduction in the empirical penalty can outweigh an increase in training MSE. The objective may therefore favour a function farther from $g_0$ because its empirical penalty is smaller. Approximation error from the neural network and numerical optimisation may also contribute to selecting a function with larger $L^2(P_X)$ error relative to $g_0$. The bound in \eqref{eq:rate_main} suggests an asymptotic order for $\lambda_n$, but not a data-driven finite-sample choice. We return to this issue in \Cref{sec:discussion}.

\subsection{Comparison with anchor regression and DRIG}\label{sec:comparison_linear}
We compare the proposed estimators with anchor regression and DRIG on a version of the linear design in \Cref{sec:linear} with an anchor variable. We add an exogenous anchor variable $A\in\{-1,0,1\}$, drawn uniformly and independently of $(Z,\varepsilon_Y,\varepsilon_{ch})$, to the child mechanism through
\begin{align*}
X_2 = \frac{1}{2}Y + \varepsilon_{ch} + \delta A,
\qquad \delta = 0.3.
\end{align*}
The target remains $g_0(X)=2X_1$, with $\beta_0=(2,0)$, and test data follow the three shifts of \Cref{sec:linear} with $A=0$. All methods are linear. The CvM-based methods use the exact criterion in \eqref{eq:sample_cvm}. We measure coefficient error by the Euclidean distance $\|\hat\beta-\beta_0\|_2$. \Cref{tab:linear_comparison} reports means over $100$ replications. Each method is evaluated at the hyperparameter minimising its mean worst-case test MSE, so the comparison is an oracle envelope over the grid rather than a deployable tuning rule. Further details are in \Cref{sec:exp_details}.

\begin{table}[!tb]
\centering
\begin{tabular}{@{}lcc@{}}
\toprule
Method & Worst-case MSE & $\|\hat\beta - \beta_0\|_2$ \\
\midrule
ERM & $13.85\ (0.02)$ & $0.856\ (0.003)$ \\
DRIG & $12.68\ (0.13)$ & $0.802\ (0.016)$ \\
Anchor regression & $6.56\ (0.14)$ & $0.370\ (0.016)$ \\
CvM & $4.22\ (0.03)$ & $0.053\ (0.005)$ \\
eCvM  & $\mathbf{4.11\ (0.02)}$ & $\mathbf{0.035\ (0.004)}$ \\
DRIG + eCvM  & $\mathbf{4.11\ (0.02)}$ & $\mathbf{0.035\ (0.004)}$ \\
\midrule
Causal target $g_0$ & $\approx 4.01$ & $0$ \\
\bottomrule
\end{tabular}
\caption{Linear design with an anchor variable. Mean worst-case test MSE and coefficient error over $100$ replications, with standard errors in parentheses and oracle grid selection. CvM does not use the anchor variable; eCvM uses its values. Column minima are bold.}
\label{tab:linear_comparison}
\end{table}

All three CvM-based estimators lie close to the causal target, whereas anchor regression and DRIG retain substantially larger coefficient errors. CvM achieves this without using the anchor variable, and eCvM improves slightly by using it. The anchor variable directly perturbs the unstable child $X_2$, but only weakly, with $\delta=0.3$. Consequently, anchor regression and DRIG receive only a small between-environment signal that $X_2$ should be downweighted. Here the DRIG term adds no further gain, so the combined estimator coincides with eCvM.

\subsection{Single-cell intervention data}\label{sec:singlecell}
The RPE1 Perturb-seq data of \citet{replogle2022mapping} record gene expression in control cells and cells after CRISPRi knockdown, which suppresses the expression of a targeted gene. Following the preprocessing of \citet{chevalley2025causalbench} and the prediction setup of \citet{shen2026causality}, we use the expression of EIF1 as $Y$ and the expression levels of the other nine selected genes as $X$. The training sample contains $11{,}485$ observational control cells, in which no gene is targeted for knockdown, and $1{,}865$ cells from ten intervention environments, each targeting one of these ten genes, giving $13{,}350$ cells in total. As in \citet{shen2026causality}, the test set comprises $50$ held-out single-gene knockdown environments targeting genes outside this ten-gene set, selected by the largest MSE of OLS fitted only to the observational cells. All methods use linear prediction functions $g_\beta(x)=x^\top\beta$. For the full-training comparison, the CvM-based methods use the smooth approximation in \eqref{eq:anneal} to facilitate gradient-based optimisation.

In \Cref{fig:singlecell_full}, we compare pooled OLS, anchor regression, DRIG, CvM, eCvM, and the combined DRIG and eCvM estimator in \eqref{eq:combined_estimator}. The combined estimator has the lowest median MSE, $0.317$, and the lowest worst-case MSE, $0.681$. Gene identifiers, preprocessing and hyperparameter details are given in \Cref{app:singlecell}.

\begin{figure}[!tb]
\centering
\includegraphics[width=0.7\linewidth]{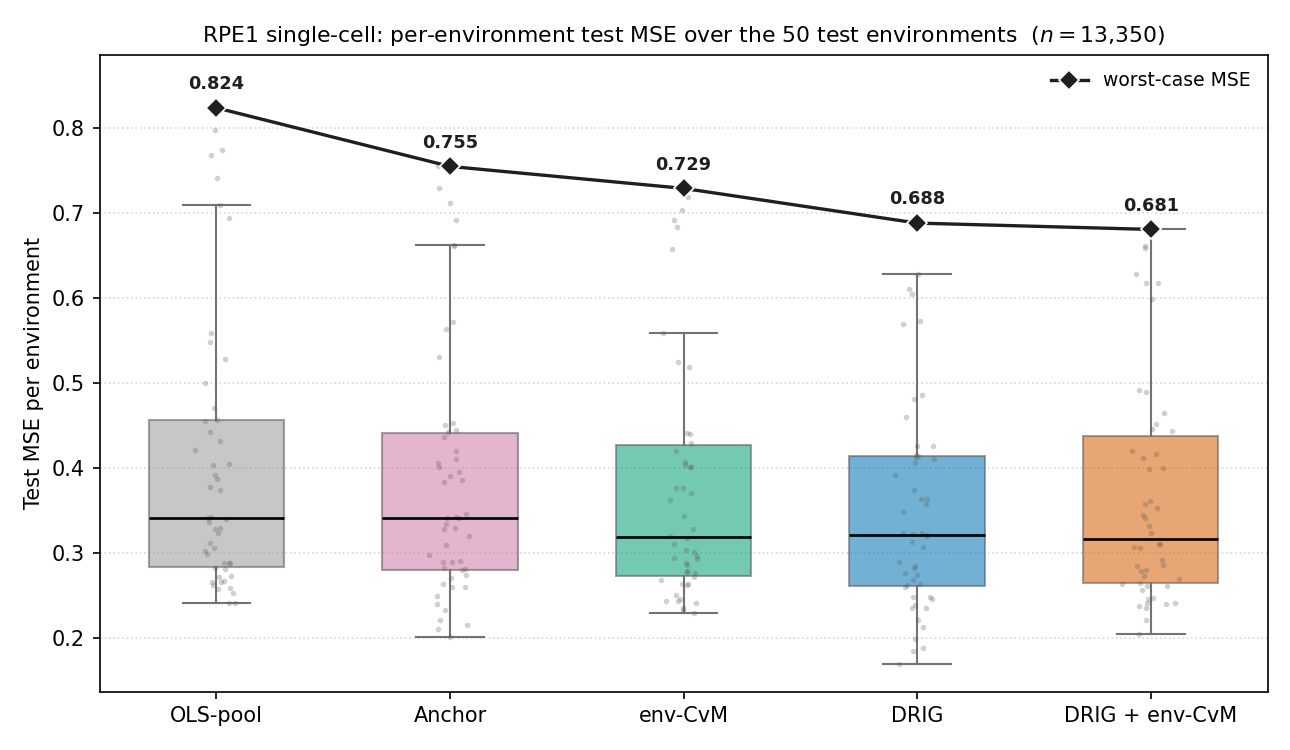}
\caption{Test MSE across $50$ hidden-gene knockdown environments in the RPE1 Perturb-seq data. All methods are trained on the full training sample.}
\label{fig:singlecell_full}
\end{figure}

We assess the value of observing only a few, and then all, intervention environments in the training set by varying $K\in\{0,1,2,3,5,7,10\}$. For each $K$, we combine a random subset of the ten environments with the observational cells and subsample to approximately $5000$ observations, using the same rows for every method. \Cref{tab:singlecell_few_envs} reports mean worst-case MSE over $20$ draws and the oracle minimum over each method's grid. Sampling and tuning details are in \Cref{sec:exp_details}.

\begin{table}[!tb]
\centering\small
\setlength{\tabcolsep}{5pt}
\begin{tabular}{@{}rccccc@{}}
\toprule
$K$ & OLS-pool & Anchor & DRIG & CvM / eCvM & DRIG + eCvM \\
\midrule
$0$  & $1.113\ (0.015)$ & --- & --- & $0.741\ (0.006)^{\dagger}$ & --- \\
$1$  & $1.113\ (0.020)$ & $1.042\ (0.037)$ & $0.993\ (0.048)$ & $0.731\ (0.002)$ & $\mathbf{0.721\ (0.007)}$ \\
$2$  & $1.034\ (0.021)$ & $0.892\ (0.032)$ & $0.828\ (0.042)$ & $0.729\ (0.002)$ & $\mathbf{0.702\ (0.011)}$ \\
$3$  & $1.033\ (0.026)$ & $0.875\ (0.034)$ & $0.776\ (0.043)$ & $0.730\ (0.004)$ & $\mathbf{0.678\ (0.018)}$ \\
$5$  & $0.939\ (0.025)$ & $0.784\ (0.013)$ & $0.705\ (0.022)$ & $0.731\ (0.002)$ & $\mathbf{0.671\ (0.015)}$ \\
$7$  & $0.887\ (0.012)$ & $0.776\ (0.014)$ & $0.662\ (0.023)$ & $0.729\ (0.002)$ & $\mathbf{0.645\ (0.019)}$ \\
$10$ & $0.821\ (0.007)$ & $0.755\ (0.000)$ & $0.688\ (0.000)$ & $0.729\ (0.000)$ & $\mathbf{0.677\ (0.010)}$ \\
\bottomrule
\addlinespace[2pt]
\multicolumn{6}{@{}p{0.93\linewidth}@{}}{\footnotesize $^{\dagger}$ At $K=0$, anchor regression is unavailable, DRIG reduces to pooled OLS, and eCvM and the combined estimator reduce to single-distribution CvM.}
\end{tabular}
\caption{RPE1 Perturb-seq. Mean worst-case test MSE over $20$ draws, with standard errors in parentheses, as the number $K$ of training intervention environments varies. Values are oracle minima over each method's grid; row minima are bold.}
\label{tab:singlecell_few_envs}
\end{table}

At $K=0$, CvM lowers worst-case MSE from $1.113$ to $0.741$, although its oracle-selected $\lambda$ lies at the top of the grid in every draw, leaving data-driven tuning from observational data unresolved. At $K=1$, eCvM attains worst-case MSE $0.731$, compared with $0.993$ for DRIG. The combined estimator improves on eCvM for every $K\geq1$. Thus eCvM supplies most of the gain with few labelled environments, while DRIG adds value as more environments are observed.

\section{Discussion}\label{sec:discussion}

Causal prediction functions are attractive under distribution shift because their structural mechanism is invariant. Under \Cref{ass:structural}, independence of the noise and the causal prediction makes the conditional MSE given $g_0(X)$ constant. Reweighting a single observational distribution through the prediction probes this restriction and generates synthetic environments. The resulting criterion has an exact worst-case representation, identifies the causal prediction function under an orthogonality condition, and yields a goodness-of-fit test.

The identifying orthogonality condition is not testable in general, and our proposed goodness-of-fit test checks invariance, not orthogonality (\Cref{rem:gof_scope}). Under misspecification, the population target is the best invariant fit. Inverting the distribution-free calibration gives the set of candidate functions not rejected as $g_0$, $C_n = \{g \in \cG : \hat{T}_n(g) \le c_\alpha\}$, which satisfies $\Pr(g_0 \in C_n) \to 1-\alpha$ under the conditions of \Cref{cor:null_g0}, as in an Anderson--Rubin inversion \citep{anderson1949estimation}. Other members of $\cG_I$ can remain in $C_n$, so its diameter need not reflect the strength of identification. For a flexible class $\cG$, the region is conceptual rather than computational.

For a fixed prediction function $g$, the robustness guarantee in \Cref{thm:dro} covers covariate reweightings that preserve $P(Y\,|\,X)$ and $P(Y\,|\,g(X))$; see also \Cref{prop:reweight}. The DRIG term instead covers structural perturbations that may alter $P(Y\,|\,X)$. By \Cref{prop:joint_bound}, the combined objective bounds the DRIG worst-case risk and the regularised eCvM criterion.

Several directions remain open. The goodness-of-fit theory is complete for a prediction function held fixed relative to the evaluation sample. Extending the calibration to a fitted prediction function would require controlling its estimation error relative to the evaluation sample size and remains open. The finite-sample choice of $\lambda$ also remains open. In the robustness representation of \Cref{thm:dro}, $\lambda$ determines how strongly large reweightings are penalised, but we do not yet have a data-driven rule for choosing this strength. This is the counterpart of choosing the radius of the uncertainty set in standard DRO. Furthermore, the convergence rates in \Cref{sec:asymptotics} hold for a generic function class. Sharper results for specific classes, including the neural networks used in the experiments, require separate arguments. Finally, the reweighting may depend on summaries of the covariates other than the prediction (\Cref{prop:reweight}). How to choose such a summary for a given family of shifts remains open.

\section*{Acknowledgements}
We thank David Blei and Zijian Guo for helpful discussions. LY thanks Robin J.~Evans for his invaluable mentorship and support, and acknowledges funding from the EPSRC Centre for Doctoral Training in Modern Statistics and Statistical Machine Learning (EP/S023151/1). CS acknowledges funding from the Swiss National Science Foundation (grant no.~214865).

\bibliography{paper-ref}
\appendix
\section{Auxiliary results}\label{sec:aux}
This section collects auxiliary results that support the main theorems.
\begin{lemma}[Permanence of pseudo-dimension]\label{lem:pdim}
Suppose \Cref{ass:regularity}(a) is satisfied, then:
\begin{enumerate}[label=(\alph*)]
\item $\cF_1\coloneqq\Big\{(x,y) \mapsto \big(y - g(x)\big)^2\mathbbm{1}\{g(x) \le s\} : g \in \cG,\, s \in \R\Big\}$ has pseudo-dimension $O(V)$;
\item $\cF_2\coloneqq \Big\{(x,y) \mapsto \big(y - g(x)\big)^2 : g \in \cG\Big\}$ has pseudo-dimension $O(V)$;
\item $\cF_3\coloneqq \big\{x \mapsto \mathbbm{1}\{g(x) \le s\} : g \in \cG,\, s \in \R\big\}$ has VC dimension at most $\mathrm{Pdim}(\cG_c)$, where $\cG_c \coloneqq \{g - s : g \in \cG,\, s \in \R\}$, hence $O(V)$ under \Cref{ass:regularity}(a).
\end{enumerate}
\end{lemma}
\begin{proof}
Recall that the pseudo-dimension of a class of functions is the VC dimension of its class of subgraphs $\{(z,t) : t < f(z)\}$. We use two permanence properties of VC classes of sets \citep[Lemma~2.6.17]{vandervaart1996weak}. Complementation leaves the VC dimension unchanged, and forming finite unions or intersections increases it by at most a constant factor. We also use, twice, the following elementary fact about level sets taken at a fixed level. If $\cF$ is a class of functions on a set $\mathcal{Z}$ with $\mathrm{Pdim}(\cF) = V'$, then the sets $\{z : f(z) > 0\}$, $f \in \cF$, have VC dimension at most $V'$. Indeed, if $z_1,\dots,z_m$ are shattered by these sets, then $z_i \in \{f > 0\}$ if and only if $(z_i,0)$ lies in the subgraph of $f$, so the $m$ points $(z_1,0),\dots,(z_m,0)$ are shattered by the subgraphs, whence $m \le V'$. That the level is held at $0$ is what makes the lift legitimate. Finally, a class of functions of $x$ is read below as a class of functions of $(x,y)$ or of $(x,y,t)$ that ignores the added coordinates, which changes neither pseudo-dimension nor VC dimension.

For (c), write $\{x : g(x)\le s\}$ as the complement of $\{x : (g-s)(x) > 0\}$. Applying the fact above to $\cG_c$ and then complementing gives $\mathrm{VC}(\cF_3) \le \mathrm{Pdim}(\cG_c)$. When $\cG$ is closed under the addition of constants, $\cG_c = \cG$ and the bound is $V$; when $\cG$ lies in a finite-dimensional vector space $\cV$, so does $\cG_c$ after adjoining the constants, and \citet[Lemma~2.6.15]{vandervaart1996weak} bounds $\mathrm{Pdim}(\cG_c)$ by $\dim(\cV)+2$.

For (b), the subgraph of $\cF_2$ is $\{(x,y,t) : t < (y-g(x))^2\}$. On $\{t < 0\}$ the condition holds identically. On $\{t \ge 0\}$ it says $|y - g(x)| > \sqrt{t}$, that is, $g(x) - y + \sqrt t < 0$ or $-g(x) + y + \sqrt t < 0$, so the subgraph is the union of $\{t<0\}$ with two sets of the form $\{f < 0\}$. The classes of functions appearing there are $\cG$ and $-\cG$ translated by the fixed function $(x,y,t) \mapsto \mp y + \sqrt{t \vee 0}$, and both negation and the addition of one fixed function leave a class VC-subgraph with an unchanged index \citep[Lemma~2.6.18]{vandervaart1996weak}. Each of the two sets is a level set at the fixed level $0$, hence of VC dimension $O(V)$ by the fact above, and the union keeps $O(V)$. Therefore $\mathrm{Pdim}(\cF_2) = O(V)$.

For (a), the subgraph of $\cF_1$ is $\{(x,y,t) : t < (y-g(x))^2\,\mathbbm{1}\{g(x)\le s\}\}$. On $\{t<0\}$ the condition again holds identically. On $\{t \ge 0\}$ it holds exactly when both $g(x) \le s$ and $t < (y-g(x))^2$, so the set is the intersection of the cylinder over a set from (c) with the subgraph from (b). Intersections and the union with $\{t<0\}$ keep the VC dimension $O(V)$, so $\mathrm{Pdim}(\cF_1)=O(V)$.
\end{proof}

\begin{lemma}[Integrable envelope]\label{lem:envelop}
Suppose \Cref{ass:regularity}(b) is satisfied, for the function classes defined in \Cref{lem:pdim}:
\begin{enumerate}[label=(\alph*)]
\item $\cF_1$ has envelope $\big(|Y|+G\big)^2$;
\item $\cF_2$ has envelope $\big(|Y|+G\big)^2$;
\item $\cF_3$ has envelope 1.
\end{enumerate}
And $\E\big[\big(|Y|+G\big)^2\big]<\infty$.
\end{lemma}
\begin{proof}
By \Cref{ass:regularity}(b) the envelope $G$ satisfies $|g| \le G$ for every $g \in \cG$, so $|y - g(x)| \le |y| + G(x)$ and therefore $\big(y-g(x)\big)^2 \le \big(|y|+G(x)\big)^2$, which gives (b). The members of $\cF_1$ multiply those of $\cF_2$ by an indicator taking values in $\{0,1\}$, so the same function bounds them, which gives (a). The members of $\cF_3$ are themselves indicators, so $1$ is an envelope, and (c) follows. Finally $\E\big[\big(|Y|+G\big)^2\big] \le 2\E[Y^2] + 2\E\big[G(X)^2\big]$, and both terms are finite by \Cref{ass:regularity}(b).
\end{proof}

\begin{lemma}[$P$-Glivenko-Cantelli for VC-subgraph classes]\label{lem:PGC_VC}
Let $\cF$ be a class of measurable functions with $\mathrm{Pdim}(\cF) < \infty$ and envelope $F$ satisfying $\E\big[F(Z)\big] < \infty$.  Then $\cF$ is $P$-Glivenko-Cantelli, i.e.
\begin{align*}
\sup_{f \in \cF}\big|P_n f - P f\big| \overset{P}{\to} 0.
\end{align*}
\end{lemma}
\begin{proof}
Let $V\ge1$ be a finite upper bound on the VC index of the subgraphs of $\cF$. For $M < \infty$ write $\cF_M = \{f\,\mathbbm{1}\{F \le M\} : f \in \cF\}$. Truncation by the fixed set $\{F \le M\}$ does not increase the pseudo-dimension and gives $\cF_M$ the constant envelope $M$, so \Cref{lem:pollard_cover} bounds its $L^2(Q)$ covering numbers by $K\,V(16e)^V(M/\delta)^{2V}$ uniformly in $Q$. Covering in $L^1(Q)$ is no harder than in $L^2(Q)$, so, taking $Q = P_n$,
\begin{align*}
\log N\big(\delta, \cF_M, L^1(P_n)\big) \le c\,V\big[1+\log(M/\delta)\big]
\end{align*}
for $0<\delta\le M$. For $\delta>M$, monotonicity gives the same bound as at $\delta=M$. These bounds are free of $n$ and hence $o_P(n)$. With the envelope condition $\E[F] < \infty$ and the measurability convention of \Cref{sec:asymptotics}, \citet[Theorem~2.4.3]{vandervaart1996weak} gives the conclusion.
\end{proof}

\begin{lemma}[Compactness]\label{lem:compact}
If $\mathrm{Pdim}(\cG) < \infty$ and $\E[G(X)^2] < \infty$, then $\cG$
is totally bounded in $L^2(P_X)$. If further $\cG$ is closed in $L^2(P_X)$, then $\cG$ is compact in   $L^2(P_X)$.
\end{lemma}

\begin{proof}
By \citet[Theorem 2.6.7]{vandervaart1996weak}, we have the polynomial covering number
bound $N\big(\epsilon, \cG, L^2(P_X)\big) < \infty$ for $\cG$, which is the definition of total boundedness.

 Since the space $L^2(P_X)$ is complete, a subset is compact if and only if it is closed and totally bounded \citep{dudley2002real}. Thus, the total boundedness and the closedness of $\cG$ yields that $\cG$ is compact in $L^2(P_X)$.
\end{proof}

\begin{lemma}[Pollard's covering bound]\label{lem:pollard_cover}
There is an absolute constant $K$ such that, for every probability measure $Q$ on $\mathcal{X}$, every class $\cF$ whose subgraphs have VC index at most $V$, and envelope $F$, and every $\varepsilon \in (0,1]$,
\begin{align}\label{eq:pollard_covering}
    N\big(\varepsilon\,\|F\|_{L^2(Q)},\,\cF,\,L^2(Q)\big) \le K\,V\,(16e)^V\,\varepsilon^{-2V}.
\end{align}
\end{lemma}
\begin{proof}
This is \citet[Theorem~2.6.7]{vandervaart1996weak}, whose exponent $2(V-1)$ we have weakened to $2V$ for $\varepsilon \le 1$. The VC index there exceeds the pseudo-dimension by at most one, so a class with $\mathrm{Pdim}(\cF) = V$ satisfies the bound with $V+1$ in place of $V$; the results below use it only through $O(V)$ and are unaffected.
\end{proof}

\begin{lemma}[VC-type maximal inequality]\label{lem:max_ineq}
Let $\cF$ be a class of measurable functions with $\mathrm{Pdim}(\cF) \le V$ for some $1\le V < \infty$ and envelope $F$ satisfying $\|F\|_{L^2(P)} < \infty$. There is an absolute constant $C$ such that
\begin{align}\label{eq:max_ineq}
    \E\Big[\sup_{f \in \cF}\big|P_n f - P f\big|\Big] \le C\,\|F\|_{L^2(P)}\,\sqrt{V/n}.
\end{align}
\end{lemma}

\begin{proof}
By \citet[Theorem~2.14.1]{vandervaart1996weak}, for any class $\cF$ with envelope $F$,
\begin{align}\label{eq:vdW_2_14_1}
    \E\big[\sup_{f\in\cF}\big|\sqrt n\,(P_n - P)\,f\big|\big]
    \le
    C'\,J(1,\cF)\,\|F\|_{L^2(P)},
\end{align}
where $C'$ is absolute and
\begin{align*}
    J(1,\cF) \coloneqq \int_0^1 \sup_Q \sqrt{1+\log N\big(\eta\,\|F\|_{L^2(Q)},\,\cF,\,L^2(Q)\big)}\,\mathrm d\eta
\end{align*}
is the uniform-entropy integral. Taking logarithms in \eqref{eq:pollard_covering} and using that the VC index is at most $V+1$,
\begin{align*}
    1+\log N\big(\eta\,\|F\|_{L^2(Q)},\,\cF,\,L^2(Q)\big)
    \le
    1+\log K + \log(V+1) + (V+1)\log(16e) + 2(V+1)\log(1/\eta)
    \le
    c_1\,V\,\big(1 + \log(1/\eta)\big),
\end{align*}
uniformly in $Q$, with $c_1$ absolute. Hence
\begin{align*}
    J(1,\cF)
    \le
    \int_0^1 \sqrt{c_1\,V\,\big(1+\log(1/\eta)\big)}d\eta
    =
    \sqrt{c_1\,V}\,\int_0^1 \sqrt{1+\log(1/\eta)}d\eta
    \eqqcolon
    C''\,\sqrt V,
\end{align*}
since $\int_0^1\sqrt{1+\log(1/\eta)}\,d\eta$ is finite. Substituting into \eqref{eq:vdW_2_14_1} and dividing by $\sqrt n$ produces \eqref{eq:max_ineq} with $C = C'\cdot C''$
\end{proof}

\section{Main proofs}
\subsection{Proof of \Cref{prop:reweight}}\label{app:reweight}
\begin{proof}
Let $T = t(X)$ be any statistic through which the weight factors, so that $w(X) = \bar w(T)$, and let $\phi$ be bounded and measurable. For every bounded $\sigma(T)$-measurable $\psi$, conditioning on $T$ gives
\begin{align*}
\E_w\big[\phi(Y)\,\psi(T)\big]
=
\E\big[\bar w(T)\,\phi(Y)\,\psi(T)\big]
=
\E\Big[\bar w(T)\,\psi(T)\,\E\big[\phi(Y) \,|\, T\big]\Big],
\end{align*}
since $\bar w(T)\psi(T)$ is $\sigma(T)$-measurable. Replacing $\phi(Y)$ by $\E_w[\phi(Y) \,|\, T]$ on the left and using $w(X) = \bar w(T)$ once more,
\begin{align*}
\E_w\big[\phi(Y)\,\psi(T)\big]
=
\E_w\Big[\E_w\big[\phi(Y) \,|\, T\big]\,\psi(T)\Big]
=
\E\Big[\bar w(T)\,\psi(T)\,\E_w\big[\phi(Y) \,|\, T\big]\Big].
\end{align*}
Equating the two and letting $\psi$ range over the bounded $\sigma(T)$-measurable functions,
\begin{align*}
\bar w(T)\Big(\E_w\big[\phi(Y) \,|\, T\big] - \E\big[\phi(Y) \,|\, T\big]\Big) = 0
\qquad P\text{-almost surely}.
\end{align*}
Now $P_w\big(\bar w(T) > 0\big) = \E\big[\bar w(T)\mathbbm{1}\{\bar w(T) > 0\}\big] = \E[\bar w(T)] = 1$, so the two conditional expectations agree $P_w$-almost surely. Since $\phi$ was an arbitrary bounded measurable function, $P_w(Y \,|\, T) = P(Y \,|\, T)$ as regular conditional distributions.

The weight factors through $V = s(X)$, through $X$ itself, and through any covariate subset $X_S$ with $S_w \subseteq S$, where $S_w$ indexes the covariates $w$ depends on. Taking $T$ to be each of these in turn gives the three statements.
\end{proof}

\subsection{Proof of \Cref{prop:indicators}}\label{app:indicators}

\begin{proof}
If $g \in \cG_I$, then $\sigma_g^2(F_g) = c_g$ almost surely, and taking expectations gives $c_g = \bar{\sigma}_g^2$. Since $\sigma_g^2(F_g) = \E[R_g \,|\, F_g]$, the law of iterated expectation gives
\begin{align*}
\E\big[R_g\,\mathbbm{1}\{F_g \leq s\}\big] = \E\big[\sigma_g^2(F_g)\,\mathbbm{1}\{F_g \leq s\}\big] = \bar{\sigma}_g^2\,P(F_g \leq s).
\end{align*}
Conversely, suppose this equality holds for all $s \in \R$. As above, $\E\big[R_g\,\mathbbm{1}\{F_g \le s\}\big] = \E\big[\sigma_g^2(F_g)\,\mathbbm{1}\{F_g \le s\}\big]$, so the supposed equality states that the two finite measures
\begin{align*}
A \mapsto \E\big[\sigma_g^2(F_g)\,\mathbbm{1}\{F_g \in A\}\big]
\qquad\text{and}\qquad
A \mapsto \bar{\sigma}_g^2\,P(F_g \in A)
\end{align*}
have the same distribution function. A finite measure on $\R$ is determined by its distribution function, so the two coincide on every Borel set $A$, giving $\E\big[\big(\sigma_g^2(F_g) - \bar{\sigma}_g^2\big)\,\mathbbm{1}\{F_g \in A\}\big] = 0$. On $A = \{t : \sigma_g^2(t) > \bar{\sigma}_g^2\}$ and on $A = \{t : \sigma_g^2(t) < \bar{\sigma}_g^2\}$ the integrand has a strict sign, so each set is $P_g$-null. Hence $\sigma_g^2(t) = \bar{\sigma}_g^2$ for $P_g$-almost every $t$, that is, $g \in \cG_I$.
\end{proof}

\subsection{Proof of \Cref{prop:completeness}} \label{app:completeness}

\begin{proof}[Proof of \Cref{prop:completeness}, forward direction]
If $g \in \cG_I$, then $\sigma_g^2(t) = \bar{\sigma}_g^2$ for $P_g$-almost every $t$, and the integrand of \eqref{eq:cvm} at threshold $s$ is
\begin{align*}
    \E\big[R_g\,\mathbbm{1}\{F_g \leq s\}\big] - \bar{\sigma}_g^2 P(F_g \leq s)
= \int_{(-\infty,s]} \big(\sigma_g^2(t) - \bar{\sigma}_g^2\big)\,\mathrm{d}P_g(t)  =0,
\end{align*}
so $\cT(g) = 0$.
\end{proof}

The converse rests on an identity that holds for every $g \in \cG$, whatever the value of $\cT(g)$. We record it separately because \Cref{app:rkhs} uses it again.

\begin{lemma}[The penalty on the quantile scale]\label{lem:quantile_form}
For every $g \in \cG$, with $\psi_g$ as in \Cref{sec:dro:rkhs} and $\Psi_g(u) = \int_0^u \psi_g(v)\,\mathrm{d}v$,
\begin{align*}
\Psi_g(0) = \Psi_g(1) = 0
\qquad\text{and}\qquad
\cT(g) = \int_0^1 \Psi_g(u)^2\,\mathrm{d}u.
\end{align*}
Moreover, the map $H_g$ carries $P_g$ to the uniform distribution on $[0,1]$, so a set of $u$ of Lebesgue measure zero pulls back to a set of $t$ of $P_g$-measure zero.
\end{lemma}

\begin{proof}
Let $U = H_g(F_g)$ and write $\sigma_U^2(u) = \sigma_g^2\big(H_g^{-1}(u)\big)$, so that $\psi_g = \sigma_U^2 - \bar\sigma_g^2$. Since $H_g$ is continuous, the probability integral transform gives $U \sim \mathrm{Uniform}(0,1)$, which is the stated pushforward property. The final claim of the lemma is then immediate, since a Lebesgue-null set $N \subseteq [0,1]$ has $P_g\big(\{t : H_g(t) \in N\}\big) = \mathrm{Leb}(N) = 0$; the direction from $[0,1]$ back to $\R$ is the only one used. An interval on which $H_g$ is constant carries no $P_g$-mass, so $F_g = H_g^{-1}(U)$ almost surely and, for every $s$, $\mathbbm{1}\{F_g \le s\} = \mathbbm{1}\{U \le H_g(s)\}$ almost surely.

Because $U$ is uniform, $\int_0^1 \psi_g(v)\,\mathrm{d}v = \E\big[\sigma_U^2(U)\big] - \bar\sigma_g^2 = \E\big[\sigma_g^2(F_g)\big] - \bar\sigma_g^2$, and the law of iterated expectation makes the last difference $\E[R_g] - \bar\sigma_g^2 = 0$. Hence $\Psi_g(1) = 0$, and $\Psi_g(0) = 0$ by definition. For every threshold $s$, iterated expectation likewise gives
\begin{align*}
    \E\big[R_g\,\mathbbm{1}\{F_g \leq s\}\big]
    = \E\big[\sigma_g^2(F_g)\,\mathbbm{1}\{F_g \leq s\}\big]
    = \E\big[\sigma^2_U(U)\,\mathbbm{1}\{U \leq H_g(s)\}\big]
    = \int_0^{H_g(s)} \sigma^2_U(v)\,\mathrm{d}v,
\end{align*}
and $P(F_g \leq s) = H_g(s)$, so the integrand of \eqref{eq:cvm} at threshold $s$ equals $\int_0^{H_g(s)} \psi_g(v)\,\mathrm{d}v = \Psi_g\big(H_g(s)\big)$. Averaging its square over $s \sim P_g$, and using $H_g(F_g) = U \sim \mathrm{Uniform}(0,1)$ once more,
\begin{align*}
    \cT(g) = \E\Big[\Psi_g\big(H_g(F_g)\big)^2\Big] = \int_0^1 \Psi_g(u)^2\,\mathrm{d}u. \qedhere
\end{align*}
\end{proof}

\begin{proof}[Proof of \Cref{prop:completeness}, converse]
Suppose $\cT(g) = 0$. By \Cref{lem:quantile_form}, $\int_0^1 \Psi_g(u)^2\,\mathrm{d}u = 0$, so $\Psi_g$ is zero at Lebesgue-almost every $u$. The function $\Psi_g$ is absolutely continuous on $[0,1]$, being the integral of $\psi_g \in L_1([0,1])$, and a continuous function that is zero almost everywhere is zero everywhere, so $\Psi_g(u) = 0$ for every $u \in [0,1]$. Differentiating, $\psi_g(u) = \Psi_g'(u) = 0$ for Lebesgue-almost every $u$, so the set $N \coloneqq \{u \in [0,1] : \sigma_U^2(u) \ne \bar\sigma_g^2\}$ is Lebesgue-null. It remains to transfer this to the $t$ scale. Since $H_g$ is constant on no interval carrying $P_g$-mass, $t = H_g^{-1}\big(H_g(t)\big)$ and hence $\sigma_g^2(t) = \sigma_U^2\big(H_g(t)\big)$ for $P_g$-almost every $t$, so up to a $P_g$-null set $\{t : \sigma_g^2(t) \ne \bar\sigma_g^2\} \subseteq \{t : H_g(t) \in N\}$. By the pushforward property in \Cref{lem:quantile_form}, $P_g\big(\{t : H_g(t) \in N\}\big) = \mathrm{Leb}(N) = 0$. Therefore $\sigma_g^2(t) = \bar{\sigma}_g^2$ for $P_g$-almost every $t$, that is, $g \in \cG_I$.
\end{proof}

\subsection{Proof of \Cref{thm:rkhs}}\label{app:rkhs}

\begin{proof}
Fix $u \in [0,1]$. The function $k_0(\cdot,u)$, with $k_0(v_1,v_2) = 1 - \max(v_1,v_2)$ as in \Cref{thm:rkhs}, satisfies $k_0(1,u) = 0$ and is absolutely continuous with $\partial_v k_0(v,u) = -\mathbbm{1}\{v > u\} \in L^2[0,1]$, so $k_0(\cdot,u) \in \cH$. For every $h \in \cH$,
\begin{align*}
\langle h, k_0(\cdot,u)\rangle_\cH
= \int_0^1 h'(v)\,\partial_v k_0(v,u)\,\mathrm{d}v
= -\int_u^1 h'(v)\,\mathrm{d}v
= h(u) - h(1) = h(u).
\end{align*}
Evaluation at $u$ is thus a bounded linear functional represented by $k_0(\cdot,u)$, so $\cH$ is an RKHS with reproducing kernel $k_0$. Equivalently, $k_0(v_1,v_2) = \min(1-v_1,\,1-v_2)$ arises from the standard kernel $\min(a,b)$ by the change of variables $a = 1-v_1$, $b = 1-v_2$, under which the boundary condition $f(0) = 0$ becomes $h(1) = 0$ \citep[Chapter~1, Exercise~14]{berlinet2004rkhs}.

Recall $\psi_g$ from \Cref{sec:dro:rkhs} and $\mu_g$ from \Cref{thm:rkhs}, and write $\Psi_g(u) \coloneqq \int_0^u \psi_g(v)\,\mathrm{d}v$. \Cref{lem:quantile_form} supplies the two facts we need about $\Psi_g$, namely $\Psi_g(0) = \Psi_g(1) = 0$ and $\cT(g) = \int_0^1 \Psi_g(u)^2\,\mathrm{d}u$. Substituting the explicit form of $k_0$,
\begin{align*}
\mu_g(u)
= \int_0^u \psi_g(v)(1-u)\,\mathrm{d}v + \int_u^1 \psi_g(v)(1-v)\,\mathrm{d}v
= (1-u)\,\Psi_g(u) + \int_u^1 \psi_g(v)(1-v)\,\mathrm{d}v,
\end{align*}
so $\mu_g(1) = 0$. Both terms on the right are integrals of $L_1$ functions against bounded factors, so $\mu_g$ is absolutely continuous, and differentiating each in turn,
\begin{align*}
\mu_g'(u) = -\Psi_g(u) + (1-u)\,\psi_g(u) - (1-u)\,\psi_g(u) = -\Psi_g(u) \qquad \text{for a.e. } u.
\end{align*}
Since $|\Psi_g(u)| \leq \int_0^1 |\psi_g| \leq 2\bar{\sigma}_g^2$, we have $\mu_g' \in L^2[0,1]$ and hence $\mu_g \in \cH$. For every $h \in \cH$, integration by parts with $\Psi_g(0) = \Psi_g(1) = 0$ gives
\begin{align}\label{eq:mu-inner}
\langle h, \mu_g\rangle_\cH
= \int_0^1 h'(u)\,\mu_g'(u)\,\mathrm{d}u
= -\int_0^1 h'(u)\,\Psi_g(u)\,\mathrm{d}u
= \int_0^1 h(u)\,\psi_g(u)\,\mathrm{d}u.
\end{align}
For the first display of the proposition, the density ratio \eqref{eq:Qh} gives
\begin{align*}
\E_{Q_h}[R_g] - \E_P[R_g] = \E_P\bigg[\Big(h(U) - \int_0^1 h(u)\,\mathrm{d}u\Big)\,R_g\bigg],
\end{align*}
and conditioning on $U$, with $\E[R_g \,|\, U] = \bar{\sigma}_g^2 + \psi_g(U)$ and $\int_0^1 \psi_g = 0$, yields $\int_0^1 \psi_g(u)\,h(u)\,\mathrm{d}u$, which equals $\langle h, \mu_g\rangle_\cH$ by \eqref{eq:mu-inner}.

Finally, taking $h = \mu_g$ and using $\mu_g' = -\Psi_g$,
\begin{align*}
\big\|\mu_g\big\|_\cH^2 = \int_0^1 \big(\mu_g'(u)\big)^2\,\mathrm{d}u = \int_0^1 \Psi_g(u)^2\,\mathrm{d}u,
\end{align*}
which is $\cT(g)$ by \Cref{lem:quantile_form}. This proves \Cref{thm:rkhs}.
\end{proof}

\subsection{Proof of \Cref{thm:dro}}\label{app:dro}

\begin{proof}
By \eqref{eq:risk-change} in \Cref{thm:rkhs},
\begin{align}\label{eq:excess-risk}
\E_{Q_h}[R_g] - \E_P[R_g] = \langle h, \mu_g\rangle_\cH \qquad \text{for all } h \in \cH.
\end{align}
Completing the square in $\cH$,
\begin{align*}
\langle h, \mu_g\rangle_\cH - \frac{\|h\|_\cH^2}{4\lambda}
= \lambda\,\|\mu_g\|_\cH^2 - \frac{\|h - 2\lambda\,\mu_g\|_\cH^2}{4\lambda},
\end{align*}
so the supremum over $h \in \cH$ equals $\lambda\,\|\mu_g\|_\cH^2$ and is attained only at $h^\ast = 2\lambda\,\mu_g$. Hence, using \eqref{eq:excess-risk} and $\|\mu_g\|_\cH^2 = \cT(g)$ from \Cref{thm:rkhs},
\begin{align*}
\sup_{h \in \cH} \bigg\{\E_{Q_h}[R_g] - \frac{\|h\|_\cH^2}{4\lambda}\bigg\}
= \E_P[R_g] + \sup_{h \in \cH} \bigg\{\langle h, \mu_g\rangle_\cH - \frac{\|h\|_\cH^2}{4\lambda}\bigg\}
= \E_P[R_g] + \lambda\,\cT(g),
\end{align*}
the first identity \eqref{eq:dro} of \Cref{thm:dro}, with supremum attained at $h^\ast = 2\lambda\,\mu_g$. Minimising both sides over $g \in \cG$, attained at $g_\lambda$, gives \eqref{eq:dro_minimax}.
\end{proof}

\subsection{Proof of \Cref{cor:signed}}\label{app:signed}
\begin{proof}
The constraint $w_h\geq0$ restricts the supremum in \eqref{eq:dro}, which gives the inequality. Equality holds whenever the unrestricted maximiser $h^\ast=2\lambda\mu_g$ is feasible. By the reproducing property,
\begin{align*}
|\mu_g(u)|
\leq \|\mu_g\|_\cH k_0(u,u)^{1/2}
\leq \cT(g)^{1/2}
\qquad (u\in[0,1]).
\end{align*}
Consequently,
\begin{align*}
w_{h^\ast}(x)
&=1+2\lambda\bigg\{\mu_g\big(H_g(g(x))\big)-\int_0^1\mu_g(u)\,\mathrm du\bigg\}\\
&\geq1-4\lambda\cT(g)^{1/2}.
\end{align*}
Thus $w_{h^\ast}\geq0$ when $\lambda\leq\{4\cT(g)^{1/2}\}^{-1}$. If $\cT(g)=0$, then $\mu_g=0$ and the conclusion holds for every $\lambda$.
\end{proof}

\subsection{Proof of \Cref{prop:unique}}\label{app:unique}
\begin{proof}
We first record that the minimisation runs over a set containing $g_0$. Under \Cref{ass:structural} the prediction error of $g_0$ is $\varepsilon$, so $\sigma_{g_0}^2(t) = \E\big[\varepsilon^2 \,\big|\, g_0(X) = t\big] = \sigma^2$, constant in $t$ by $\varepsilon \indep g_0(X)$; hence $g_0 \in \cG_I$.

Since $\varepsilon \indep g_0(X)$ and $\E[\varepsilon]=0$, we have $\E[g_0(X)\varepsilon]=0$, and \Cref{ass:ortho} gives $\E[g(X)\varepsilon]=0$ for every $g \in \cG_I$. Hence, for any $g \in \cG_I$,
\begin{align*}
    \E\Big[\big(Y-g(X)\big)^2\Big] &=\E\Big[\big(Y- g_0(X)+g_0(X)-g(X)\big)^2\Big]\\
    &=\E\Big[\big(Y- g_0(X)\big)^2\Big] + \E\Big[\big(g_0- g(X)\big)^2\Big] + 2\E\Big[\big(Y- g_0(X)\big)\big(g_0- g(X)\big)\Big]\\
    &=\E\Big[\big(Y- g_0(X)\big)^2\Big] + \E\Big[\big(g_0- g(X)\big)^2\Big] + 2\E\Big[\varepsilon\big(g_0- g(X)\big)\Big]\\
    & = \E\Big[\big(Y- g_0(X)\big)^2\Big] + \E\Big[\big(g_0- g(X)\big)^2\Big]\\
    & \ge \E\Big[\big(Y- g_0(X)\big)^2\Big].
\end{align*}
Equality in the last step forces $\E\big[\big(g_0(X)-g(X)\big)^2\big] = 0$, that is, $g(X) = g_0(X)$ almost surely.
Thus $g_0$ attains the minimum, and every minimiser $g$ satisfies $\|g-g_0\|_{P_X}=0$.
\end{proof}
\subsection{Continuity of the penalty}\label{app:continuity}

The limiting statement of \Cref{prop:identification} and the consistency results of \Cref{sec:asymptotics} use the following continuity property of the penalty. Its only moment requirement is $\E[Y^2] < \infty$, which holds automatically under \Cref{ass:structural}, since $Y = g_0(X) + \varepsilon$ with $g_0(X) \in L^2(P_X)$ and $\E[\varepsilon^2] < \infty$.

\begin{lemma}[Continuity of $\cT$ in $L^2(P_X)$]\label{lem:continuity}
Assume $\E[Y^2] < \infty$. Let $g^\ast \in \cG$ and let $(g_k)_{k\ge 1}$ be any sequence in $\cG$ with $\|g_k - g^\ast\|_{P_X} \to 0$ as $k \to \infty$. Then $\cT(g_k) \to \cT(g^\ast)$.
\end{lemma}

\begin{proof}
For each $g \in \cG$ and $s \in \R$, define the population deviation
\begin{align*}
\nu(g,s) \coloneqq \E\big[R_g\, \mathbbm{1}\{F_g \le s\}\big] - \bar\sigma_g^2\, P(F_g \le s),
\end{align*}
so that $\cT(g) = \E\big[\nu\big(g, g(X)\big)^2\big]$ and $|\nu(g,s)| \le 2\bar\sigma_g^2$ for all $s$. Since $g_k \to g^\ast$ in $L^2(P_X)$, we have $g_k(X) \to g^\ast(X)$ in probability, and hence, for every $s$,
\begin{align}\label{eq:indicator_conv}
\mathbbm{1}\{g_k(X) \le s\} \to \mathbbm{1}\{g^\ast(X) \le s\} \quad\text{in probability},
\end{align}
because $P\big(g^\ast(X) = s\big) = 0$ by the continuity of $H_{g^\ast}$ (\Cref{def:Gclass}).

We first prove the uniform convergence $\sup_s |\nu(g_k,s) - \nu(g^\ast,s)| \to 0$. For any fixed $s$,
\begin{align*}
\big|\nu(g_k,s) - \nu(g^\ast,s)\big|
\le{}& \underbrace{\Big|\E\big[R_{g_k}\mathbbm{1}\{F_{g_k} \le s\}\big] - \E\big[R_{g^\ast}\mathbbm{1}\{F_{g^\ast} \le s\}\big]\Big|}_{\textcircled{1}} \\
&+ \underbrace{\big|\bar\sigma_{g_k}^2 - \bar\sigma_{g^\ast}^2\big|}_{\textcircled{2}}
+ \underbrace{\bar\sigma_{g_k}^2\, \big|P(F_{g_k} \le s) - P(F_{g^\ast} \le s)\big|}_{\textcircled{3}}.
\end{align*}
For $\textcircled{1}$, since $0 \le \mathbbm{1}\{F_{g_k} \le s\} \le 1$ pointwise,
\begin{align*}
\textcircled{1} \le \E\big|R_{g_k} - R_{g^\ast}\big| + \Big|\E\big[R_{g^\ast}\mathbbm{1}\{F_{g_k} \le s\}\big] - \E\big[R_{g^\ast}\mathbbm{1}\{F_{g^\ast} \le s\}\big]\Big|.
\end{align*}
Expanding $R_{g_k} - R_{g^\ast}$ and applying the Cauchy--Schwarz inequality,
\begin{align*}
\E\big|R_{g_k} - R_{g^\ast}\big|
&= \E\Big[\big|g^\ast(X) - g_k(X)\big|\, \big|2Y - g_k(X) - g^\ast(X)\big|\Big] \\
&\le \big\|g_k - g^\ast\big\|_{P_X} \Big(2\,\E[Y^2]^{1/2} + \E\big[g_k(X)^2\big]^{1/2} + \E\big[g^\ast(X)^2\big]^{1/2}\Big),
\end{align*}
which converges to $0$, the second factor being bounded since $\E[Y^2] < \infty$ and $g_k$ converges in $L^2(P_X)$. This bound does not depend on $s$, and it also gives $\textcircled{2} \le \E|R_{g_k} - R_{g^\ast}| \to 0$; in particular $\bar\sigma_{g_k}^2 \to \bar\sigma_{g^\ast}^2$.

For the remaining part of $\textcircled{1}$, define a tilted probability measure $\tilde{P}$ with density
\begin{align*}
\frac{\mathrm{d}\tilde{P}}{\mathrm{d}P} = \frac{R_{g^\ast}}{\bar\sigma_{g^\ast}^2},
\end{align*}
which is non-negative with $\E_P[R_{g^\ast}/\bar\sigma_{g^\ast}^2] = 1$. Under $\tilde{P}$,
\begin{align*}
\frac{\E\big[R_{g^\ast}\mathbbm{1}\{F_{g_k} \le s\}\big]}{\bar\sigma_{g^\ast}^2} = \tilde{P}\big(g_k(X) \le s\big)
\quad\text{and}\quad
\frac{\E\big[R_{g^\ast}\mathbbm{1}\{F_{g^\ast} \le s\}\big]}{\bar\sigma_{g^\ast}^2} = \tilde{P}\big(g^\ast(X) \le s\big),
\end{align*}
the distribution functions of $g_k(X)$ and $g^\ast(X)$ under $\tilde{P}$.

For fixed $s$, put
\begin{align*}
D_{k,s}=\big|\mathbbm{1}\{g_k(X)\leq s\}-\mathbbm{1}\{g^\ast(X)\leq s\}\big|.
\end{align*}
By \eqref{eq:indicator_conv}, $D_{k,s}\to0$ in probability. Since $0\leq D_{k,s}\leq1$, it follows that $\E[D_{k,s}]\to0$. Moreover, for every $M>0$,
\begin{align*}
\E[R_{g^\ast}D_{k,s}]
\leq M\E[D_{k,s}]
+\E\big[R_{g^\ast}\mathbbm{1}\{R_{g^\ast}>M\}\big].
\end{align*}
Taking the limit superior as $k\to\infty$ gives
\begin{align*}
\limsup_{k\to\infty}\E[R_{g^\ast}D_{k,s}]
\leq \E\big[R_{g^\ast}\mathbbm{1}\{R_{g^\ast}>M\}\big].
\end{align*}
Since $R_{g^\ast}\in L^1(P)$, the right-hand side tends to $0$ as $M\to\infty$. Letting first $k\to\infty$ and then $M\to\infty$ gives $\E[R_{g^\ast}D_{k,s}]\to0$. Hence $\tilde P(g_k(X)\leq s)\to\tilde P(g^\ast(X)\leq s)$ for every $s$.
Since $\tilde{P}$ is absolutely continuous with respect to $P$ and $P(g^\ast(X) = s) = 0$ for every $s$, the distribution of $g^\ast(X)$ under $\tilde{P}$ has no atoms. Its distribution function is therefore continuous, and P\'olya's theorem \citep[Lemma~2.11]{van1998asymptotic} yields
\begin{align*}
\sup_s \Big|\E\big[R_{g^\ast}\mathbbm{1}\{F_{g_k} \le s\}\big] - \E\big[R_{g^\ast}\mathbbm{1}\{F_{g^\ast} \le s\}\big]\Big| \to 0.
\end{align*}
For $\textcircled{3}$,
$|P(F_{g_k}\leq s)-P(F_{g^\ast}\leq s)|\leq\E[D_{k,s}]\to0$ for every $s$, and $H_{g^\ast}$ is continuous. P\'olya's theorem again gives $\sup_s |P(F_{g_k} \le s) - P(F_{g^\ast} \le s)| \to 0$, while the factor $\bar\sigma_{g_k}^2$ is bounded. Combining the three terms proves the uniform convergence.

By the triangle inequality,
\begin{align*}
\Big|\nu\big(g_k, g_k(X)\big) - \nu\big(g^\ast, g^\ast(X)\big)\Big|
\le \sup_s \big|\nu(g_k,s) - \nu(g^\ast,s)\big| + \Big|\nu\big(g^\ast, g_k(X)\big) - \nu\big(g^\ast, g^\ast(X)\big)\Big|.
\end{align*}
The map $s \mapsto \nu(g^\ast,s) = \bar\sigma_{g^\ast}^2\big\{\tilde{P}(g^\ast(X) \le s) - H_{g^\ast}(s)\big\}$ is bounded and continuous, being a difference of two continuous distribution functions. Since $g_k(X) \to g^\ast(X)$ in probability, the second term converges to $0$ in probability, and with the first term the whole difference does. The bounds $|\nu(g_k,\cdot)| \le 2\bar\sigma_{g_k}^2$ and $|\nu(g^\ast,\cdot)| \le 2\bar\sigma_{g^\ast}^2$, together with $\bar\sigma_{g_k}^2 \to \bar\sigma_{g^\ast}^2$, make the squared variables uniformly bounded and hence uniformly integrable. Convergence in probability and uniform integrability give convergence in $L^1$, so
\begin{align*}
\E\Big[\nu\big(g_k, g_k(X)\big)^2\Big] \to \E\Big[\nu\big(g^\ast, g^\ast(X)\big)^2\Big],
\end{align*}
that is, $\cT(g_k) \to \cT(g^\ast)$.
\end{proof}

\subsection{Proof of \Cref{prop:identification}}\label{app:identification}
\begin{proof}
Since \eqref{eq:estimator} need not have a unique minimiser, fix one selection $\lambda \mapsto g_\lambda$ throughout; the argument applies to any such selection. Under \Cref{ass:structural} the prediction error of $g_0$ is $\varepsilon$, independent of $g_0(X)$, so $\sigma_{g_0}^2$ is constant and $g_0 \in \cG_I$; hence $\cT(g_0) = 0$ by \Cref{prop:completeness}, and $\E\big[\big(Y - g_0(X)\big)^2\big] = \E[\varepsilon^2] = \sigma^2$. Now let $\tilde g \in \cG$ be an $L^2(P_X)$-limit point of the selected family, and let $(g_{\lambda_k})_{k\ge 1}$ be a subsequence with $g_{\lambda_k} \to \tilde g$ in $L^2(P_X)$ and $\lambda_k \to \infty$. By the optimality of $g_\lambda$ in \eqref{eq:estimator}, comparing it with $g_0$,
\begin{align*}
\E\big[\big(Y - g_\lambda(X)\big)^2\big] + \lambda\, \cT(g_\lambda) \le \E\big[\big(Y - g_0(X)\big)^2\big] + \lambda\,\cT(g_0) = \sigma^2.
\end{align*}
As $\cT(g_\lambda) \ge 0$ and $\E\big[\big(Y - g_\lambda(X)\big)^2\big] \ge 0$, this yields both
\begin{align*}
\cT(g_\lambda) \le \frac{\sigma^2}{\lambda} \quad\xrightarrow{\ \lambda\to\infty\ }\quad 0
\qquad\text{and}\qquad
\E\big[\big(Y - g_\lambda(X)\big)^2\big] \le \sigma^2 .
\end{align*}
Along the subsequence $\cT(g_{\lambda_k}) \to 0$, and since $g_{\lambda_k} \to \tilde g$ in $L^2(P_X)$, \Cref{lem:continuity} gives $\cT(\tilde g) = \lim_k \cT(g_{\lambda_k}) = 0$; hence $\tilde g \in \cG_I$ by \Cref{prop:completeness}.

It remains to show $\tilde g = g_0$. The risk is continuous along the subsequence. With $R_g = \big(Y-g(X)\big)^2$, the identity $a^2-b^2=(a-b)(a+b)$ and the Cauchy--Schwarz inequality give
\begin{align*}
\Big|\E[R_{g_{\lambda_k}}] - \E[R_{\tilde g}]\Big|
\le \E\big|R_{g_{\lambda_k}} - R_{\tilde g}\big|
&= \E\Big[\big|\tilde g(X) - g_{\lambda_k}(X)\big|\,\big|2Y - g_{\lambda_k}(X) - \tilde g(X)\big|\Big] \\
&\le \big\|g_{\lambda_k} - \tilde g\big\|_{P_X}\,\big\|2Y - g_{\lambda_k}(X) - \tilde g(X)\big\|_{P},
\end{align*}
and the second factor is bounded since $\E[Y^2] < \infty$ and $g_{\lambda_k} \to \tilde g$ in $L^2(P_X)$, so $\E[R_{g_{\lambda_k}}] \to \E[R_{\tilde g}]$ and therefore $\E\big[\big(Y - \tilde g(X)\big)^2\big] \le \sigma^2 = \E\big[\big(Y - g_0(X)\big)^2\big]$.
Since $\tilde g\in\cG_I$, the reverse risk inequality follows from \Cref{prop:unique}. Thus $\tilde g$ is a minimiser and $\|\tilde g-g_0\|_{P_X}=0$.
\end{proof}

\subsection{Proof of \Cref{prop:misspecified}}\label{app:misspecified}
\begin{proof}
Fix $g^\dagger\in\cG_I$. Since $\cT(g^\dagger)=0$, optimality of $g_\lambda$ gives
\begin{align*}
\cL(g_\lambda)+\lambda\cT(g_\lambda)
\leq \cL(g^\dagger).
\end{align*}
Hence $\cT(g_\lambda)\leq\cL(g^\dagger)/\lambda\to0$ and $\cL(g_\lambda)\leq\cL(g^\dagger)$. Let $\tilde g\in\cG$ be an $L^2(P_X)$-limit point. By \Cref{lem:continuity}, $\cT(\tilde g)=0$, so $\tilde g\in\cG_I$. The risk is continuous under the same convergence, as shown in \Cref{app:identification}; therefore $\cL(\tilde g)\leq\cL(g^\dagger)$. Since $g^\dagger$ was arbitrary, $\tilde g\in\argmin_{g\in\cG_I}\cL(g)$.
\end{proof}

\subsection{Proof of \Cref{thm:uniform_t}}\label{app:uniform_t}
\begin{proof}[Proof of part (a)]
    For each $g \in \cG$ and $s \in \R$, define 
\begin{align*}
\nu(g,s) &\coloneqq \E\big[R_g\,\mathbbm{1}\{F_g \le s\}\big] - \bar{\sigma}_g^2\,P(F_g \le s), \\
\nu_n(g,s) &\coloneqq P_n\big[R_g\,\mathbbm{1}\{F_g \le s\}\big] - \bar{R}_g\,P_n\big[\mathbbm{1}\{F_g\leq s\}\big],
\end{align*}
where $\bar{R}_g = P_n[R_g]$. Thus, 
\begin{align*}
    \hat{\cT}_n(g) &= \frac{1}{n}\sum_{i=1}^n \nu_n\big(g,\, g(X_i)\big)^2, \\
    \cT(g) &= \E_P\Big[\nu\big(g,\, g(X)\big)^2\Big].
\end{align*}
Decompose $\hat{\cT}_n(g)-\cT(g)$:
\begin{align*}
\hat{\cT}_n(g) - \cT(g)
&= \underbrace{\frac{1}{n}\sum_{i=1}^{n} \Big[\nu_n\big(g,g(X_i)\big)^2 - \nu\big(g,g(X_i)\big)^2\Big]}_{\textcircled{1}}
+ \underbrace{\frac{1}{n}\sum_{i=1}^{n} \nu\big(g,g(X_i)\big)^2 - \E_P\Big[\nu\big(g,g(X)\big)^2\Big]}_{\textcircled{2}}.
\end{align*}
For \textcircled{1}:
\begin{align*}
|\textcircled{1}|\le\frac{1}{n}\sum_{i=1}^{n}\bigg\{\Big|\nu_n\big(g, g(X_i)\big) - \nu\big(g, g(X_i)\big)\Big| \cdot
\Big|\nu_n\big(g, g(X_i)\big)+\nu\big(g, g(X_i)\big)\Big|\bigg\}.
\end{align*}
Bounding pointwise by suprema:
\begin{align*}
    \sup_{g\in \cG}|\textcircled{1}| \leq \underbrace{\sup_{g\in\cG,s\in\R}\Big|\nu_n(g,s)-\nu(g,s)\Big|}_{A_n} \cdot 2 \underbrace{\sup_{g\in\cG,s\in\R}\bigg(\Big|\nu_n(g,s)\Big|+\Big|\nu(g,s)\Big|\bigg)}_{B_n}.
\end{align*}
For $A_n$, further decompose $\nu_n(g,s)-\nu(g,s)$:
\begin{align*}
    \nu_n(g,s) - \nu(g,s) &= \Big\{P_n\big[R_g \mathbbm{1}\{g(X) \le s\}\big] - \E\big[R_g \mathbbm{1}\{g(X) \le s\}\big]\Big\} \nonumber\\
    &~~~~- \Big\{P_n[R_g] - \E[R_g]\Big\} \cdot P_n\big[\mathbbm{1}\{g(X) \le s\}\big]\nonumber\\
    &~~~~- \bar\sigma_g^2 \cdot \Big\{P_n\big[\mathbbm{1}\{g(X) \le s\}\big] - \E\big[\mathbbm{1}\{g(X) \le s\}\big]\Big\}.
\end{align*}
Each of the three empirical-process terms is indexed by a function class with finite pseudo-dimension and an integrable envelope (\Cref{lem:pdim,lem:envelop}), so \Cref{lem:PGC_VC} gives:
\begin{itemize}
    \item $\sup_{g\in\cG,s\in \R} \Big|P_n\big[R_g \mathbbm{1}\{g(X) \le s\}\big] - \E\big[R_g \mathbbm{1}\{g(X) \le s\}\big]\Big|\overset{P}{\to}0$;
    \item $\sup_{g\in\cG}\big|P_n[R_g] - \E[R_g]\big| \overset{P}{\to}0$;
    \item $\sup_{g\in\cG,s\in\R}\Big|P_n\big[\mathbbm{1}\{g(X) \le s\}\big] - \E\big[\mathbbm{1}\{g(X) \le s\}\big]\Big|\overset{P}{\to}0$.
\end{itemize}
The remaining factors are bounded. Indeed, $P_n\big[\mathbbm{1}\{g(X) \le s\}\big] \leq 1$ and $\bar{\sigma}_g^2 \le \E\Big[\big(|Y|+G\big)^2\Big]<\infty$. Combining these results yields $A_n\overset{P}{\to}0$.

For $B_n$, for every $g \in \cG$ and $s \in \R$,
\begin{align*}
    \big|\nu(g, s)\big| &\le \Big|\E\big[R_g\,\mathbbm{1}\{F_g \le s\}\big]\Big| + \bar\sigma_g^2P(F_g \le s) \\
    &\le \E[R_g] + \bar\sigma_g^2 \\
    &= 2\bar\sigma_g^2 \\
    &= 2\E\big[\big(Y - g(X)\big)^2\big] \\
    &\le 2\E\big[\big(|Y| + G(X)\big)^2\big] <\infty,
\end{align*}
where the last step uses \Cref{lem:envelop}(b), and
\begin{align*}
    \big|\nu_n(g,s)\big| \le P_n\big[R_g\, \mathbbm{1}\{g(X) \le s\}\big] + \bar{R}_g \le 2\,P_n\big[\big(|Y|+G(X)\big)^2\big],
\end{align*}
and $P_n\big[\big(|Y|+G(X)\big)^2\big] \overset{P}{\to} \E\big[\big(|Y|+G(X)\big)^2\big] < \infty$ by the weak law of large numbers and \Cref{ass:regularity}(b). Hence $B_n = O_P(1)$. Combining it with  $A_n\overset{P}{\to}0$, we get
\begin{align*}
    \sup_{g\in\cG}\bigg|\frac{1}{n}\sum_{i=1}^{n} \Big[\nu_n\big(g,g(X_i)\big)^2 - \nu\big(g,g(X_i)\big)^2\Big]\bigg|\overset{P}{\to}0.
\end{align*}

\textcircled{2} is an empirical-process deviation for the function class $\{\nu\big(g,g(x)\big)^2 : g \in \cG\}$. Under \Cref{ass:regularity}(d), we have:
\begin{align*}
    \sup_{g\in\cG} \Big|P_n \big[\nu\big(g,g(x)\big)^2\big] - \E\big[\nu\big(g,g(x)\big)^2\big]\Big| \overset{P}{\to} 0.
\end{align*}
Therefore, we obtain $\sup_{g\in\cG}\Big|\hat{\cT}_n(g)-\cT(g)\Big|\overset{P}{\to}0$.
The second empirical-process bound used in the control of $A_n$ also gives
\begin{align*}
\sup_{g\in\cG}\big|P_n[R_g]-\E[R_g]\big|\overset{P}{\longrightarrow}0.
\end{align*}
For each fixed $\lambda\geq0$, the triangle inequality now gives
\begin{align*}
&\sup_{g\in\cG}\Big|
P_n[R_g]+\lambda\hat\cT_n(g)-\cL_{\rm CvM}(g;\lambda)
\Big|\\
&\qquad\leq
\sup_{g\in\cG}\big|P_n[R_g]-\E[R_g]\big|
+\lambda\sup_{g\in\cG}\big|\hat\cT_n(g)-\cT(g)\big|
\overset{P}{\longrightarrow}0.
\end{align*}
\end{proof}
\begin{proof}[Proof of part (b)]
By \Cref{lem:pdim,lem:envelop}, the class $\cF_2=\{(x,y)\mapsto(y-g(x))^2:g\in\cG\}$ has pseudo-dimension $O(V)$ and the $L^2(P)$ envelope $(|Y|+G)^2$. Hence \Cref{lem:max_ineq} and Markov's inequality give
\begin{align*}
\sup_{g\in\cG}\Big|P_n\big[(Y-g(X))^2\big]-\cL(g)\Big|
=O_P\big(\sqrt{V/n}\big).
\end{align*}

For the convergence rate of $\sup_{g\in\cG}\big|\hat\cT_n(g) - \cT(g)\big|$, we reuse the decomposition $\hat\cT_n(g) - \cT(g) = \textcircled{1} + \textcircled{2}$ from the proof of part (a), where
\begin{align*}
    \textcircled{1} = \frac{1}{n}\sum_{i=1}^n\Big[\nu_n(g,g(X_i))^2 - \nu(g,g(X_i))^2\Big],\quad
    \textcircled{2} = \frac{1}{n}\sum_{i=1}^n\nu(g,g(X_i))^2 - \E\big[\nu(g,g(X))^2\big].
\end{align*}

For \textcircled{1}, following the same factorisation as in part (a), $\sup_{g\in \cG}|\textcircled{1}| \le 2\,A_n\,B_n$ with $A_n = \sup_{g,s}|\nu_n(g,s)-\nu(g,s)|$ and $B_n = \sup_{g,s}\big(|\nu_n(g,s)|+|\nu(g,s)|\big)$. The proof of part (a) establishes $B_n = O_P(1)$. For $A_n$, the decomposition 
\begin{align*}
    \nu_n(g,s)-\nu(g,s)
    &= \big\{P_n[R_g\,\mathbbm{1}\{F_g\le s\}] - \E[R_g\,\mathbbm{1}\{F_g\le s\}]\big\}\\
    &\quad - \big\{P_n[R_g] - \E[R_g]\big\}\,P_n[\mathbbm{1}\{F_g\le s\}]\\
    &\quad - \bar\sigma_g^2\,\big\{P_n[\mathbbm{1}\{F_g\le s\}] - P(F_g\le s)\big\},
\end{align*}
indexed by the three classes $\cF_1, \cF_2, \cF_3$ of \Cref{lem:pdim}. Each of them has pseudo-dimension $O(V)$. The corresponding envelopes are $(|Y|+G)^2$, $(|Y|+G)^2$ and $1$, all in $L^2(P)$ by \Cref{ass:regularity}(b). Applying \Cref{lem:max_ineq} to each class and combined with $P_n[\mathbbm{1}\{F_g\le s\}]\le 1$ and $\bar\sigma_g^2\le\E\big[\big(|Y|+G(X)\big)^2\big]$, we obtain
\begin{align*}
    \E\Big[\sup_{g,s}\big|P_n[R_g\,\mathbbm{1}\{F_g\le s\}] - \E[R_g\,\mathbbm{1}\{F_g\le s\}]\big|\Big] &= O\big(\sqrt{V/n}\,\big), \\
    \E\Big[\sup_{g,s}\big|\big\{P_n[R_g] - \E[R_g]\big\}\,P_n[\mathbbm{1}\{F_g\le s\}]\big|\Big] &= O\big(\sqrt{V/n}\,\big), \\
    \E\Big[\sup_{g,s}\big|\bar\sigma_g^2\,\big\{P_n[\mathbbm{1}\{F_g\le s\}] - P(F_g\le s)\big\}\big|\Big] &= O\big(\sqrt{V/n}\,\big) .
\end{align*}
Hence  $\E[A_n] =  O\big(\sqrt{V/n}\,\big)$. By Markov's inequality, $A_n = O_P\big(\sqrt{V/n}\,\big)$. Thus, $\sup_{g\in \cG}|\textcircled{1}|\leq 2A_nB_n = O_P\big(\sqrt{V/n}\,\big)$.

For $\textcircled{2}$,  \Cref{ass:reg_rate} assumes that the class $\{x\mapsto\nu(g,g(x))^2:g\in\cG\}$ has finite pseudo-dimension $V_2$. Since $|\nu(g,s)| \le 2\bar\sigma_g^2 \le 2\,\E\big[\big(|Y|+G(X)\big)^2\big]$ for every $g \in \cG$ and $s \in \R$, the class has the constant envelope $4\Big(\E\big[\big(|Y|+G(X)\big)^2\big]\Big)^2$, finite by \Cref{ass:regularity}(b). \Cref{lem:max_ineq} gives
\begin{align*}
    \E\big[\sup_{g\in\cG}|\textcircled{2}|\big] = O\big(\sqrt{V_2/n}\,\big) = O\big(\sqrt{V/n}\,\big).
\end{align*}
By Markov's inequality, $\sup_{g\in\cG}|\textcircled{2}| = O_P\big(\sqrt{V/n}\,\big)$.

Combining the rates for $\sup_{g\in \cG}|\textcircled{1}|$ and $\sup_{g\in \cG}|\textcircled{2}|$, we obtain $\sup_{g\in\cG}\big|\hat\cT_n(g) - \cT(g)\big|=O_P\big(\sqrt{V/n}\,\big)$.

For every fixed $\lambda\geq0$, the triangle inequality and the two rates above yield
\begin{align*}
\sup_{g\in\cG}\Big|
P_n[(Y-g(X))^2]+\lambda\hat\cT_n(g)
-\cL_{\rm CvM}(g;\lambda)
\Big|
=O_P\big((1+\lambda)\sqrt{V/n}\big).
\end{align*}
\end{proof}

\subsection{Proof of \Cref{prop:fixed_lambda}}\label{app:fixed_lambda}
\begin{proof}
Let
\begin{align*}
\Delta_{n,\lambda}
\coloneqq
\sup_{g\in\cG}
\Big|
P_n[(Y-g(X))^2]+\lambda\hat\cT_n(g)
-\cL_{\rm CvM}(g;\lambda)
\Big|.
\end{align*}
For fixed $\lambda$, \Cref{thm:uniform_t}(a) gives $\Delta_{n,\lambda}=o_P(1)$.

For any $g_\lambda\in\cG_\lambda^\ast$, empirical optimality gives
\begin{align*}
0
\leq
\cL_{\rm CvM}(\hat g_{n,\lambda};\lambda)
-\cL_{\rm CvM}(g_\lambda;\lambda)
\leq 2\Delta_{n,\lambda}.
\end{align*}
This proves the objective-gap claim.

Define $d_{P_X}(g,A)=\inf_{f\in A}\|g-f\|_{P_X}$. By \Cref{lem:compact}, $\cG$ is compact in $L^2(P_X)$. The risk and $\cT$ are continuous on $\cG$, the latter by \Cref{lem:continuity}, so $\cG_\lambda^\ast$ is nonempty and compact. Fix $\varepsilon>0$ and consider
\begin{align*}
K_\varepsilon
=
\{g\in\cG:d_{P_X}(g,\cG_\lambda^\ast)\geq\varepsilon\}.
\end{align*}
If $K_\varepsilon$ is nonempty, compactness and continuity give
\begin{align*}
c_\varepsilon
\coloneqq
\min_{g\in K_\varepsilon}\cL_{\rm CvM}(g;\lambda)
-\min_{g\in\cG}\cL_{\rm CvM}(g;\lambda)
>0.
\end{align*}
Hence
\begin{align*}
\Pr\!\big(d_{P_X}(\hat g_{n,\lambda},\cG_\lambda^\ast)\geq\varepsilon\big)
\leq
\Pr(2\Delta_{n,\lambda}\geq c_\varepsilon)
\longrightarrow0.
\end{align*}
The claim is immediate if $K_\varepsilon$ is empty. When $\cG_\lambda^\ast=\{g_\lambda\}$, the distance to the set equals $\|\hat g_{n,\lambda}-g_\lambda\|_{P_X}$.
\end{proof}

\subsection{Proof of \Cref{thm:consistency_g}}\label{app:consistency_g}
\begin{proof}[Proof of part (a)]
We first show that the penalty at $\hat{g}_{n,\lambda_n}$ converges to $0$, and then turn this into consistency through the functional $M$.

\emph{Step 1: $\cT(\hat{g}_{n,\lambda_n}) \overset{P}{\to} 0$.} The estimator $\hat{g}_{n,\lambda_n}$ minimises the empirical objective, so comparing to $g_0$:
\begin{align}\label{eq:objective_inequality}
    P_n\Big[\big(Y - \hat{g}_{n,\lambda_n}(X)\big)^2\Big] + \lambda_n\hat{\cT}_n(\hat{g}_{n,\lambda_n})\le P_n\Big[\big(Y - g_0(X)\big)^2\Big] + \lambda_n\hat{\cT}_n(g_0).
\end{align}
 Dividing both sides by $\lambda_n > 0$:
\begin{align}
    \hat{\cT}_n\big(\hat{g}_{n,\lambda_n}\big) &\le \frac{1}{\lambda_n }P_n\Big[\big(Y - \hat{g}_{n,\lambda_n}(X)\big)^2\Big] + \hat{\cT}_n(\hat{g}_{n,\lambda_n}) \nonumber\\
    &\le \underbrace{\frac{P_n\Big[\big(Y - g_0(X)\big)^2\Big]}{\lambda_n}}_{\textcircled{1}} +\underbrace{\hat\cT_n(g_0)}_{\textcircled{2}}.\label{eq:penalty_division}
\end{align}
Under \Cref{ass:structural}, $\E[\varepsilon^2] = \sigma^2 < \infty$, so the weak law of large numbers gives $P_n\Big[\big(Y - g_0(X)\big)^2\Big]\overset{P}{\to} \sigma^2$. As $\lambda_n \to \infty$, 
 $\textcircled{1} =o_P(1)$.

For \textcircled{2}, $g_0 \in \cG_I$ under \Cref{ass:structural}, so $\cT(g_0) = 0$ by \Cref{prop:completeness}, and \Cref{thm:uniform_t}(a) gives $\hat{\cT}_n(g_0)\overset{P}{\to} 0$. Combining \textcircled{1} and \textcircled{2} in \eqref{eq:penalty_division} gives $\hat{\cT}_n(\hat{g}_{n,\lambda_n})\overset{P}{\to} 0$.

It remains to pass from the empirical penalty to the population one. Adding and subtracting,
\begin{align*}
    0 \le \cT(\hat{g}_{n,\lambda_n}) \le \hat{\cT}_n(\hat{g}_{n,\lambda_n}) + \sup_{g\in\cG}\big|\cT(g)-\hat{\cT}_n(g)\big|.
\end{align*}
The first term on the right converges to $0$ in probability by the preceding paragraph, and the second by \Cref{thm:uniform_t}(a). The upper bound is therefore $o_P(1)$ while the lower bound is $0$, so $\cT(\hat{g}_{n,\lambda_n}) \overset{P}{\to} 0$.

\emph{Step 2: consistency.} Recall the risk $\cL(g) = \E\big[\big(Y - g(X)\big)^2\big]$ and the functional $M(g) = \cT(g) + \big(\cL(g) - \cL(g_0)\big)_+$ from \Cref{sec:asymptotics}.  By \eqref{eq:objective_inequality} and dropping the non-negative term $\lambda_n\hat{\cT}_n(\hat{g}_{n,\lambda_n})$,
\begin{align}\label{eq:convg_emp}
P_n\Big[\big(Y - \hat{g}_{n,\lambda_n}(X)\big)^2\Big]
\le
P_n\Big[\big(Y - g_0(X)\big)^2\Big] + \lambda_n \hat{\cT}_n(g_0).
\end{align}
Let $\Delta_n \coloneqq \sup_{g \in \cG}\big|P_n[(Y-g(X))^2]-\cL(g)\big|$. By \Cref{thm:uniform_t}(a) with $\lambda=0$, $\Delta_n\overset{P}{\to}0$. Combining this with \eqref{eq:convg_emp},
\begin{align}
\label{eq:basic_inequality_derivation}
\cL(\hat{g}_{n,\lambda_n})
&\le P_n\Big[\big(Y-\hat{g}_{n,\lambda_n}(X)\big)^2\Big] + \Delta_n
\le P_n\Big[\big(Y-g_0(X)\big)^2\Big] + \lambda_n \hat{\cT}_n(g_0) + \Delta_n \nonumber\\
&\le \cL(g_0) + 2\Delta_n + \lambda_n \hat{\cT}_n(g_0).
\end{align}
The residual of $g_0$ is $\varepsilon = Y - g_0(X)$, so $\E[\varepsilon^4] \le 8\,\E[Y^4] + 8\,\E\big[G(X)^4\big] < \infty$ by \Cref{ass:regularity}(b), and hence $\Var(R_{g_0}) = \Var(\varepsilon^2) < \infty$. \Cref{cor:null_general} applied at the fixed function $g_0$ therefore gives $\hat{\cT}_n(g_0)=O_P(1/n)$, and $\lambda_n = o(n)$ yields $\lambda_n \hat{\cT}_n(g_0) = o_P(1)$. Hence
\begin{align}\label{eq:convg_L}
\cL(\hat{g}_{n,\lambda_n}) \le \cL(g_0) + o_P(1).
\end{align}
Therefore,
\begin{align*}
    \big(\cL(\hat{g}_{n,\lambda_n})-\cL(g_0)\big)_+ \overset{P}{\to} 0.
\end{align*}

Meanwhile, by Step 1,
\begin{align}\label{eq:convg_T}
\cT(\hat{g}_{n,\lambda_n}) \overset{P}{\to} 0.
\end{align}

Fix $\varepsilon > 0$ and let
\begin{align*}
K_\varepsilon \coloneqq \Big\{g \in \cG : \|g - g_0\|_{P_X} \ge \varepsilon\Big\}.
\end{align*}
Since $\cG$ is compact in $L^2(P_X)$ by \Cref{lem:compact}, whose conditions hold under \Cref{ass:regularity}(a)--(c), and $K_\varepsilon$ is closed in $\cG$, $K_\varepsilon$ is itself compact. By \Cref{lem:continuity}, whose moment condition $\E[Y^2] < \infty$ holds under \Cref{ass:structural}, $\cT$ is continuous on $\cG$ in $L^2(P_X)$. The risk $\cL$ is also continuous in $L^2(P_X)$: for any $g,g'\in\cG$,
\begin{align*}
\big|\cL(g) - \cL(g')\big|
&=
\bigg|\E\Big[\big(g'(X)-g(X)\big)\big(2Y - g(X) - g'(X)\big)\Big]\bigg|\\
&\le \|g-g'\|_{P_X} \cdot \big\|2Y - g - g'\big\|_{P}\\
&\le\|g-g'\|_{P_X}\cdot \big\|2|Y| + |g(X)| + |g'(X)|\big\|_{P} \\
&  \le 2\|g-g'\|_{P_X} \cdot \Big(\E\big[(|Y|+G(X))^2\big]\Big)^{1/2},
\end{align*}
where we use Cauchy--Schwarz and \Cref{ass:regularity}(b). Therefore $M$ is continuous on the compact $K_\varepsilon$ and attains its minimum on $K_\varepsilon$.

We claim that $M(g) > 0$ on $K_\varepsilon$. If not, $M(g) = 0$ for some $g\in K_\varepsilon$. Then both $\cT(g) = 0$ and $\cL(g) \le \cL(g_0)$. By \Cref{prop:completeness}, $\cT(g) = 0$ implies $g \in \cG_I$.
By \Cref{prop:unique}, $\cL(g)\leq\cL(g_0)$ then gives $\|g-g_0\|_{P_X}=0$, contradicting $g\in K_\varepsilon$. Hence
\begin{align*}
c_\varepsilon \coloneqq \min_{g \in K_\varepsilon} M(g) > 0.
\end{align*}

If $\|\hat{g}_{n,\lambda_n} - g_0\|_{P_X} \ge \varepsilon$, then $\hat{g}_{n,\lambda_n} \in K_\varepsilon$, so $M(\hat{g}_{n,\lambda_n}) \ge c_\varepsilon$. Therefore
\begin{align*}
\Pr\big(\|\hat{g}_{n,\lambda_n} - g_0\|_{P_X} \ge \varepsilon\big)
\le
\Pr\big(M(\hat{g}_{n,\lambda_n}) \ge c_\varepsilon\big).
\end{align*}
By \eqref{eq:convg_L} and \eqref{eq:convg_T},
\begin{align*}
M(\hat{g}_{n,\lambda_n}) = \cT(\hat{g}_{n,\lambda_n}) + \big(\cL(\hat{g}_{n,\lambda_n})-\cL(g_0)\big)_+ \overset{P}{\to} 0,
\end{align*}
so $\Pr\big(M(\hat{g}_{n,\lambda_n})\ge c_\varepsilon\big) \to 0$. As $\varepsilon>0$ is arbitrary, $\|\hat g_{n,\lambda_n}-g_0\|_{P_X}^2 \to 0$ in probability.
\end{proof}

\begin{proof}[Proof of part (b)]
    Denote the two uniform fluctuations
\begin{align*}
    \Delta_n \coloneqq \sup_{g\in\cG}\Big|P_n\big[\big(Y-g(X)\big)^2\big] - \cL(g)\Big|,\quad
    \delta_n \coloneqq \sup_{g\in\cG}\big|\hat\cT_n(g) - \cT(g)\big|.
\end{align*}
By \Cref{thm:uniform_t}(b), the criterion rate at $\lambda=0$ gives $\Delta_n=O_P\big(\sqrt{V/n}\big)$, while the penalty rate gives $\delta_n=O_P\big(\sqrt{V/n}\big)$. By \Cref{cor:null_general}, $\hat\cT_n(g_0) = O_P(1/n)$. By \eqref{eq:basic_inequality_derivation} in the proof of part (a), $\cL(\hat{g}_{n,\lambda_n}) - \cL(g_0) \le 2\Delta_n + \lambda_n\hat{\cT}_n(g_0)$, hence
\begin{align}\label{eq:rate_L}
    \big(\cL(\hat{g}_{n,\lambda_n}) - \cL(g_0)\big)_+ \le 2\Delta_n + \lambda_n\,\hat\cT_n(g_0) = O_P\big(\sqrt{\frac{V}{n}} + \frac{\lambda_n}{n}\big).
\end{align}
Returning to \eqref{eq:penalty_division} from part (a), Step 1 gives $P_n\big[\big(Y-g_0(X)\big)^2] = O_P(1)$, thus 
\begin{align*}
    \hat\cT_n(\hat{g}_{n,\lambda_n})  = O_P\Big(\frac{1}{\lambda_n}+\frac{1}{n}\Big)
\end{align*}
Together with $\delta_n = O_P\big(\sqrt{V/n}\big)$, we obtain
\begin{align}\label{eq:rate_T}
    \cT(\hat{g}_{n,\lambda_n}) \le \hat\cT_n(\hat{g}_{n,\lambda_n})  + O_P\big(\sqrt{V/n}\,\big) = O_P\Big(\frac{1}{\lambda_n} + \frac{1}{n} + \sqrt{\frac{V}{n}}\Big) = O_P\Big(\frac{1}{\lambda_n} + \sqrt{\frac{V}{n}}\Big).
\end{align}
Adding \eqref{eq:rate_L} and \eqref{eq:rate_T}, we have
\begin{align}\label{eq:rate_M}
    M(\hat{g}_{n,\lambda_n}) = O_P\Big(\frac{1}{\lambda_n} + \sqrt{\frac{V}{n}}+\frac{\lambda_n}{n}\Big)
\end{align}
\Cref{ass:curvature} states that there exist $c_0,\rho > 0$ such that
\begin{align*}
    M(g) \ge c_0\,\|g - g_0\|_{P_X}^2 \qquad \text{for every } g \in \cG \text{ with } \|g - g_0\|_{P_X} \le \rho.
\end{align*}

Define the event $E_n \coloneqq \{\|\hat{g}_{n,\lambda_n} - g_0\|_{P_X} \le \rho\}$ and set $a_n\coloneqq \lambda_n^{-1} + \sqrt{V/n} + \lambda_n/n$. By part (a), $\|\hat{g}_{n,\lambda_n} - g_0\|_{P_X} \xrightarrow{P} 0$, so as $n\to\infty$,
\begin{align}\label{eq:E_n_consistent}
    \Pr(E_n^c) \longrightarrow 0.
\end{align}
On $E_n$, \Cref{ass:curvature} applied at $g = \hat{g}_{n,\lambda_n}$ yields $c_0\|\hat{g}_{n,\lambda_n} - g_0\|_{P_X}^2 \le M(\hat{g}_{n,\lambda_n})$, hence
\begin{align}\label{eq:curvature_event}
    \mathbbm{1}_{E_n}\,\|\hat{g}_{n,\lambda_n} - g_0\|_{P_X}^2 \le c_0^{-1}\,M(\hat{g}_{n,\lambda_n}) \qquad \text{a.s.}
\end{align}

We now show $\|\hat{g}_{n,\lambda_n} - g_0\|_{P_X}^2 = O_P(a_n)$. Fix $\varepsilon > 0$. By \eqref{eq:rate_M}, $M(\hat{g}_{n,\lambda_n}) = O_P(a_n)$, so there exist $K_\varepsilon < \infty$ and $N_\varepsilon \in \mathbb{N}$ such that
\begin{align*}
    \Pr\big(M(\hat{g}_{n,\lambda_n}) > c_0\,K_\varepsilon\,a_n\big) \le \frac{\varepsilon}{2} \qquad \text{for all } n \ge N_\varepsilon.
\end{align*}
By \eqref{eq:E_n_consistent}, there exists $N'_\varepsilon \in \mathbb{N}$ such that $\Pr(E_n^c) \le \varepsilon/2$ for all $n \ge N'_\varepsilon$. For $n \ge \max(N_\varepsilon, N'_\varepsilon)$, decomposing on $E_n$ and $E_n^c$ and using \eqref{eq:curvature_event},
\begin{align*}
    \Pr\big(\|\hat{g}_{n,\lambda_n} - g_0\|_{P_X}^2 > K_\varepsilon\,a_n\big)
    &\le \Pr(E_n^c) + \Pr\big(\{c_0^{-1}M(\hat{g}_{n,\lambda_n}) > K_\varepsilon\,a_n\}\cap E_n\big) \\
    &\le \Pr(E_n^c) + \Pr\big(M(\hat{g}_{n,\lambda_n}) > c_0\,K_\varepsilon\,a_n\big)
    \le \frac{\varepsilon}{2} + \frac{\varepsilon}{2} = \varepsilon.
\end{align*}
Since $\varepsilon > 0$ was arbitrary, the definition of $O_P$ gives
\begin{align*}
    \|\hat{g}_{n,\lambda_n} - g_0\|_{P_X}^2 = O_P(a_n) = O_P\big(\frac{1}{\lambda_n} + \sqrt{\frac{V}{n}} + \frac{\lambda_n}{n}\big),
\end{align*}
which is \eqref{eq:rate_main}. 
\end{proof}

\subsection{Proof of \Cref{cor:null_general}}\label{app:null}
\begin{proof}
Recall $V_n$ from \eqref{eq:marked_process}. Under $H_0$, the regression function of $R_g$ on $F_g$ is the constant $\bar\sigma_g^2$. We may therefore apply the constant-model case of \citet{stute1997nonparametric} with covariate $F_g$, response $R_g$, model $m(\cdot,\theta)=\theta$, and estimator $\bar R_g$. Their marked empirical process then coincides with $V_n$.

The limit theory rests on two conditions, both imposed on the model being checked by \citet{stute1997nonparametric} and distinct from the assumptions of the present paper. We verify each for the constant model.

The first, Assumption~1 of that paper, asks the estimator of the model parameter to be asymptotically linear, that is, $\sqrt{n}(\theta_n - \theta_0) = n^{-1/2}\sum_{i=1}^n \varphi(X_i,Y_i) + o_P(1)$ for an influence function $\varphi$ with mean zero and finite variance. The condition is what controls the estimation error, since the residuals entering $R_n^1$ are formed at $\theta_n$ rather than at $\theta_0$, and the sampling error of $\theta_n$ therefore propagates into the process; writing it as an average lets it be absorbed into the same limit. Here $\theta_n = \bar{R}_g$ estimates $\theta_0 = \bar{\sigma}_g^2$, and the expansion holds with no remainder at all, since
\begin{align*}
\sqrt{n}\,\big(\bar{R}_g - \bar{\sigma}_g^2\big) = \frac{1}{\sqrt{n}} \sum_{i=1}^{n} \big(R_{g,i} - \bar{\sigma}_g^2\big),
\end{align*}
so the influence function is $\varphi = R_g - \bar{\sigma}_g^2$.

The second, Assumption~2 of that paper, asks $m(\cdot,\theta)$ to be smooth in $\theta$, which enters through a linearisation of $m(X,\theta_n)$ around $\theta_0$. For the constant model $m(\cdot,\theta) = \theta$ this map is the identity, and the condition is immediate.

The sample paths of $V_n$ are step functions, so $V_n$ is a random element of $D[-\infty, \infty]$, the space of right-continuous functions with left limits in which \citet{stute1997nonparametric} states the limit theory. Since $\Var(R_g) < \infty$ gives $\E[R_g^2] < \infty$, \citet[Corollary~1.3]{stute1997nonparametric} shows that $V_n$ converges in distribution to a mean-zero Gaussian process $G^0$ with covariance
\begin{align*}
K^0(s,t) &= T_g(s \wedge t) - H_g(s)\,T_g(t) - H_g(t)\,T_g(s) + H_g(s)\,H_g(t)\,T_g(\infty),\\
T_g(s) &= \int_{-\infty}^{s} \Var(R_g \,|\, F_g = u)\,\mathrm{d}H_g(u).
\end{align*}
The function $T_g$ is the cumulative conditional variance of the marks, and coincides with \citet[display~(1.4)]{stute1997nonparametric} for this pair. In particular, $T_g(\infty) = \E\big[\Var(R_g \,|\, F_g)\big] = \Var(R_g)$, since $\E[R_g \,|\, F_g]$ is constant under $H_0$.

It remains to pass from $V_n$ to the statistic. Write $\hat{H}_{n,g}$ for the empirical distribution function of $F_{g,1},\dots,F_{g,n}$, so that
\begin{align*}
\hat{T}_n(g) = \widehat{\Var}_n(R_g)^{-1} \int V_n(s)^2 \,\mathrm{d}\hat{H}_{n,g}(s),
\end{align*}
in which the integrating measure is empirical and has to be replaced by $H_g$. Since $H_g$ is continuous, the Glivenko--Cantelli theorem gives $\|\hat{H}_{n,g} - H_g\|_\infty \to 0$ almost surely. The limit $G^0$ has continuous sample paths, since $G^0(s) = W\big(T_g(s)\big) - H_g(s)\,W\big(T_g(\infty)\big)$ for a standard Brownian motion $W$, as one checks by computing the covariance of the right-hand side and comparing with $K^0$, and both $T_g$ and $H_g$ are continuous, $T_g$ because the measure it induces has a density with respect to the atomless $\mathrm{d}H_g$. A sequence converging weakly and a sequence converging in probability to a constant converge jointly, so $(V_n, \hat{H}_{n,g}) \Rightarrow (G^0, H_g)$. Weak convergence in $D[-\infty,\infty]$ to a continuous limit is uniform, and the map $(v,F) \mapsto \int v^2\,\mathrm{d}F$ is continuous at pairs of continuous arguments, so the continuous mapping theorem gives $\int V_n^2\,\mathrm{d}\hat{H}_{n,g} \Rightarrow \int (G^0)^2\,\mathrm{d}H_g$.

The weak law of large numbers, using $\E[R_g^2] < \infty$, gives $\widehat{\Var}_n(R_g) \overset{P}{\to} \Var(R_g)$. By Slutsky's lemma, $\hat{T}_n(g)$ converges in distribution to $\Var(R_g)^{-1} \int G^0(s)^2\,\mathrm{d}H_g(s)$, and convergence in distribution implies $\hat{T}_n(g) = O_P(1)$, so $\hat{\cT}_n(g) = \hat{T}_n(g)\,\widehat{\Var}_n(R_g)/n = O_P(1/n)$. If instead $\Var(R_g) = 0$, then $R_g = \bar{\sigma}_g^2$ almost surely, the process $V_n$ is identically $0$, and $\hat{\cT}_n(g) = 0$.
\end{proof}

\subsection{Proof of \Cref{cor:null_g0}}\label{app:null_g0}
\begin{proof}
Under \Cref{ass:structural}, $R_{g_0} = \varepsilon^2$ and $\varepsilon \indep g_0(X)$. The conditional variance is then constant, $\Var(R_{g_0} \,|\, F_{g_0}) = \Var(\varepsilon^2)$, so $T_{g_0} = \Var(\varepsilon^2)\, H_{g_0}$, with $T_{g_0}$ and $K^0$ as in \Cref{app:null}. Substituting into $K^0$ gives
\begin{align*}
K^0(s,t) = \Var(\varepsilon^2)\,\big(H_{g_0}(s \wedge t) - H_{g_0}(s)\,H_{g_0}(t)\big).
\end{align*}
Recall that a standard Brownian bridge $B$ is the mean-zero Gaussian process on $[0,1]$ with covariance $\Cov\big(B(u), B(w)\big) = u \wedge w - u\,w$. Because $H_{g_0}$ is nondecreasing, $H_{g_0}(s) \wedge H_{g_0}(t) = H_{g_0}(s \wedge t)$, so $K^0$ is the covariance of the process $\sqrt{\Var(\varepsilon^2)}\, B\big(H_{g_0}(\cdot)\big)$, and $G^0$ has this distribution. Since $H_{g_0}$ is continuous (\Cref{def:Gclass}), $H_{g_0}(F_{g_0}) \sim \mathrm{Uniform}(0,1)$, and the limit of $\hat{T}_n(g_0)$ in \Cref{cor:null_general} becomes
\begin{align*}
\frac{1}{\Var(\varepsilon^2)} \int_0^1 \Var(\varepsilon^2)\, B(u)^2\,\mathrm{d}u = \int_0^1 B(u)^2\,\mathrm{d}u,
\end{align*}
the classical Cram\'er--von Mises distribution \citep{anderson1952asymptotic}.
\end{proof}

\subsection{Proof of \Cref{lem:envcvm_causal_in}}\label{app:envcvm_causal_in}
\begin{proof}
Fix an environment $e$. Then
\begin{align*}
Y-g_0(X) = \varepsilon_e.
\end{align*}
Because $\varepsilon_e \indep g_0(X)$, we have
\begin{align*}
\E_{P_e}\big[(Y-g_0(X))^2 \,\big|\, g_0(X)\big]
=
\E_{P_e}\big[\varepsilon_e^2 \,\big|\, g_0(X)\big]
=
\E_{P_e}[\varepsilon_e^2],
\end{align*}
which is constant in $g_0(X)$. Therefore $g_0 \in \cG_e$. Since $e$ was arbitrary,
$g_0 \in \cG_\cap$.
\end{proof}

\subsection{Derivation for \Cref{ex:heterogeneity_linear}}\label{app:heterogeneity_linear}
Let $Z_e=\varepsilon_e+X^{e\top}(\beta_0-\beta)$ and $T_e=X^{e\top}\beta$. Joint Gaussianity and centring give
\begin{align*}
Z_e\,|\, T_e=t
\sim
\mathcal N\bigg(
\frac{\Cov_{P_e}(Z_e,T_e)}{\Var_{P_e}(T_e)}t,
\Var_{P_e}(Z_e)-\frac{\Cov_{P_e}(Z_e,T_e)^2}{\Var_{P_e}(T_e)}
\bigg).
\end{align*}
The conditional variance is constant in $t$. Because $T_e$ is nondegenerate Gaussian, $\E[Z_e^2\,|\, T_e=t]$ is constant in $t$ if and only if $\Cov_{P_e}(Z_e,T_e)=0$. Finally,
\begin{align*}
\Cov_{P_e}(Z_e,T_e)
=\Cov_{P_e}(\varepsilon_e,X^{e\top}\beta)
+\beta^\top\Sigma_e(\beta_0-\beta),
\end{align*}
which proves the equivalence in \Cref{ex:heterogeneity_linear}.

\subsection{Proof of \Cref{thm:envcvm_identification}}\label{app:envcvm_identification}
\begin{proof}
\Cref{lem:envcvm_causal_in} gives $g_0 \in \cG_e$ for every $e$, hence $g_0 \in \cG_\cap$. Conversely, take any $g \in \cG_\cap$. Then $\cT^e(g) = 0$ for every $e$. If $g \ne g_0$, \Cref{ass:heterogeneity} gives an environment $e^\ast$ with $\cT^{e^\ast}(g) > 0$, contradicting $g \in \cG_{e^\ast}$. Hence $g = g_0$, and $\cG_\cap = \{g_0\}$.

Since each $\cT^e$ is nonnegative and $w_e > 0$, $\cT^{\rm ew}(g) = 0$ if and only if $\cT^e(g) = 0$ for every $e$, that is, if and only if $g \in \cG_\cap$. Therefore $g_0$ is the unique zero of $\cT^{\rm ew}$ over $\cG$.
\end{proof}

\subsection{Proof of \Cref{prop:envcvm_uniform}}\label{app:envcvm_uniform}
\begin{proof}
Write $a_e^g \coloneqq \hat{\cT}_n^e(g) - \cT^e(g)$. By the triangle inequality and non-negativity of $w_e$,
\begin{align*}
\sup_{g\in \cG} \big|\hat{\cT}_n^{\rm ew}(g) - \cT^{\rm ew}(g)\big|
= \sup_{g\in \cG} \Big|\sum_e w_e a_e^g\Big|
\le \sum_e w_e \sup_{g\in \cG} |a_e^g|.
\end{align*}
\Cref{thm:uniform_t}(b) applied to environment $e$ gives $\sup_{g\in \cG} |a_e^g| = O_P\big(\sqrt{V/n_e}\big)$, and $V$ is a fixed constant of the class, so $\sup_{g\in \cG} |a_e^g| = O_P(n_e^{-1/2})$. Summing over the finitely many environments, the right-hand side is $\sum_e w_e \cdot O_P(n_e^{-1/2}) = O_P\big(\sum_e w_e n_e^{-1/2}\big)$ as claimed.
\end{proof}

\subsection{Proof of \Cref{prop:envcvm_consistency}}\label{app:envcvm_consistency}
\begin{proof}
The estimator $\hat{g}_{n,\lambda_n}^{\rm ew}$ minimises the empirical objective, so comparing to $g_0$,
\begin{align}\label{eq:ew_basic}
\cL_n^{\rm ew}(\hat{g}_{n,\lambda_n}^{\rm ew}) + \lambda_n \hat{\cT}_n^{\rm ew}(\hat{g}_{n,\lambda_n}^{\rm ew})
\le
\cL_n^{\rm ew}(g_0) + \lambda_n \hat{\cT}_n^{\rm ew}(g_0).
\end{align}
\emph{Step 1: the penalty at the truth is negligible.} By \Cref{lem:envcvm_causal_in}, $g_0 \in \cG_e$ for every $e$, so $\cT^e(g_0) = 0$ for every $e$ and hence $\cT^{\rm ew}(g_0) = 0$. Applied within environment $e$, where the sample is i.i.d.\ from $P_e$, \Cref{cor:null_general} gives $\hat{\cT}_n^e(g_0) = O_P(1/n_e)$. Averaging,
\begin{align*}
\hat{\cT}_n^{\rm ew}(g_0) = \sum_{e=1}^{K} w_e\,O_P(1/n_e) = O_P(K/n),
\end{align*}
the last equality using $w_e = n_e/n$. Since $\lambda_n = o(\min_e n_e)$ and $\min_e n_e \le n/K$, we get $\lambda_n = o(n/K)$ and therefore
\begin{align}\label{eq:ew_truth_negligible}
\lambda_n\,\hat{\cT}_n^{\rm ew}(g_0) = o_P(1).
\end{align}

\emph{Step 2: two uniform laws.} The empirical risk $\cL_n^{\rm ew}(g) = \sum_e w_e\, P_{n_e}^e\big[(Y-g(X))^2\big]$ is a weighted average of per-environment empirical risks, so
\begin{align*}
\sup_{g\in\cG}\big|\cL_n^{\rm ew}(g) - \cL^{\rm ew}(g)\big|
\le
\sum_e w_e \sup_{g\in\cG}\Big|P_{n_e}^e\big[(Y-g(X))^2\big] - \E_{P_e}\big[(Y-g(X))^2\big]\Big|
\overset{P}{\to} 0,
\end{align*}
each environment-wise uniform law holding on the i.i.d.\ sample from $P_e$ by \Cref{lem:pdim,lem:envelop,lem:PGC_VC}. The same averaging applied to \Cref{thm:uniform_t}(a) within each environment gives
\begin{align}\label{eq:ew_uniform}
\delta_n \coloneqq \sup_{g\in\cG}\big|\hat{\cT}_n^{\rm ew}(g) - \cT^{\rm ew}(g)\big| \overset{P}{\to} 0.
\end{align}

\emph{Step 3: the penalty at the estimator converges to $0$.} Rearranging \eqref{eq:ew_basic} and using $\cL_n^{\rm ew} \ge 0$,
\begin{align*}
\lambda_n\,\hat{\cT}_n^{\rm ew}(\hat{g}_{n,\lambda_n}^{\rm ew})
\le
\cL_n^{\rm ew}(g_0) - \cL_n^{\rm ew}(\hat{g}_{n,\lambda_n}^{\rm ew}) + \lambda_n\,\hat{\cT}_n^{\rm ew}(g_0)
\le
\cL_n^{\rm ew}(g_0) + \lambda_n\,\hat{\cT}_n^{\rm ew}(g_0).
\end{align*}
The first term is $O_P(1)$, since Step 2 gives $\cL_n^{\rm ew}(g_0) \to \cL^{\rm ew}(g_0) < \infty$ in probability, and the second is $o_P(1)$ by \eqref{eq:ew_truth_negligible}. Hence $\lambda_n\,\hat{\cT}_n^{\rm ew}(\hat{g}_{n,\lambda_n}^{\rm ew}) = O_P(1)$, that is,
\begin{align*}
\hat{\cT}_n^{\rm ew}(\hat{g}_{n,\lambda_n}^{\rm ew}) = O_P(1/\lambda_n) = o_P(1),
\end{align*}
because $\lambda_n \to \infty$. Combining with \eqref{eq:ew_uniform},
\begin{align}\label{eq:ew_T_to_zero}
0 \le \cT^{\rm ew}(\hat{g}_{n,\lambda_n}^{\rm ew}) \le \hat{\cT}_n^{\rm ew}(\hat{g}_{n,\lambda_n}^{\rm ew}) + \delta_n \overset{P}{\to} 0 .
\end{align}

\emph{Step 4: from the penalty to the estimator.} Fix $\varepsilon > 0$ and set $K_\varepsilon \coloneqq \{g \in \cG : \|g - g_0\|_{L^2(P_X^{\mathrm{mix}})} \ge \varepsilon\}$.
The envelope and pseudo-dimension conditions under $P_{\mathrm{mix}}$ follow from the corresponding conditions in the finitely many environments. Together with the mixture closedness assumed in \Cref{prop:envcvm_consistency}, \Cref{lem:compact} makes $\cG$ compact in $L^2(P_X^{\mathrm{mix}})$. The closed subset $K_\varepsilon$ is therefore compact. Because
\begin{align*}
\|g - g'\|_{L^2(P_X^{\mathrm{mix}})}^2 = \sum_e w_e\,\|g - g'\|_{L^2(P_{X,e})}^2 ,
\end{align*}
convergence in $L^2(P_X^{\mathrm{mix}})$ implies convergence in every $L^2(P_{X,e})$, so \Cref{lem:continuity} applied within each environment makes every $\cT^e$, and therefore $\cT^{\rm ew}$, continuous on $\cG$ in $L^2(P_X^{\mathrm{mix}})$. A continuous function on a compact set attains its minimum, so $c_\varepsilon \coloneqq \min_{g \in K_\varepsilon} \cT^{\rm ew}(g)$ is attained. This minimum is positive. If $\cT^{\rm ew}(g) = 0$ for some $g \in K_\varepsilon$, then $\cT^e(g) = 0$ for every $e$ because the $\cT^e$ are nonnegative and the $w_e$ positive, so $g \in \cG_\cap$, and \Cref{thm:envcvm_identification} forces $g = g_0$, contradicting $g \in K_\varepsilon$.

Consequently $\|\hat{g}_{n,\lambda_n}^{\rm ew} - g_0\|_{L^2(P_X^{\mathrm{mix}})} \ge \varepsilon$ entails $\cT^{\rm ew}(\hat{g}_{n,\lambda_n}^{\rm ew}) \ge c_\varepsilon$, whence
\begin{align*}
\Pr\Big(\big\|\hat{g}_{n,\lambda_n}^{\rm ew} - g_0\big\|_{L^2(P_X^{\mathrm{mix}})} \ge \varepsilon\Big)
\le
\Pr\Big(\cT^{\rm ew}(\hat{g}_{n,\lambda_n}^{\rm ew}) \ge c_\varepsilon\Big)
\longrightarrow 0
\end{align*}
by \eqref{eq:ew_T_to_zero}. As $\varepsilon > 0$ was arbitrary, $\hat{g}_{n,\lambda_n}^{\rm ew} \to g_0$ in $L^2(P_X^{\mathrm{mix}})$ in probability, and by the displayed decomposition of the norm, in every $L^2(P_{X,e})$ as well.
\end{proof}

\subsection{Proof of \Cref{prop:joint_bound}}\label{app:joint_bound}
\begin{proof}
The inequality $\cL_\lambda(g;\gamma)\geq\cL_{\mathrm{DRIG}}(g;\gamma)$ follows from $\cT^{\rm ew}(g)\geq0$. For $g_\beta(x)=x^\top\beta$, put $a=(-\beta^\top,1)^\top$ and $b=(\mathrm{Id}-\mathrm B)^{-\top}a$. Under the linear DRIG model,
\begin{align*}
\sum_e\omega_e\big\{\E_{P_e}[R_{g_\beta}]-\E_{P_{e_0}}[R_{g_\beta}]\big\}
=b^\top\bigg\{\sum_e\omega_e(S_e-S_{e_0})\bigg\}b
\geq0.
\end{align*}
Therefore $\cL_{\mathrm{DRIG}}(g_\beta;\gamma)$ is nondecreasing in $\gamma$. At $\gamma=1$ it equals the pooled risk, so for $\gamma\geq1$,
\begin{align*}
\cL_\lambda(g_\beta;\gamma)
&=\cL_{\mathrm{DRIG}}(g_\beta;\gamma)+\lambda\cT^{\rm ew}(g_\beta)\\
&\geq \E_{P_{\mathrm{mix}}}[R_{g_\beta}]+\lambda\cT^{\rm ew}(g_\beta).
\end{align*}
Combining the two inequalities proves \eqref{eq:max_domination}.
\end{proof}

\subsection{Proof of \Cref{prop:drig_envcvm_limit}}\label{app:drig_envcvm_limit}

\begin{proof}
Let $\tilde g \in \cG_\cap$, a nonempty set by hypothesis; then $\cT^{\rm ew}(\tilde g) = 0$ by definition of $\cG_\cap$. Since $g_{\gamma,\lambda}$ minimises $\mathcal{L}_\lambda(\,\cdot\,;\gamma)$ over $\cG$ and $\tilde g \in \cG$,
\begin{align}\label{eq:basic_ineq_lambda}
\mathcal{L}_{\mathrm{DRIG}}(g_{\gamma,\lambda}; \gamma) + \lambda\, \cT^{\rm ew}(g_{\gamma,\lambda}) = \mathcal{L}_\lambda(g_{\gamma,\lambda}; \gamma) \le \mathcal{L}_\lambda(\tilde g; \gamma) = \mathcal{L}_{\mathrm{DRIG}}(\tilde g; \gamma).
\end{align}
Continuity of $\mathcal{L}_{\mathrm{DRIG}}(\,\cdot\,; \gamma)$ on the compact set $\cG$ yields $\underline{\cL} \coloneqq \min_{g \in \cG} \mathcal{L}_{\mathrm{DRIG}}(g; \gamma) > -\infty$. Rearranging \eqref{eq:basic_ineq_lambda} and using $\mathcal{L}_{\mathrm{DRIG}}(g_{\gamma,\lambda}; \gamma) \ge \underline{\cL}$,
\begin{align}\label{eq:Tew_to_zero}
0 \le \cT^{\rm ew}(g_{\gamma,\lambda}) \le \frac{\mathcal{L}_{\mathrm{DRIG}}(\tilde g; \gamma) - \underline{\cL}}{\lambda} \xrightarrow[\lambda \to \infty]{} 0.
\end{align}

Let $\lambda_k \to \infty$, and let $\bar g \in \cG$ be the limit in $L^2(P_X^{\mathrm{mix}})$ of some subsequence of $(g_{\gamma,\lambda_k})_k$, which exists because $\cG$ is compact. Along this subsequence we have $g_{\gamma,\lambda_k} \to \bar g$, thus the continuity of $\cT^{\rm ew}$ gives $\cT^{\rm ew}(g_{\gamma,\lambda_k}) \to \cT^{\rm ew}(\bar g)$, while \eqref{eq:Tew_to_zero} gives $\cT^{\rm ew}(g_{\gamma,\lambda_k}) \to 0$. The two limits agree, so $\cT^{\rm ew}(\bar g) = 0$. As $\cT^{\rm ew} = \sum_e w_e \cT^e$ with $w_e > 0$ and each $\cT^e \ge 0$, this forces $\cT^e(\bar g) = 0$ for every $e$, that is, $\bar g \in \cG_\cap$.

The term $\lambda_k\, \cT^{\rm ew}(g_{\gamma,\lambda_k})$ is nonnegative. Dropping it from the left of \eqref{eq:basic_ineq_lambda} gives $\mathcal{L}_{\mathrm{DRIG}}(g_{\gamma,\lambda_k}; \gamma) \le \mathcal{L}_{\mathrm{DRIG}}(\tilde g; \gamma)$. As
$k \to \infty$ along the same subsequence, continuity of $\mathcal{L}_{\mathrm{DRIG}}(\,\cdot\,; \gamma)$ gives
\begin{align*}
\mathcal{L}_{\mathrm{DRIG}}(\bar g; \gamma) \le \mathcal{L}_{\mathrm{DRIG}}(\tilde g; \gamma) \qquad \text{for every } \tilde g \in \cG_\cap.
\end{align*}
With $\bar g \in \cG_\cap$, this makes $\bar g$ a minimiser of $\mathcal{L}_{\mathrm{DRIG}}(\,\cdot\,; \gamma)$ over $\cG_\cap$, so $\bar g = g^\ast_\gamma$ by uniqueness. Every subsequence of $(g_{\gamma,\lambda_k})_k$ that converges in $L^2(P_X^{\mathrm{mix}})$ therefore has limit $g^\ast_\gamma$. Since the sequence lies in the compact set $\cG$ and has $g^\ast_\gamma$ as its only cluster point, it converges, $g_{\gamma,\lambda_k} \to g^\ast_\gamma$ in $L^2(P_X^{\mathrm{mix}})$. As $\lambda_k \to \infty$ is chosen  arbitrarily, $g_{\gamma,\lambda} \to g^\ast_\gamma$ in $L^2(P_X^{\mathrm{mix}})$ as $\lambda \to \infty$.

\end{proof}

\section{Wild bootstrap for the goodness-of-fit test}\label{app:bootstrap}
The null limit of \Cref{cor:null_general} involves the unknown conditional variance $\Var(R_g \,|\, F_g)$, and the wild bootstrap reproduces it from the data \citep{mammen1993bootstrap,stute1998bootstrap}. The procedure has three steps.

\emph{Step 1: multipliers.} Draw $\xi_1, \dots, \xi_n$ i.i.d., independent of the data, with $\E[\xi_1] = 0$, $\Var(\xi_1) = 1$, and $\E|\xi_1|^{2+\delta} < \infty$ for some $\delta > 0$, which the standard normal and Rademacher choices both satisfy.

\emph{Step 2: bootstrap statistic.} Multiply each centred residual in \eqref{eq:marked_process} by its $\xi_i$ and apply the same bridge correction,
\begin{align*}
V_n^\xi(s) = \frac{1}{\sqrt{n}}\sum_{i=1}^{n} \xi_i\,(R_{g,i} - \bar{R}_g)\,\mathbbm{1}\{F_{g,i} \le s\} - \hat{H}_n(s)\,\frac{1}{\sqrt{n}}\sum_{i=1}^{n} \xi_i\,(R_{g,i} - \bar{R}_g),
\end{align*}
where $\hat{H}_n$ is the empirical distribution function of $F_{g,1}, \dots, F_{g,n}$, and set $\hat{T}_n^\xi(g) = \big(n\,\widehat{\Var}_n(R_g)\big)^{-1}\sum_{j=1}^{n} V_n^\xi(F_{g,j})^2$.

\emph{Step 3: critical value.} Repeat Step 2 over $B_n$ independent multiplier draws and let $\hat c_\alpha$ be the empirical $(1-\alpha)$ quantile of the bootstrap statistics. For the asymptotic level statement, take $B_n\to\infty$, or interpret $\hat c_\alpha$ as the exact conditional bootstrap quantile. The test rejects when $\hat{T}_n(g) > \hat{c}_\alpha$.

\begin{proposition}[Validity of the wild bootstrap]\label{prop:bootstrap_validity}
Let $g\in\cG$ be fixed and suppose $\E[R_g^2]<\infty$ and $\Var(R_g)>0$. Under $H_0$, where $g\in\cG_I$, conditionally on the data, $\hat T_n^\xi(g)$ converges weakly in probability to the same limit as $\hat T_n(g)$ in \Cref{cor:null_general}. Hence the test using the conditional $(1-\alpha)$ quantile $\hat c_\alpha$ has asymptotic level $\alpha$ whenever the corresponding limit quantile is a continuity point. Under $H_1$, where $g\notin\cG_I$, $\hat c_\alpha=O_P(1)$ while $\hat T_n(g)\to\infty$, so the test is consistent.
\end{proposition}

For a fixed $g\in\cG_I$, the proposition justifies estimating the unknown null critical value by the conditional bootstrap quantile. At $g_0$, the known critical values in \Cref{cor:null_g0} may be used instead.

\begin{proof}
Write $\mu_g=\bar\sigma_g^2$ and, for $s\in\R$, define
\begin{align*}
f_s(r,f)
=
(r-\mu_g)\big(\mathbbm{1}\{f\leq s\}-H_g(s)\big).
\end{align*}
The class $\mathcal F_g=\{f_s:s\in\R\}$ is pointwise measurable and has the $L^2(P)$ envelope $2|R_g-\mu_g|$. Since the half-line indicators form a VC class, $\mathcal F_g$ is $P$-Donsker. The moment condition on the multipliers in Step~1 satisfies the condition of the multiplier central limit theorem \citep[Theorem~2.9.6]{vandervaart1996weak}. Therefore, conditionally on the data,
\begin{align*}
\widetilde V_n^\xi(s)
=
\frac{1}{\sqrt n}\sum_{i=1}^n
\xi_i(R_{g,i}-\mu_g)
\big(\mathbbm{1}\{F_{g,i}\leq s\}-H_g(s)\big)
\end{align*}
converges weakly in probability to the Gaussian process $G^0$ in \Cref{cor:null_general}.

The empirical centering and bridge correction have no effect on this limit. Indeed,
\begin{align*}
\sup_{s\in\R}\big|V_n^\xi(s)-\widetilde V_n^\xi(s)\big|
&\leq
\|\hat H_n-H_g\|_\infty
\left|
\frac{1}{\sqrt n}\sum_{i=1}^n\xi_i(R_{g,i}-\mu_g)
\right|\\
&\quad+
|\bar R_g-\mu_g|
\sup_{s\in\R}
\left|
\frac{1}{\sqrt n}\sum_{i=1}^n
\xi_i\big(\mathbbm{1}\{F_{g,i}\leq s\}-\hat H_n(s)\big)
\right|
=o_P(1)
\end{align*}
conditionally in probability. Here $\|\hat H_n-H_g\|_\infty=o_P(1)$ and $|\bar R_g-\mu_g|=o_P(1)$, while both multiplier terms are conditionally $O_P(1)$ by the same multiplier theorem.

Finally,
\begin{align*}
\hat T_n^\xi(g)
=
\frac{1}{\widehat{\Var}_n(R_g)}
\int V_n^\xi(s)^2\,\mathrm{d}\hat H_n(s).
\end{align*}
Since $\widehat{\Var}_n(R_g)\to\Var(R_g)>0$ and $\hat H_n\to H_g$ uniformly, the continuous mapping theorem gives the asserted conditional convergence and hence the quantile conclusion. Under $H_1$, the multiplier process remains conditionally tight, so $\hat c_\alpha=O_P(1)$. Moreover, $\hat\cT_n(g)\to\cT(g)>0$ by the Glivenko--Cantelli theorem for half-lines and \Cref{prop:completeness}; hence $\hat T_n(g)\to\infty$.
\end{proof}

\section{Experimental details}\label{sec:exp_details}
\subsection{\texorpdfstring{Nonlinear design (\Cref{sec:nonlinear})}{Nonlinear design}}\label{app:nonlinear}
For the nonlinear benchmark, we use a two-layer fully connected MLP with hidden
dimension $8$, ReLU activations, and Adam optimisation with weight decay
$5\times 10^{-3}$. For gradient-based training, we replace the indicator in
\eqref{eq:sample_cvm} by the smoothing approximation in \eqref{eq:anneal} and
anneal $\tau$ from $0.3$ to $0.01$ during training.

\subsection{Linear non-Gaussian comparison (\Cref{sec:comparison_linear})}
As in \Cref{sec:linear}, $Z \sim \cN(0, 1)$, the noise variables $\varepsilon_Y$ and $\varepsilon_{ch}$ are bimodal, $Y = 2Z + \varepsilon_Y$ and $X_1 = Z$. The anchor variable $A \in \{-1, 0, 1\}$ is drawn uniformly and independently of $(Z, \varepsilon_Y, \varepsilon_{ch})$, and enters only the child mechanism, $X_2 = \frac{1}{2}Y + \varepsilon_{ch} + \delta A$. We use $n_{\mathrm{train}} = 2000$, $n_{\mathrm{test}} = 1000$, $\delta = 0.3$; the test set is fixed across replicates, with $A = 0$ at test time.

ERM is OLS on the pooled training data. Anchor regression and DRIG use the discrete anchor variable $A$, with $\gamma \in \{1, 5, 10, 50, 100\}$; DRIG treats $A = 0$ as the observational environment. The proposed estimator is fitted at $\lambda \in \{0, 10, 10^2, 10^3, 10^4, 10^5, 10^6\}$ by exhaustive grid search on a $101 \times 101$ grid in $(\beta_1, \beta_2) \in [-1, 4]^2$, using the exact empirical penalty \eqref{eq:sample_cvm}. Results are means over $100$ replications.

\subsection{\texorpdfstring{The causal target is minimax in the linear design}{The causal target is minimax in the linear design}}\label{app:linear_minimax}
Write $L_e(\beta)$ for the risk of $g_\beta(x) = \beta_1x_1 + \beta_2x_2$ under the $e$th test mechanism of \Cref{sec:linear}, with $e = 2$ the tripled child noise and $e = 3$ the strong sign flip. Under the second, $X_2 = 2Z + \varepsilon_Y + 3\varepsilon_{ch}$, so the prediction error of $g_\beta$ is $(2-\beta_1-2\beta_2)Z + (1-\beta_2)\varepsilon_Y - 3\beta_2\varepsilon_{ch}$; under the third, $X_2 = -2Z-\varepsilon_Y+\varepsilon_{ch}$ and the prediction error is $(2-\beta_1+2\beta_2)Z + (1+\beta_2)\varepsilon_Y - \beta_2\varepsilon_{ch}$. The three noise variables are independent with mean $0$, $\Var(Z) = 1$, $\Var(\varepsilon_Y) = 4.01$ and $\Var(\varepsilon_{ch}) = 2.26$, so averaging the two risks collapses the cross terms in $\beta_2$ and leaves
\begin{align*}
\frac{L_2(\beta) + L_3(\beta)}{2} = 4.01 + (\beta_1-2)^2 + 19.31\,\beta_2^2 .
\end{align*}
Hence $\max_e L_e(\beta) \ge \big(L_2(\beta)+L_3(\beta)\big)/2 \ge 4.01$, with equality only at $\beta = (2,0)$. All three risks equal $4.01$ there, so $g_0$ attains the bound and is the unique minimax linear prediction function over the three mechanisms, though \Cref{sec:linear} records that it is optimal within none of them separately.

\subsection{Single-cell RPE1 (\Cref{sec:singlecell})}\label{app:singlecell}

We preprocess the RPE1 Perturb-seq data of \citet{replogle2022mapping} following \citet{chevalley2025causalbench}, then follow the ten-gene prediction setup of \citet{shen2026causality}, using the expression of EIF1 (Ensembl ID ENSG00000173812) as the response and the expression levels of the remaining nine genes as covariates; all gene-expression values are centred by the per-gene non-targeting mean. The nine covariate genes have Ensembl IDs \mbox{ENSG00000187514}, \mbox{ENSG00000075624}, \mbox{ENSG00000147604}, \mbox{ENSG00000110700}, \mbox{ENSG00000172757}, \mbox{ENSG00000133112}, \mbox{ENSG00000067225}, \mbox{ENSG00000108518}, and \mbox{ENSG00000125691}. The training pool is the observational sample (cells under the non-targeting guide, $11{,}485$ cells) together with $10$ interventional samples, one CRISPR knockdown for each of the $10$ genes, of sizes between $101$ and $552$ cells and $1{,}865$ cells in total, giving $n = 13{,}350$ in our processed data. The full processed file holds $162{,}733$ cells, the remainder being the cells excluded in preprocessing and those from the interventions on genes outside the $10$, which supply the test environments. Following \citet{shen2026causality}, the $50$ test environments are single-gene knockdowns of genes outside these $10$, on which observational OLS has the highest mean-squared error.

One of the $10$ training environments intervenes on the response gene $Y$ itself and therefore need not satisfy \Cref{ass:structural_multi}. We retain it so that the training data coincide with those of \citet{shen2026causality}. The RPE1 analysis is accordingly an empirical stress test, and we do not assert that every hypothesis of \Cref{thm:envcvm_identification} holds for these data.

We compare $6$ estimators, with $\cG$ the linear class. These are OLS on the pooled training data (OLS-pool), the proposed estimator at $\lambda \in \{1, 10, 10^2, 10^3, 10^4, 10^5, 10^6\}$ in two versions, one fitted on the pooled data as if there were no environment labels and one the environment-weighted CvM, DRIG and anchor regression at $\gamma \in \{1, 5, 10, 20, 50, 100\}$ with the environment indicator as the anchor variable, and the combined DRIG and CvM penalty at $\gamma \in \{5, 10, 20\}$ and $\lambda \in \{500, 1000, 2000\}$.

The configuration shown for each method in \Cref{fig:singlecell_full} is the one attaining the smallest worst-case test MSE over that method's grid, which is the same oracle-envelope rule used for \Cref{tab:singlecell_few_envs} and is not a deployable tuning procedure. For the combined estimator the number of training epochs enters that grid too, selected from a trace logged every $200$ epochs over $5000$ and giving $3600$. The remaining iteratively fitted methods are read at the end of training. Sweeping the CvM and the environment-weighted CvM fits over the same epoch schedule shows that both attain their smallest worst-case test MSE at the last logged epoch, so early stopping does not improve either fit.

For the sweep in \Cref{tab:singlecell_few_envs}, at each $K \in \{1, 2, 3, 5, 7, 10\}$ and each of $20$ random seeds we draw $K$ of the $10$ interventional environments uniformly at random, pool them with the observational sample, and subsample proportionally so that the training pool has about $5000$ rows; all methods see the same rows. The $K = 0$ row follows the same protocol with no interventional environment drawn, so the training pool is a subsample of about $5000$ observational cells, and the two estimators it admits are OLS on those cells and the CvM estimator \eqref{eq:finite_estimator} fitted on them. The $\lambda$ selected there is $10^5$ or $10^6$ in every seed, at the top of the grid, so the row records what the criterion attains and not that $\lambda$ is easy to choose without labels. Extending the grid by one decade does not change the value. At $\lambda = 10^7$ the mean worst-case test MSE is $0.740$ against $0.741$ over the grid ending at $10^6$, so the reported figure has levelled off rather than been truncated at the boundary. The hyperparameter grids are as above. At every $(K, \mathrm{seed})$ we record, for each method, the configuration over its grid that gives the smallest worst-case test MSE. The test set is as in \Cref{fig:singlecell_full} and is held fixed across $K$ and across seeds.

\end{document}